%% file: main.tex
\documentclass{article}

\usepackage{ijcai11}
\usepackage{times}

 \usepackage{amsmath,amssymb,amsthm}
 \usepackage{epsfig}
 \usepackage{color}
 \usepackage{bbm}
 \usepackage{boxedminipage}
 \usepackage{pifont}
 \usepackage{xspace}
 \usepackage{graphics}
 \usepackage{misc}
 \usepackage{url}
 \usepackage{psfrag}
 \usepackage{graphicx}

 \usepackage{subfig}

\DeclareMathSymbol{\mybigsqcup}{\mathop}{largesymbols}{"46}
\DeclareSymbolFont{stmaryrd}{U}{stmry}{m}{n}
\DeclareMathSymbol{\mybigsqcap}{\mathop}{stmaryrd}{"64}

\title{
Description Logic TBoxes: Model-Theoretic
  Characterizations and Rewritability}

\author{
Carsten Lutz$^1$ $\qquad$ Robert Piro$^2$ $\qquad$ Frank Wolter$^2$\\[2mm]
  \begin{tabular}[c]{ccc}
\textnormal{$^1$ Fachbereich Informatik} && $^2$ \textnormal{Department of Computer Science}\\
\textnormal{Universit\"at Bremen, Germany} && \textnormal{University of Liverpool, UK}
  \end{tabular}
}

\begin{document}

\maketitle

\begin{abstract}
  We characterize the expressive power of description logic (DL)
  TBoxes, both for expressive DLs such as \ALC and \ALCQIO and
  lightweight DLs such as DL-Lite and \EL.  Our characterizations are
  relative to first-order logic, based on a wide range of semantic
  notions such as bisimulation, equisimulation, disjoint union, and
  direct product. We exemplify the use of the characterizations by a
  first study of the following novel family of decision problems:
  given a TBox \Tmc formulated in a DL \Lmc,
  decide whether \Tmc can be equivalently rewritten as a TBox in the fragment
  $\Lmc'$ of~\Lmc. 
\end{abstract}

\input{intro}

Most proofs in this paper are deferred to the appendix.

\input{prelim}

\input{conceptdef}

\input{tboxdef}


\begin{small}

\end{small}


\cleardoublepage
 \newpage

\appendix

\input{appendix}

\end{document}

%% file: intro.tex
\section{Introduction}

Since the emergence of description logics (DLs) in the 1970s and
80s, research in the area has been driven by the fundamental
trade-off between expressive power and computational complexity
\cite{Baader-et-al-03b}.  Over the years, the idea of what complexity is
`acceptable' has varied tremendously, from insisting on tractability
in the 1980s gradually up to \NExpTime- or even 2\NExpTime-hard DLs
in the 2000s, soon intermixed with a revival of DLs for
which reasoning is tractable or even in {\sc AC}$_0$ (in a database
context). Nowadays, it is widely accepted that there is no universal
definition of acceptable computational complexity, but that a variety
of DLs is needed to cater for the needs of different applications. For
example, this is reflected in the recent OWL 2 standard by the W3C,
which comprises one very expressive (and 2\NExpTime-complete) DL
and three tractable `profiles' to be used in
applications where the full expressive power is not needed and
efficient reasoning is crucial. 

While DLs have greatly benefited from this development, becoming much
more varied and usable, there are also new challenges that arise: how
to choose a DL for a given application? What to do when you have an
ontology formulated in a DL \Lmc, but would prefer to use a different
DL $\Lmc'$ in your application? How do the various DLs interrelate?
The first aim of this paper is to lay ground for the study of these
and similar questions by providing exact model-theoretic
characterizations of the expressive power of TBoxes formulated in the
most important DLs, including expressive ones such as \ALC and \ALCQIO
(the core of the expressive DL formalized as OWL 2) and lightweight ones
such as \EL and DL-Lite (the cores of two of the OWL 2 profiles).  We
characterize the expressive power of DL TBoxes relative to first-order
logic (FO) as a reference point, which (indirectly) also yields a
characterization of the expressive power of a DL relative to other
DLs. The second aim of this paper is to exemplify the use of the
obtained characterizations by developing algorithms for the novel
decision problem \emph{$\Lmc_1$-to-$\Lmc_2$-TBox rewritability}: given
an $\Lmc_1$-TBox~\Tmc, decide whether there is an $\Lmc_2$-TBox that
is equivalent to~\Tmc.  Note the connection to TBox approximation,
studied e.g.\ in
\cite{DBLP:conf/aaai/RenPZ10,DBLP:conf/aimsa/BotoevaCR10,DBLP:conf/rr/TserendorjRKH08}:
when $\Lmc_1$ is computationally complex and the goal is to
approximate \Tmc in a less expressive DL $\Lmc_2$, the optimal result
is of course an equivalent $\Lmc_2$-TBox~$\Tmc'$, i.e., when \Tmc can
be rewritten into $\Lmc_2$ without any loss of information.

We prepare the study of TBox expressive power with a characterization
of the expressive power of DL \emph{concepts} in
Section~\ref{sect:concdef}. These are in the spirit of the well-known
van Benthem Theorem~\cite{GorankoOtto}, giving an exact condition for
when an FO-formula with one free variable is equivalent to a DL
concept. We use different versions of bisimulation for \ALC and its
extensions, and simulations and direct products for \EL and DL-Lite.
There is related work by de Rijke and Kurtonina \cite{DBLP:journals/ai/KurtoninaR99},
which, however, does not cover those DLs that are considered central
today. We then move on to our main topics, characterizing the
expressive power of DL TBoxes and studying TBox rewritability in
Sections~\ref{sect:tboxdef} and~\ref{LightweightDLs}. To characterize
when a TBox is equivalent to an FO sentence, we use `global' and
symmetric versions of the model-theoretic constructions in
Section~\ref{sect:concdef}, enriched with various versions of
(disjoint and non-disjoint) unions and direct products. These results
are loosely related to work by Borgida
\cite{DBLP:journals/ai/Borgida96}, who focusses on DLs with complex
role constructors, and by Baader \cite{DBLP:journals/logcom/Baader96},
who uses a more liberal definition of expressive power. We use
our characterizations to establish decidability of TBox rewritability
for the $\ALCI$-to-\ALC and $\ALC$-to-$\EL$ cases. The
algorithms are highly non-trivial and a more detailed study of TBox
rewritability has to remain as future work.



%% file: prelim.tex
\section{Preliminaries}
\label{sect:prelims}

In DLs, \emph{concepts} are defined inductively based on a set
of \emph{constructors}, starting with a set $\NC$ of \emph{concept
  names}, a set $\NR$ of \emph{role names}, and 
a set $\NI$
of \emph{individual names} (all countably infinite). The
concepts of the expressive DL $\ALCQIO$ are formed using the
constructors shown in Figure~\ref{tab:syntax-semantics}.

\begin{figure}[t]
  \begin{center}
    \small
    \leavevmode
\resizebox{\columnwidth}{!}{
    \begin{tabular}{|@{\hspace*{2pt}} l @{\hspace*{2pt}}|@{\hspace*{2pt}}c@{\hspace*{2pt}}|l@{\hspace*{2pt}}|}
     \hline Name &Syntax&Semantics\\ \hline\hline & &\\[-.85em]
      inverse role &$r^-$& 
      $(r^\Imc)^\smallsmile = \{(d,e)\mid (e,d)\in r^\Imc\}$
      \\ \hline\hline & &\\[-1em]
      nominal&$\{ a \}$& $\{ a^\Imc \}$ \\ \hline & &\\[-.85em]
      negation&$\neg C$&$\Delta^\I \setminus C^\I$\\ \hline & &\\[-.85em]
      conjunction&$C\sqcap D$&$C^\I\cap D^\I$\\ \hline & &\\[-.85em]
      disjunction&$C\sqcup D$&$C^\I\cup D^\I$\\ \hline & &\\[-.85em]
  
       at-least
       restriction
      &$\qnrgeq n r C$& 
      $\{ d \in \Delta^\Imc\mid \# (r^\I(d) \cap C^\I)  \geq n \}$\\ \hline & &\\[-.85em]
       at-most
       restriction
      &$\qnrleq n r C$& 
      $\{ d \in \Delta^\Imc\mid \#(r^\I(d) \cap C^\I)  \leq n \}$ 
      \\ \hline 
       \end{tabular} 
}
    \caption{Syntax and semantics of \ALCQIO.}
    \label{tab:syntax-semantics}
  \end{center}
\vspace*{-1ex}
\end{figure}

%
In
Figure~\ref{tab:syntax-semantics} and in general, we use 
$r^\I(d)$ to denote the set of all $r$-successors of $d$ in $\I$, 
$\#S$ for the cardinality of a set $S$, $a$ and $b$ to denote
individual names, $r$ and $s$ to denote roles (i.e., role names and
inverses thereof), $A,B$ to denote concept names, and $C,D$ to denote
(possibly compound) concepts.  As usual, we use $\top$ as abbreviation
for $A \sqcup \neg A$, $\bot$ for $\neg
\top$, $\rightarrow$ and $\leftrightarrow$ for the usual Boolean
abbreviations, $\exists r . C$ (\emph{existential restriction}) for
$\qnrgeq 1 r C$, and $\forall r .  C$ (\emph{universal restriction})
for $\qnrleq 0 r {\neg C}$.

Throughout the paper, we consider the expressive DL $\ALCQIO$, which
can be viewed as a core of the OWL 2 recommendation, and several
relevant fragments; a basic such fragment underlying the OWL 2 EL
profile of OWL 2 is the lightweight DL $\mathcal{EL}$, which allows
only for $\top$, $\bot$, conjunction, and existential restrictions.
By adding negation, one obtains the basic Boolean-closed DL
$\mathcal{ALC}$. Additional constructors are indicated by
concatenation of a corresponding letter: $\mathcal{Q}$ stands for
number restrictions, $\mathcal{I}$ for inverse roles, and
$\mathcal{O}$ for nominals. This explains the name $\ALCQIO$ and
allows us to refer to fragments such as $\ALCI$ and $\ALCQ$. From the
DL-Lite family of lightweight DLs \cite{CDLLR05,aaai07}, which
underlies the OWL 2 QL profile of OWL 2, we consider
\emph{DL-Lite$_\mn{horn}$} whose concepts are conjunctions of
\emph{basic concepts} of the form $A$, $\exists r.\top$, $\bot$, or
$\top$, where $A\in \NC$ and $r$ is a role name or its inverse. We
will also consider the DL-Lite$_\mn{core}$ variant, but defer a
detailed definition to Section~\ref{sect:tboxdef}. 
We use $\mn{DL}$ to denote the set of DLs just introduced, and
${\sf ExpDL}$ to denote the set of 
\emph{expressive DLs}, i.e., \ALC and its extensions introduced above.

The semantics of DLs is defined in terms of an
\emph{interpretation} $\I=(\Delta^\Imc,\cdot^\Imc)$, where
$\Delta^\Imc$ is a non-empty set and $\cdot^\Imc$ maps each concept name
$A\in\NC$ to a subset $A^\I$ of $\Delta^\Imc$, each role name
$r\in\NR$ to a binary relation $r^\I$ on $\Delta^\Imc$, and 
each individual name $a \in \NI$ to an $a^\Imc \in \Delta^\Imc$.
The extension of $\cdot^\Imc$ to inverse roles and arbitrary concepts
is inductively defined as shown in the third column of
Figure~\ref{tab:syntax-semantics}. 

For $\Lmc \in \mn{DL}$, an $\Lmc$-\emph{TBox} is a finite set of
\emph{concept inclusions} (CIs) $C \sqsubseteq D$, where $C$ and $D$
are $\Lmc$ concepts.  An interpretation \Imc \emph{satisfies} a CI $C
\sqsubseteq D$ if $C^\Imc \subseteq D^\Imc$ and is a \emph{model} of a
TBox \Tmc if it satisfies all inclusions in \Tmc.

Concepts and TBoxes formulated in any $\Lmc \in \mn{DL}$ can be
regarded as formulas in first-order logic (FO) with equality using
unary predicates from $\NC$, binary predicates from $\NR$, and
constants from $\NI$.  More precisely, for every concept $C$ there is
an FO-formula $C^{\sharp}(x)$
such that $\mathcal{I} \models C^{\sharp}[d]$ iff $d\in
C^{\mathcal{I}}$, for all interpretations $\mathcal{I}$ and $d\in
\Delta^{\Imc}$ \cite{Baader-et-al-03b}.  For every TBox $\Tmc$, the
FO sentence
\vspace*{-1mm}
$$
\mathcal{T}^{\sharp} = \bigwedge_{C \sqsubseteq D\in \mathcal{T}} 
\forall x.(C^{\sharp}(x) \rightarrow D^{\sharp}(x))
\vspace*{-1mm}
$$
is logically equivalent to $\mathcal{T}$. We will often not explicitly
distinguish between DL-concepts and TBoxes and their translation into
FO. For example, we write $\Tmc \equiv \vp$ for a TBox \Tmc and an
FO-sentence $\vp$ whenever $\Tmc^{\sharp}$ is equivalent to $\vp$.



%% file: conceptdef.tex
\section{Characterizing Concepts}
\label{sect:concdef}

We characterize DL-concepts relative to FO-formulas with one free
variable, mainly to provide a foundation for subsequent
characterizations on the TBox level. We use the notion of an
\emph{object} $(\Imc,d)$, which consists of an interpretation $\Imc$
and a $d\in \Delta^{\Imc}$ and, intuitively, represents an object from
the real world.  Two objects $(\Imc_1,d_{1})$ and $(\Imc_2,d_{2})$ are
\emph{$\mathcal{L}$-equivalent}, written $(\Imc_1,d_{1}) \equiv_{\Lmc}
(\Imc_2,d_{2})$, if $d_{1}\in C^{\Imc_{1}}$ $\Leftrightarrow$
$d_{2}\in C^{\Imc_{2}}$ for all $\Lmc$-concepts $C$. Our first aim is to
provide, for each $\Lmc \in \mn{DL}$, a relation $\sim_{\Lmc}$ on
objects such that $\equiv_{\Lmc} \; \supseteq \; \sim_{\Lmc}$ and the
converse holds for a large class of interpretations.  To ease
notation, we use only $d$ to denote the object $(\Imc,d)$ when \Imc is
understood. 

\begin{figure}[tb]
  \begin{center}
    \small
    \leavevmode
    \begin{tabular}{|@{\,}l@{\,}|l|}
     \hline 
[Atom] & for all $(d_{1},d_{2})\in S$:
  $d_{1} \in A^{\Imc_{1}}$ iff $d_{2}\in A^{\Imc_{2}}$ \\ \hline
      [AtomR] & if $(d_{1},d_{2})\in S$ and
  $d_{1} \in A^{\Imc_{1}}$, then $d_{2}\in A^{\Imc_{2}}$\\ \hline
[Forth] & if $(d_1,d_2)\in S$ and $d_1'\in \text{succ}_{r}^{\Imc_{1}}(d_{1})$, $r \in \NR$,
then \\
 & there is a $d_2'\in \text{succ}_{r}^{\Imc_{2}}(d_{2})$ with $(d_1',d_2')\in S$.\\ \hline
[Back] & dual of [Forth]\\ \hline
$[$QForth$]$ & if $(d_1,d_2)\in S$ and  $D_1
\subseteq \text{succ}_{r}^{\Imc_{1}}(d_{1})$ finite, $r \in \NR$, \\
 & then there is a $D_2\subseteq \text{succ}_{r}^{\Imc_{2}}(d_{2})$
such that $S$ contains \\
& a bijection between $D_1$ and $D_{2}$.             
      \\ \hline 
$[$QBack$]$ & dual of [QForth]\\ \hline
$[$FSucc$]$ & if $(d_1,d_2)\in S$, $r$ a role, and
$\text{succ}_{r}^{\Imc_{1}}(d_{1})\not=\emptyset$,\\
& then $\text{succ}_{r}^{\Imc_{2}}(d_{2})\not=\emptyset$.\\ \hline 
       \end{tabular} 
    \caption{Conditions on $S \subseteq \Delta^{\Imc_{1}}\times \Delta^{\Imc_{2}}$.}
    \label{tab:conditions}
  \end{center}
\end{figure}

We start by introducing the classical notion of a bisimulation, which
corresponds to $\equiv_{\ALC}$ in the described sense. Two objects
$(\Imc_{1},d_{1})$ and $(\Imc_{2},d_{2})$ are \emph{bisimilar}, in
symbols $(\Imc_{1},d_{1}) \sim_{\ALC} (\Imc_{2},d_{2})$, if there
exists a relation $S \subseteq \Delta^{\Imc_{1}}\times
\Delta^{\Imc}_{2}$ such that the conditions [Atom] (for $A\in \NC$),
[Forth] and [Back] from Figure~\ref{tab:conditions} hold, where ${\sf
  succ}_{r}^{\Imc}(d) = \{ d' \in \Delta^{\Imc}\mid (d,d')\in
r^{\Imc}\}$ and `dual' refers to swapping the r\^oles of
$\Imc_1,d_1,d'_1$ and $\Imc_2,d_2,d'_2$; we call such an $S$ a
\emph{bisimulation} between $(\Imc_1,d_1)$ and $(\Imc_2,d_2)$.  To
address \ALCQ, we extend this to \emph{counting bisimilarity}
(cf. \cite{JaninLenzi}), in symbols
$\sim_{\mathcal{ALCQ}}$, and defined as bisimilarity, but with [Forth]
and [Back] replaced by [QForth] and [QBack] from
Figure~\ref{tab:conditions}.
Given $\sim_{\Lmc}$, the relation $\sim_{\mathcal{L}\mathcal{O}}$ for
the extension $\Lmc\Omc$ of \Lmc with nominals is defined by
additionally requiring $S$ to satisfy [Atom] for all concepts $A=\{a\}$
with $a\in \NI$. Similarly, $\sim_{\mathcal{L}\mathcal{I}}$ for the
extension $\Lmc\Imc$ of \Lmc with inverse roles demands that in all conditions of $\sim_{\Lmc}$,
$r$ additionally ranges over inverse roles. 
\begin{example}
{\em In Figure~\ref{fig:7} (L), $d_{1}\sim_{\ALC} d_{2}$ and a bisimulation
is indicated by 
dashed arrows. In contrast, $d_{1}\not\sim_{\Lmc}d_{2}$ for
$\Lmc\in \{\mathcal{ALCQ},\mathcal{ALCO},\mathcal{ALCI}\}$. It is instructive
to construct $\Lmc$-concepts $C$ that show $d_{1}\not\equiv_{\Lmc}d_{2}$.
}
\end{example}
We have provided a relation $\sim_{\Lmc}$ for each $\Lmc\in {\sf
  ExpDL}$.  For lightweight DLs with their restricted use of negation,
it
will be useful to consider \emph{non-symmetric} relations between
objects.  A relation $S\subseteq \Delta^{\Imc_{1}}\times
\Delta^{\Imc_{2}}$ is an \emph{$\EL$-simulation} from $\Imc_{1}$ to
$\Imc_{2}$ if it satisfies [AtomR] (for $A\in \NC$) and [Forth] from
Figure~\ref{tab:conditions}.  $S$ is a \emph{DL-Lite$_{\sf
    horn}$-simulation} from $\Imc_{1}$ to $\Imc_{2}$ if it satisfies
[AtomR] (for $A\in \NC$) and [FSucc]. Let $\Lmc\in \{\EL,\text{DL-Lite}_{\sf
  horn}\}$.  Then $(\Imc_{1},d_{1})$ is \emph{$\Lmc$-simulated} by
$(\Imc_{2},d_{2})$, in symbols $d_{1} \leq_{\Lmc} d_{2}$, if there
exists an $\Lmc$-simulation $S$ with $(d_{1},d_{2})\in S$.  The
relation $\sim_{\Lmc}$ that corresponds to (the inherently symmetric)
$\equiv_\Lmc$ is $\Lmc$-equisimilarity: $d_1$ and $d_2$ are
\emph{$\Lmc$-equisimilar}, written $d_{1} \sim_{\mathcal{\Lmc}}
d_{2}$, if $d_{1} \leq_{\Lmc} d_{2}$ and $d_{2} \leq_{\Lmc} d_{1}$.

\begin{example}
{\em In Figure~\ref{fig:7} (R), $d_{1}\sim_{\EL}d_{2}$, the $\EL$-simulations
are indicated by the dashed arrows. But $d_{1}\not\sim_{\ALC}d_{2}$.
}
\end{example}

It is known from modal logic that ${\equiv}_\ALC \supseteq
{\sim}_\ALC$ \cite{GorankoOtto}, 
but that the converse holds only for certain classes of
interpretations, called Hennessy-Milner classes, such as the class of
all interpretations of finite out-degree. For our purposes, we need a
class such that (i)~${\equiv}_\Lmc \subseteq {\sim}_\Lmc$ holds in
this class, for all $\Lmc \in \mn{DL}$ and (ii) every interpretation
is elementary equivalent (indistinguishable by FO sentences) to an
interpretation in the class. These conditions are satisfied by the
class of all $\omega$-saturated interpretations, as known from
classical model theory \cite{ChangKeisler} and defined in full detail
in the long version. For the reader, it is most important that this class
satisfies the above Conditions~(i) and~(ii).
It can be seen that every finite interpretation 
and modally saturated interpretation in the sense of \cite{GorankoOtto} is
$\omega$-saturated. 
%
\begin{figure}[t]
\begin{boxedminipage}{\columnwidth}
\begin{center}
\includegraphics[height=45pt]{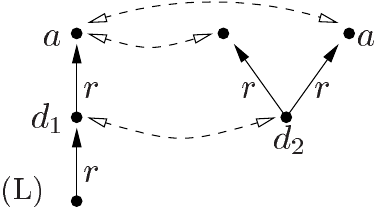}\hspace{1cm}
\includegraphics[height=45pt]{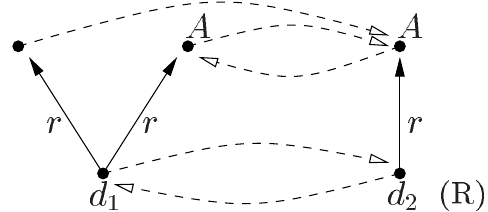}\hspace{1cm}
\end{center}
\end{boxedminipage}
\caption{Examples for $d_{1}\sim_{\Lmc}d_{2}$}
\label{fig:7}
\end{figure}

\begin{theorem}\label{equivalence}
Let $\mathcal{L}\in {\sf DL}$
and $(\mathcal{I}_{1},d_{1})$ and $(\mathcal{I}_{2},d_{2})$
be objects.
\vspace*{-4mm}
\begin{enumerate}
\setlength{\itemsep}{0pt}
\item If $d_{1} \sim_{\mathcal{L}} d_{2}$, then $d_{1} \equiv_{\mathcal{L}}d_{2}$;
\item If $d_{1}\equiv_{\mathcal{L}} d_{2}$ and $\Imc_1,\Imc_2$ are $\omega$-saturated, 
then $d_{1} \sim_{\mathcal{L}} d_{2}$.
\end{enumerate}
\end{theorem}
We now characterize concepts formulated in expressive
DLs relative to FO. An FO-formula $\varphi(x)$ is
\emph{invariant under $\sim_{\mathcal{L}}$} if for any two objects
$(\mathcal{I}_{1},d_{1})$ and $(\mathcal{I}_{2},d_{2})$, from
$\mathcal{I}_{1} \models \varphi[d_{1}]$ and $d_{1}\sim_{\mathcal{L}}
d_{2}$ it follows that $\mathcal{I}_{2} \models \varphi[d_{2}]$.
\begin{theorem}\label{local}
Let $\mathcal{L}\in {\sf ExpDL}$ and $\varphi(x)$ an FO-formula. 
Then the following conditions are equivalent:
\vspace*{-0.5mm}
\begin{enumerate}
\setlength{\itemsep}{0pt}

\item there exists an $\mathcal{L}$-concept $C$ such that
  $C \equiv \varphi(x)$;

\item $\varphi(x)$ is  invariant under $\sim_{\mathcal{L}}$.

\end{enumerate}
\end{theorem}
For $\mathcal{ALC}$, this result is exactly van Benthem's
characterization of modal formulae as the bisimulation invariant
fragment of FO \cite{GorankoOtto}.  For the modal logic variant of
$\mathcal{ALCQ}$, a similar, though more complex, characterization has
been given in \cite{DBLP:journals/sLogica/Rijke00}.

Concept definability in the lightweight DLs $\mathcal{EL}$ and
DL-Lite$_\mn{horn}$ cannot be characterized exactly as in
Theorem~\ref{equivalence}.  In fact, one can show that invariance
under $\sim_{\mathcal{EL}}$ characterizes FO-formulae equivalent to
\emph{Boolean combinations} of $\mathcal{EL}$-concepts, and invariance
under $\sim_{\text{DL-Lite}_\mn{horn}}$ characterizes FO-formulae
equivalent to DL-Lite$_{bool}$-concepts, see \cite{aaai07}.  To fix
this problem, we switch from $\sim_\Lmc$ to $\leq_\Lmc$ and
additionally require the FO-formula $\vp(x)$ to be preserved under
direct products. 
 Intuitively, the first modification addresses the
 restricted use of negation and the second one the lack of disjunction
 in \EL and $\text{DL-Lite}_\mn{horn}$. 

Let $\Imc_{i}$, $i\in I$, be a family of interpretations. 
The \emph{(direct) product} $\prod_{i\in I}\mathcal{I}_{i}$ is the interpretation defined as follows:
%
$$
\begin{array}{r@{\;}c@{\;}l}
\Delta^{\prod \Imc_i} &=&  \{ \bar d:I\rightarrow \bigcup_{i\in I} \Delta^{\Imc_i} \mid 
\text{ for }i\in I: \bar d_i =\bar d(i)\in \Delta^{\Imc_i}\} \\
A^{\prod \Imc_i} &=&\{ \bar d \in \Delta^{\prod \Imc_i} \mid \text{for } i\in I: d_i\in A^{\Imc_i}\} \quad \text{ for } 
A\in \NC\\
r^{\prod \Imc_i} &=&\{ (\bar d,\bar e) \mid \text{for } i\in I: (d_{i},e_{i}) \in r^{\Imc_i}\} \quad \text{ for } r\in \NR
\end{array}
$$
%
Note that products are closely related to Horn logic, both in the case of full
FO \cite{ChangKeisler} and modal
logic~\cite{DBLP:journals/sLogica/Sturm00a}.  An FO-formula
$\varphi(x)$ is \emph{preserved under products} if for all families
$(\Imc_i)_{i \in I}$ of interpretations and all $\bar d \in
\Delta^{\prod \Imc_i}$ with $\mathcal{I}_{i}\models \varphi[\bar
d_{i}]$ for all $i\in I$, we have $\prod_{i\in
  I}\mathcal{I}_{i},\models \varphi[\bar d]$. This notion is adapted
in the obvious way to FO sentences. For $\mathcal{L}\in
\{\mathcal{EL},\text{DL-Lite}_\mn{horn}\}$, an FO-formula $\varphi(x)$
is \emph{preserved under $\leq_{\Lmc}$} if
$(\Imc_{1},d_{1})\leq_{\Lmc} (\Imc_{2},d_{2})$ and $\Imc_{1}\models
\varphi[d_{1}]$ imply $\Imc_{2}\models \varphi[d_{2}]$.
\begin{theorem}\label{local1}
Let $\mathcal{L}\in \{\mathcal{EL},\text{DL-Lite}_\mn{horn}\}$
and $\varphi(x)$ an FO-formula. 
Then the following conditions are equivalent:
\begin{enumerate}
\setlength{\itemsep}{0pt}
\item there exists an $\mathcal{L}$-concept $C$ such that $C \equiv \varphi(x)$;
\item $\varphi(x)$ is preserved under $\leq_{\mathcal{L}}$ and under products.
\end{enumerate}
\end{theorem}
\begin{example}
{\em In Figure~\ref{fig:9}, $d_{i}\in (\exists r.A_{1}\sqcup \exists r.A_{2})^{\Imc_{i}}$ for
$i=1,2$, but $(d_{1},d_{2})\not\in (\exists r.A_{1}\sqcup \exists r.A_{2})^{{\Imc_{1}}\times \Imc_{2}}$.
Thus, disjunctions of $\EL$-concepts are not preserved under products.
}
\end{example}
\begin{figure}[t]
\begin{boxedminipage}{\columnwidth}
 \begin{center}
 \includegraphics[height=44.3pt]{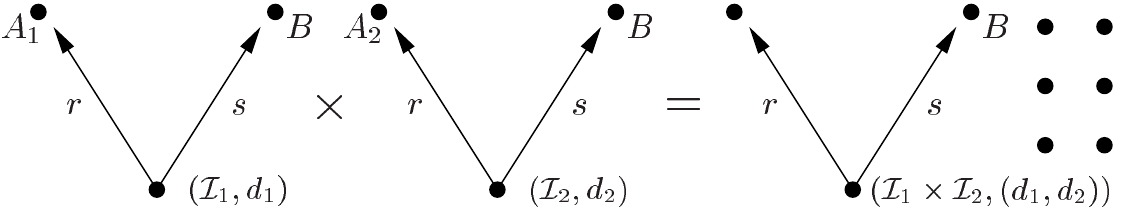}
 \end{center}
\end{boxedminipage}
\caption{A product}
\label{fig:9}
 \end{figure}
 It is known that an FO-formula is preserved under products in the
 above sense iff it is preserved under binary products (where $I$ has
 cardinality~2) \cite{ChangKeisler}.  Likewise (and because of that),
 all results stated in this paper hold both for unrestricted produces
 and for binary ones.


%% file: tboxdef.tex
\section{Characterizing TBoxes, Expressive DLs}
\label{sect:tboxdef}

A natural first idea for lifting Theorem~\ref{local} from the concept
level to the level of TBoxes is to replace the `local' relations
$\sim_\Lmc$ with their `global' counterpart $\sim_{\mathcal{L}}^{g}$, i.e.,
 $\mathcal{I}_{1} \sim_{\mathcal{L}}^{g} \mathcal{I}_{2}$ iff
for all $d_{1}\in \Delta^{\mathcal{I}_{1}}$ there exists $ d_{2}\in
\Delta^{\mathcal{I}_{2}}$ with
$(\mathcal{I}_{1},d_{1})\sim_{\mathcal{L}} (\mathcal{I}_{2},d_{2})$
and vice versa.  It turns out that, in this way, we characterize
Boolean \Lmc-TBoxes rather than \Lmc-TBoxes for all $\Lmc \in
\mn{ExpDL}$, where a \emph{Boolean \Lmc-TBox} is an expression built
up from \Lmc-concept inclusions and the Boolean operators
$\neg$,~$\wedge$,~$\vee$. 
The proof exploits compactness and Theorem~\ref{equivalence}.
\begin{theorem}\label{global}
  Let $\mathcal{L}\in {\sf ExpDL}$ and $\varphi$ an FO-sentence.  Then
  the following conditions are equivalent:
\begin{enumerate}
\setlength{\itemsep}{0pt}

\item there exists a Boolean $\mathcal{L}$-TBox $\mathcal{T}$ such
  that $\mathcal{T}\equiv \varphi$;

\item $\varphi$ is invariant under $\sim_{\mathcal{L}}^{g}$.

\end{enumerate}
\end{theorem}
%

%
%
%
To characterize TBoxes rather than Boolean TBoxes, we thus need to
strengthen the conditions on $\vp$. We first consider 
DLs without nominals. Let
$(\mathcal{I}_{i})_{i\in I}$ be a family of interpretations.
The \emph{union}
$\sum_{i\in I}\mathcal{I}_{i}$
 is defined by setting
\begin{itemize}
\setlength{\itemsep}{0pt}

\item $\Delta^{\sum_{i\in I}\mathcal{I}_{i}} = \bigcup_{i\in I} \Delta^{\mathcal{I}_{i}}$;

\item $X^{\sum_{i\in I}\mathcal{I}_{i}} =  \bigcup_{i\in I} X^{\mathcal{I}_{i}}$
  for $X \in \NC \cup \NR$.

\end{itemize}
If $\Delta^{\Imc_i} \cap \Delta^{\Imc_j} = \emptyset$ for all distinct
$i,j \in I$, then $\sum_{i\in I}\mathcal{I}_{i}$ is a \emph{disjoint
 union}. An FO-sentence $\vp$ is \emph{invariant under
 disjoint unions} if for all families $(\mathcal{I}_{i})_{i\in I}$ of
interpretations with pairwise disjoint domains, we have $\sum_{i\in
  I}\mathcal{I}_{i} \models \varphi$ iff $\mathcal{I}_{i} \models
\varphi$ for all $i\in I$.
Similar to products, one can show that an FO-sentence is invariant under disjoint unions
iff it is invariant under binary disjoint unions.
\begin{example} {\em Examples of Boolean TBoxes not invariant under
    disjoint unions are (i)~$\vp_1 = (\top \sqsubseteq A) \vee (\top
    \sqsubseteq B)$, since the disjoint union $\Imc$ of
    interpretations $\Imc_{1},\Imc_{2}$ with
    $A^{\Imc_{1}}=\Delta^{\Imc_{1}}$, $B^{\Imc_{1}}=\emptyset$, and,
    respectively, $B^{\Imc_{2}}=\Delta^{\Imc_{2}}$,
    $A^{\Imc_{2}}=\emptyset$ is not a model of $\vp_1$; and (ii)~$\vp_2=\neg
    (\top \sqsubseteq A)$, since $\Imc$ is a model of $\vp_2$, but $\Imc_{1}$ is not.  }
\end{example}
\begin{theorem}
\label{thm:withdisjunion}
  Let $\mathcal{L}\in {\sf ExpDL}$ not contain nominals and
  $\varphi$ be an FO-sentence. The following conditions
  are equivalent:
\begin{enumerate}
\setlength{\itemsep}{0pt}

\item there exists a $\mathcal{L}$-TBox $\mathcal{T}$ such that
  $\mathcal{T}\equiv \varphi$;

\item $\varphi$ is invariant under $\sim_{\mathcal{L}}^{g}$ and disjoint unions.

\end{enumerate}
\end{theorem}
\begin{proof} (sketch) The direction $1 \Rightarrow 2$ is
  straightforward based on Theorem~\ref{equivalence}, Point~1.  For the
  converse, let $\vp$ be invariant under $\sim_{\mathcal{L}}^{g}$ and
  disjoint unions and consider the set $\mn{cons}(\vp)$ of all
  \Lmc-concept inclusions $C \sqsubseteq D$ such that $\vp
  \models C \sqsubseteq D$. We are done if we can show that
  $\mn{cons}(\vp) \models \vp$: by compactness, one can find a
  finite $\Tmc \subseteq \mn{cons}(\vp)$ with $\Tmc \models \vp$, thus
  $\Tmc$ is the desired \Lmc-TBox.  Assume to the contrary that
  $\mn{cons}(\vp) \not\models \vp$. Our aim is to construct
  $\omega$-saturated interpretations $\Imc^-$ and $\Imc^+$ such that
  $\Imc^- \not\models \vp$, $\Imc^+ \models \vp$, and for all $d_{1}\in
  \Delta^{\mathcal{I}_{1}}$ there exists $ d_{2}\in
  \Delta^{\mathcal{I}_{2}}$ with
  $(\mathcal{I}_{1},d_{1})\equiv_{\mathcal{L}}
  (\mathcal{I}_{2},d_{2})$ and vice versa. By
  Theorem~\ref{equivalence}, this implies $\Imc^- \sim^g_\Lmc \Imc^+$,
  in contradiction to $\vp$ being invariant under $\sim^g_\Lmc$.  For
  each \Lmc-concept inclusion $C \sqsubseteq D \notin \mn{cons}(\vp)$,
  take a model $\Imc_{C \not\sqsubseteq D}$ of $\vp$ that refutes $C
  \sqsubseteq D$. Then $\Imc^+$ is defined as the disjoint union of all $\Imc_{C
    \not\sqsubseteq D}$ and $\Imc^-$ is defined as the disjoint union of $\Imc^+$
  with a model of $\mn{cons}(\vp) \cup \{\neg \vp \}$. It follows from
  invariance of $\vp$ under disjoint unions that
  $\Imc^{-}\not\models\vp$ and $\Imc^{+}\models \vp$. Moreover,
  $\Imc^{-}$ and $\Imc^{+}$ satisfy the same $\Lmc$-concept
  inclusions. Using the condition that $\Lmc\in {\sf ExpDL}$, one can now show that
  $\omega$-saturated interpretations
  that are elementary equivalent to $\Imc^{+}$ and $\Imc^{-}$ are as required.
\end{proof}
In a modal logic context, disjoint unions have first been used to
characterize global consequence in \cite{RijStu}. We
exploit the purely model-theoretic characterizations given in
Theorems~\ref{global} and~\ref{thm:withdisjunion} to obtain an easy,
worst-case optimal algorithm deciding whether a Boolean TBox is equivalent to a TBox.
\begin{theorem}\label{Alg:Boolean0}
  Let $\mathcal{L}\in {\sf ExpDL}$ not contain nominals.  Then it is
  \ExpTime-complete to decide whether a Boolean $\Lmc$-TBox is
  invariant under disjoint unions (equivalently, whether it is
  equivalent to an $\Lmc$-TBox).
\end{theorem}
\begin{proof} (sketch) The proof is by mutual reduction with the
  unsatisfiability problem for Boolean $\Lmc$-TBoxes, which is
  \ExpTime-complete in all cases \cite{Baader-et-al-03b}. We focus on the upper bound.  Let
  $\vp$ be a Boolean $\Lmc$-TBox. For a concept name $A$, denote by $\vp_{A}$ the
  relativization of $\varphi$ to $A$, i.e., a Boolean TBox such that
  any interpretation $\Imc$ is a model of $\vp_{A}$ iff the
  restriction of $\Imc$ to the domain $A^{\Imc}$ is a model of $\vp$. Take
  fresh concept names $A_1,A_2$ and let $\chi$ be the conjunction
  of 
\vspace*{-0.2cm}
$$
A_{1} \sqcap A_{2}\sqsubseteq \bot, \;\; \top \sqsubseteq A_1 \sqcup
A_2, \;\; A_{i} \sqsubseteq \forall r. A_i, \;\;
\neg (A_{i}\sqsubseteq \bot), 
\vspace*{-0.2cm}
$$
for all role names $r$ in $\vp$ and $i \in\{1,2\}$, expressing that
\Imc is partitioned into two disjoint and unconnected parts,
identified by $A_1$ and $A_2$.  Then $\vp$ is invariant under binary
disjoint unions iff the Boolean $\Lmc$-TBox $ \chi \rightarrow
(\vp_{A_{1}}\wedge \vp_{A_{2}} \leftrightarrow \vp) $ is a tautology.
\end{proof}
%
%
%
A further algorithmic application of Theorem~\ref{thm:withdisjunion}
and of other characterizations that we will establish later is based
on the following notion.
\begin{definition}[TBox-rewritability]
  Let $\Lmc_{1},\Lmc_{2}\in {\sf DL}$. A TBox \Tmc is
  \emph{$\Lmc_1$-rewritable} if it is equivalent to some
  $\Lmc_1$-TBox. Then \emph{$\Lmc_{1}$-to-$\Lmc_{2}$
    TBox-rewritability} is the problem to decide whether a given
  $\Lmc_{1}$-TBox is $\Lmc_2$-rewritable.
\end{definition}
If $\Lmc_{1},\Lmc_{2}\in {\sf ExpDL}$ do not contain nominals, then it
follows from Theorem~\ref{thm:withdisjunion} that an $\Lmc_{1}$-TBox
$\Tmc$ is $\Lmc_{2}$-rewritable iff $\Tmc$ it is invariant under
$\sim_{\Lmc_{2}}^{g}$. This provides a way to obtain decision
procedures for TBox-rewritability, which we explore for the first
few steps in this paper:
%
we consider $\ALCI$-to-$\ALC$ rewritability in this section, and
$\ALC$-to-$\mathcal{EL}$ and $\ALCI$-to-DL-Lite rewritability in the
subsequent one. The basis of the algorithms is that a TBox $\Tmc$ is
\emph{not} $\Lmc_{2}$-rewritable iff there are two interpretations
related by $\sim_{\Lmc_{2}}^{g}$ such that one is a model of $\Tmc$,
but the other one is not.
\begin{example} {\em A typical rewriting between $\ALCI$ and $\ALC$
    are range restrictions, which can be expressed by $\exists
    r^{-}.\top \sqsubseteq B$ in $\ALCI$ and rewritten as
    $\top\sqsubseteq \forall r.B$ in $\ALC$. Contrastingly, the
    $\ALCI$-TBox $\Tmc=\{\exists r^{-}.\top\sqcap \exists
    s^{-}.\top\sqsubseteq B\}$ is not invariant under
    $\sim_{\ALC}^{g}$: in Figure~\ref{fig:bisim}, $\Tmc$ is satisfied
    in $\Imc_{2}$, but not in $\Imc_{1}$ (where
    $B^{\Imc_{1}}=B^{\Imc_{2}}=\emptyset$). Thus, $\Tmc$ is not
    equivalent to any $\ALC$-TBox.  }
\end{example}

\begin{figure}[t]
\begin{boxedminipage}{\columnwidth}
\centering
 \includegraphics[height=60pt]{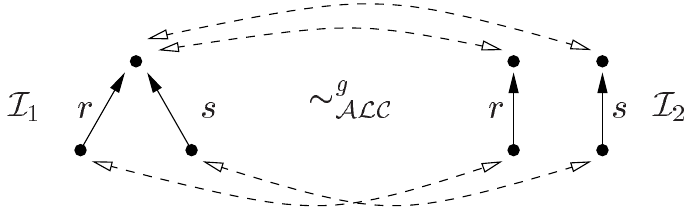}
\end{boxedminipage}
\caption{Globally bisimilar interpretations}
\label{fig:bisim}
\end{figure}

The following result is proved by a non-trivial 
refinement of the method of type
elimination known from complexity proofs in modal and description logic. 
We leave a matching lower complexity bound as an open
problem for now.
\begin{theorem}\label{ALCItoALC}
$\ALCI$-to-$\ALC$ TBox rewritability is decidable in 2-\ExpTime.
\end{theorem}
Theorem~\ref{thm:withdisjunion} excludes DLs
with nominals since it is not clear how to interpret nominals in a
disjoint union such that they are still singletons. In the following, we
devise a relaxed variant of disjoint unions that respects nominals. For simplicity, 
we only consider DLs with nominals that have inverse roles as well
(our approach can also be made to work otherwise, but becomes
more technical).

A \emph{component} of an interpretation \Imc is a set $D \subseteq
\Delta^\Imc$ that is closed under neighbors, i.e., if $d \in D$ and
$(d,d') \in \bigcup_{r \in \NR} r^\Imc \cup (r^-)^\Imc$,
then $d' \in D$.  A \emph{component interpretation} of \Imc is the
restriction \Jmc of \Imc to some domain $\Delta^\Jmc \subseteq
\Delta^\Imc$ that is a component of \Imc, i.e., $A^{\Jmc}= A^{\Imc}
\cap \Delta^\Jmc$ for all $A\in \NC$, $r^{\Jmc}= r^{\Imc} \cap
(\Delta^{\Jmc} \times \Delta^{\Jmc})$ for all $r\in \NR$, and
$a^{\Jmc}= a^{\Imc}$ for $a\in \NI$ if $a^{\Imc}\in \Delta^\Jmc$;
otherwise, $a^{\Jmc}$ is simply undefined. We denote by ${\sf
  Nom}(\Jmc)$ the set of individual names interpreted by~$\Jmc$.  
Now let $(\Jmc_{i})_{i\in I}$ be a family of component interpretations
such that 
%
%
\begin{itemize}
\item $\bigcup_{i\in I}{\sf Nom}(\Jmc_{i}) = \NI$;
\item ${\sf Nom}(\Jmc_{i}) \cap {\sf Nom}(\Jmc_{j}) = \emptyset$ for all $i\not=j$.
\end{itemize}
Then the \emph{nominal disjoint union} of $(\Jmc_{i})_{i\in I}$,
denoted $\sum^\mn{nom}_{i\in I} \Jmc_{i}$, is the interpretation
obtained by taking the disjoint union of $(\Jmc_{i})_{i\in I}$ and
then interpreting each $a \in \NI$ as $a^{\Jmc_{i}}$ for the
unique $i \in I$ with $a^{\Jmc_{i}}$ defined.

An FO-sentence $\varphi$ is \emph{invariant under nominal disjoint
unions} if the following conditions hold for all families
$(\Imc_i,\Jmc_i)_{i \in I}$ with $\Imc_i$ an interpretation and $\Jmc_i$
a component interpretation of $\Imc_i$, for all $i \in I$:
\begin{itemize}
\setlength{\itemsep}{0pt}

\item[(a)] if $\Imc_i$ is a model of $\varphi$ for all $i \in I$, then
  so is $\sum^\mn{nom}_{i\in I} \Jmc_{i}$;

\item[(b)] if $\sum^\mn{nom}_{i\in I} \Jmc_{i}$ is a model of $\varphi$ and
  $\Imc_{i_{0}} = \Jmc_{i_{0}}$ for some $i_{0}\in I$, then
  $\Imc_{i_{0}}$ is a model of $\varphi$.

\end{itemize}
Note that, in Condition~(b), $\Imc_{i_{0}} = \Jmc_{i_{0}}$ implies
that $\mn{Nom}(\Jmc_{i_0})$ is the set of all individual names, but
not necessarily that $\sum^\mn{nom}_{i\in I} \Jmc_{i} =
\Jmc_{i_0}$. We can now characterize TBoxes formulated in expressive
DLs with nominals.
\begin{theorem}
\label{thm:charnom}
Let $\Lmc\in \{\mathcal{ALCIO},\mathcal{ALCQIO}\}$ and $\varphi$ be an FO-sentence.
Then the following conditions are equivalent:
\begin{enumerate}
\setlength{\itemsep}{0pt}

\item there exists an $\Lmc$-TBox $\Tmc$ such that $\Tmc \equiv \varphi$;

\item $\varphi$ is invariant under $\sim_{\Lmc}^{g}$ and nominal disjoint unions.

\end{enumerate}
\end{theorem}
\begin{example}
{\em Condition~(a) of nominal disjoint unions can be used to
show that $\varphi=A(a) \vee A(b)$ cannot be rewritten as an
$\mathcal{ALCQIO}$-TBox. To see this, observe that $\Imc_{1}$ and $\Imc_{2}$ of
Figure~\ref{fig:4} satisfy $\varphi$ and $\sum_{i=1,2}^{\mn{nom}}\Jmc_{i}$ does not
satisfy $\varphi$. 
}
\end{example}

\begin{figure}[t]
\begin{center}
 \includegraphics[width=\columnwidth]{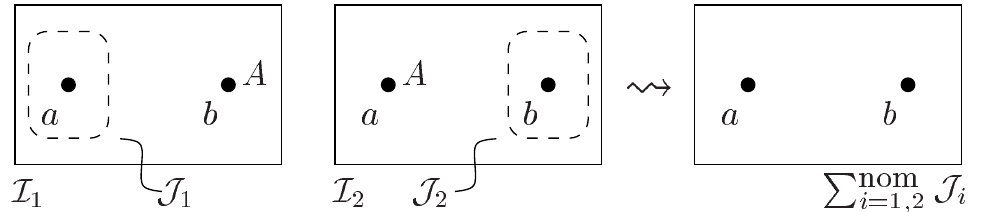}
\end{center}
\vspace*{-.5cm}
\caption{Nominal disjoint union}
\label{fig:4}
\end{figure}


Similar to the proof of Theorem~\ref{Alg:Boolean0}, one can use
relativization to reduce the problem of checking invariance under
nominal disjoint unions of Boolean $\Lmc$-TBoxes to the
unsatisfiability problem for Boolean $\Lmc$-TBoxes (which is \ExpTime-complete
for $\ALCIO$ and co\NExpTime-complete for $\mathcal{ALCQIO}$ \cite{Baader-et-al-03b}):
\begin{theorem}\label{Boolnom}
  It is \ExpTime-complete to decide whether a Boolean
  $\mathcal{ALCIO}$-TBox is invariant under nominal disjoint unions
  (equivalently, whether it is equivalent to an $\ALCIO$-TBox).  The
  problem is co\NExpTime-complete for Boolean
  $\mathcal{ALCQIO}$-TBoxes.
\end{theorem}

\section{Characterizing TBoxes, Lightweight DLs}
\label{LightweightDLs}

We characterize TBoxes formulated in $\mathcal{EL}$ and members of the DL-Lite families.
We start with an analogue of Theorem~\ref{local1}: since the
considered DLs are `Horn' in nature, we add products to the closure
properties identified in Section~\ref{sect:tboxdef} and refine our
proofs accordingly.
\begin{theorem}\label{DLLite-TBox}
Let $\Lmc\in \{\EL,\text{DL-Lite}_\mn{horn}\}$ and let $\varphi$ be an FO-sentence.
The following conditions are equivalent:
\begin{enumerate}
\setlength{\itemsep}{0pt}

\item $\varphi$ is equivalent to an $\Lmc$-TBox;
\item $\varphi$ is invariant under $\sim_{\Lmc}^{g}$ and disjoint unions, and preserved under
 products.
\end{enumerate}
\end{theorem}
\begin{proof} (sketch) In principle, we follow the strategy of the
  proof of Theorem~\ref{thm:withdisjunion}. A problem is posed by the
  fact that, unlike in the case of expressive DLs, two
  $\omega$-saturated interpretations $\Imc^-$ and $\Imc^+$ that
  satisfy the same $\Lmc$-CIs need not satisfy $\Imc^- \equiv^g_\Lmc
  \Imc^+$ (e.g.\ when $\Imc^-$ consists of three elements that satisfy
  $A\sqcap \neg B$, and $B\sqcap \neg A$, and $\neg A \sqcap \neg B$, respectively, 
  and $\Imc^+$ consists of two elements
  that satisfy $A\sqcap \neg B$ and $B \sqcap \neg A$, respectively). To deal with
  this, we ensure that $\Imc^-$ and $\Imc^+$ satisfy the same
  \emph{disjunctive \Lmc-CIs}, i.e., CIs of the form $C \sqsubseteq
  D_1 \sqcup \cdots \sqcup D_n$ with $C, D_1,\dots,D_n$ \Lmc-concepts;
  this suffices to prove $\Imc^- \equiv_g \Imc^+$ as required. The
  construction of $\Imc^-$ is essentially as in the proof of
  Theorem~\ref{thm:withdisjunion} while the construction of $\Imc^+$
  uses products to bridge the gap between \Lmc-CIs and disjunctive
  \Lmc-CIs.
\end{proof}
We apply Theorem~\ref{DLLite-TBox} to TBox rewritability, starting
with the $\ALC$-to-$\EL$ case. By Theorems~\ref{thm:withdisjunion}
and~\ref{DLLite-TBox}, an $\ALC$-TBox is equivalent to some $\EL$-TBox
iff it is invariant under $\sim_{\EL}^{g}$ and preserved under binary
products. The following theorem, the proof of which is rather involved,
establishes the complexity of both problems.
\begin{theorem}~\label{ALCdown}
Invariance of $\ALC$-TBoxes under $\sim_{\EL}^{g}$ is \ExpTime-complete.
Preservation of $\ALC$-TBoxes under products is co\NExpTime-complete.
\end{theorem}
From Theorems~\ref{ALCdown} and ~\ref{DLLite-TBox} we obtain:
\begin{theorem}
$\ALC$-to-$\EL$ TBox rewritability is in co-\NExpTime.
\end{theorem}
One can easily show \ExpTime-hardness of $\ALC$-to-$\EL$ TBox
rewritability by reduction of satisfiability of $\ALC$-TBoxes.
Namely, $\Tmc$ is satisfiable iff $\Tmc\cup \{ A\sqsubseteq \forall r.B\}$
cannot be rewritten into an $\EL$-TBox, where $A,B,r$ do not occur in~$\Tmc$.
Finding a tight bound remains open.

We now consider $\ALCI$-to-DL-Lite$_{\mn{horn}}$ TBox rewritability
and establish \ExpTime-completeness.  In contrast to $\ALC$-to-$\EL$
rewritability, where it is not clear whether or not the
computationally expensive check for preservation under products can be
avoided, here a rather direct approach is possible that relies only on
deciding invariance under $\sim_{\text{DL-Lite}_{\sf horn}}$.
\begin{theorem}
$\ALCI$-to-DL-Lite$_{\sf horn}$-TBox rewritability is \ExpTime-complete.
\end{theorem}
\begin{proof} (sketch) First decide in \ExpTime whether $\Tmc$ is
  invariant under $\sim_{\text{DL-Lite}_{\sf horn}}$. If not, then
  $\Tmc$ is not equivalent to any DL-Lite$_{\sf horn}$-TBox.  If yes,
  check, in exponential time, whether for every $B_{1}\sqcap \cdots
  \sqcap B_{n} \sqsubseteq B_{1}'\sqcup \cdots \sqcup B_{m}'$ that
  follows from $\Tmc$ with all $B_{i},B_{i}'$ basic concepts, there
  exists $j$ such that $B_{1}\sqcap \cdots \sqcap B_{n} \sqsubseteq
  B_{j}'$ follows from $\Tmc$.  $\Tmc$ is equivalent to some
  DL-Lite$_{\sf horn}$-TBox iff this is the case.
\end{proof}

The original DL-Lite dialects do not admit conjunction as a concept
constructor, or only to express disjointness constraints. More
precisely, a \emph{DL-Lite$_\mn{core}$-TBox} is a finite set of inclusions
$B_{1} \sqsubseteq B_{2}$, where $B_{1},B_{2}$ are basic DL-Lite
concepts as defined in Section~\ref{sect:prelims}.  A
\emph{DL-Lite$^{d}_\mn{core}$-TBox} admits, in addition, inclusions
$B_{1}\sqcap B_{2} \sqsubseteq \bot$ expressing disjointness of
$B_{1}$ and $B_{2}$. To characterize TBoxes formulated in
DL-Lite$_\mn{core}$ and DL-Lite$^d_\mn{core}$, we additionally require
preservation under (non-disjoint) unions and compatible unions,
respectively. The latter are unions of interpretations
$(\Imc_{i})_{i\in I}$ that can be formed only if the family
$(\Imc_{i})_{i\in I}$ is \emph{compatible}, i.e., for any $d\in
\Delta^{\Imc_{i}}\cap \Delta^{\Imc_{j}}$ and basic DL-Lite concepts
$B_{1},B_{2}$ such that $d\in B_{1}^{\Imc_{i}} \cap B_{2}^{\Imc_{j}}$
there exists $\Imc_{\ell}$ with $(B_{1}\sqcap
B_{2})^{\Imc_{\ell}}\not=\emptyset$. 
Preservation of FO-sentences under (compatible) unions is defined in
the obvious way.
The proof of the following
theorem is similar to that of Theorem~\ref{DLLite-TBox}, except that
the construction of $\Imc^+$ is yet a bit more intricate.
\begin{theorem}
\label{thm:dllitechar}
Let $\varphi$ be an FO-sentence.
Then the following conditions are equivalent:
%
\\[1mm]
\phantom{a}1.~$\varphi$ is equivalent to a DL-Lite$_\mn{core}$-TBox 
  (DL-Lite$^{d}_\mn{core}$-TBox);
\\[1mm]
\phantom{a}2.~$\varphi$ is invariant under $\sim_{\text{DL-Lite}_\mn{horn}}^{g}$ and disjoint unions, and \\
$~~~~~~~~$preserved under products and unions  
(compatible unions). 
%
\end{theorem}
%
%
Note that it is not possible to strengthen Condition~2 of
Theorem~\ref{thm:dllitechar} by requiring $\vp$ to be 
\emph{invariant} under unions as this results in failure of
the implication $1 \Rightarrow 2$. 

Because of the fact that there are only polynomially many concept
inclusions over any finite signature, TBox rewritability into
DL-Lite$_{\sf core}$ and DL-Lite$_{\sf core}^{d}$ is a comparably
simple problem and semantic characterizations are less fundamental
here than for more expressive DLs. In fact, for $\Lmc\in {\sf ExpDL}$
that contains inverse roles, one can reduce \Lmc-to-DL-Lite$_{\sf
  core}$ rewritability to Boolean \Lmc-TBox
unsatisfiability. Conversely (and trivially), \Lmc-TBox
unsatisfiability can be reduced to \Lmc-to-DL-Lite$_{\sf core}$ TBox
rewritability. As for all expressive DLs in this paper the complexity
of TBox satisfiability and Boolean TBox satisfiability coincide, this
yields tight complexity bounds.  The same holds for DL-Lite$_{\sf
  core}^{d}$. For a related study of approximation in DL-Lite, see
\cite{DBLP:conf/aimsa/BotoevaCR10}.

\section{Discussion}

We believe that the results established in this paper have many
potential applications in areas where the expressive power of TBoxes
plays a central role, such as TBox approximation and modularity. We
also believe that the problem of TBox rewritability, studied here as
an example application of our characterization results, is interesting
in its own right. A more comprehensive study, including the actual
computation of rewritten TBoxes, remains as future work.

The DLs standardized as OWL 2 and its profiles have additional
expressive power compared to the `core DLs' studied in this
paper. While full OWL 2 is probably too complex to admit really
succinct characterizations of the kind established here, some
extensions are possible as follows: each of
Theorems~\ref{thm:withdisjunion},~\ref{thm:charnom},
and~\ref{DLLite-TBox} still holds when the admissible interpretations
are restricted to some class that is definable by an FO-sentence
preserved under the notion of (disjoint) union and product used in
that theorem. This captures many features of OWL such as transitive
roles, role hierarchy axioms, and even role
inclusion axioms.


%% file: appendix.tex
\section{Proofs for Section~\ref{sect:concdef}}

%
%

To begin this section, we give a precise definition of $\omega$-saturated
interpretations. In what follows we assume that $\NC \cup \NR \cup \NI$
and the domain $\Delta^{\Imc}$ of an interpretation $\Imc$ are disjoint sets.
We can regard elements of $\Delta^{\Imc}$ as additional individual symbols
that have a fixed interpretation in $\Imc$, defined by setting $a^{\Imc}= a$ 
for all $a\in \Delta^{\Imc}$.

Let $\Imc$ be an interpretation.
A set $\Gamma$ of FO-formulas with free variables among $x_{1},\ldots,x_{n}$,
predicate symbols from $\NC \cup \NR$, and individual symbols from $\NI \cup
\Delta^{\Imc}$ is called 
\begin{itemize}
\item \emph{realizable} in $\Imc$ if there exists a variable
assignment $a(x_{i})\in \Delta^{\Imc}$, $1 \leq i \leq n$, such that 
$\Imc\models_{a}\varphi$ for all $\varphi\in \Gamma$. 
\item \emph{finitely realizable} in $\Imc$ if for every finite subset
$\Gamma'$ of $\Gamma$ there exists a variable
assignment $a(x_{i})\in \Delta^{\Imc}$, $1 \leq i \leq n$, such that 
$\Imc\models_{a}\varphi$ for all $\varphi\in \Gamma'$. 
\end{itemize}
We call an interpretation $\Imc$ \emph{$\omega$-saturated} if 
the following holds for every such set $\Gamma$ that uses only finitely
many individual symbols from $\Delta^{\Imc}$:
if $\Gamma$ is finitely realizable in $\Imc$, then $\Gamma$
is realizable in $\Imc$.

We apply the following existence theorem for $\omega$-saturated 
interpretations (cf. \cite{ChangKeisler}).
\begin{theorem}\label{exists}
For every interpretation $\Imc$ there exists an interpretation
$\Imc^{\ast}$ that is $\omega$-saturated and satisfies the same
FO-sentences as $\Imc$ (is elementary equivalent to $\Imc$).
\end{theorem}
In our proofs, we will often use the notion of a type. Formally, for
a DL \Lmc, an interpretation \Imc, and a $d \in \Delta^\Imc$, the
\emph{\Lmc-type} of $d$ in \Imc, denoted $t^\Imc_\Lmc(d)$, is the set
of \Lmc-concepts $C$ such that $d \in C^\Imc$.


%

We are in the position now to prove the results of 
Section~\ref{sect:concdef}.

\medskip

\noindent 
{\bf Theorem~\ref{equivalence}}
Let $\mathcal{L}\in \{\mathcal{EL},\text{DL-Lite}_{\sf horn}\}\cup {\sf ExpDL}$
and let $(\mathcal{I}_{1},d_{1})$ and $(\mathcal{I}_{2},d_{2})$ objects.
\begin{itemize}
\item If $d_{1} \sim_{\mathcal{L}} d_{2}$, then $d_{1} \equiv_{\mathcal{L}}d_{2}$;
\item If $d_{1}\equiv_{\mathcal{L}} d_{2}$ and both objects are $\omega$-saturated, 
then $d_{1} \sim_{\mathcal{L}} d_{2}$.
\end{itemize}

\medskip

For $\mathcal{ALC}$, various proofs of this result are known, mostly from the modal logic literature
\cite{GorankoOtto}. Many of them are easily extended so as to cover 
$\mathcal{ALCO}$, $\mathcal{ALCI}$, and $\mathcal{ALCIO}$.
Here we present proofs for $\mathcal{ALCQ}$, 
$\mathcal{EL}$, and DL-Lite$_{\sf horn}$. The extensions to the remaining
members of ${\sf ExpDL}$ ($\mathcal{ALCQI}$, $\mathcal{ALCQIO}$) are 
straightforward and left to the reader. 

\medskip

\noindent
{\bf Proof for $\mathcal{ALCQ}$}. 
Assume first that $(\Imc_{1},d_{1}) \sim_{\mathcal{ALCQ}} (\Imc_{2},d_{2})$
and let $S\subseteq \Delta^{\Imc_{1}}\times \Delta^{\Imc_{2}}$ satisfy [Atom] for all $A\in \NC$, 
[QForth], and [QBack] such that $(d_{1},d_{2})\in S$. We show 
$e_{1} \equiv_{\mathcal{ALCQ}} e_{2}$ for all $(e_{1},e_{2})\in S$; it follows that
$d_{1} \equiv_{\mathcal{ALCQ}} d_{2}$, as required. The proof is by induction over
the construction of $\mathcal{ALCQ}$-concepts. Thus, we show by induction 
for all $\mathcal{ALCQ}$-concepts $C$:

\medskip

\noindent
Claim 1. $e_{1} \in C^{\Imc_{1}}$ iff $e_{2}\in C^{\Imc_{2}}$, for all 
$(e_{1},e_{2})\in S$.

\medskip

If $C$ is a concept name, then Claim 1 follows from [Atom].
The steps for the Boolean connectives are straightforward. Now assume
$C = \qnrgeq n r D$ and let $e_{1}\in \qnrgeq n r D^{\Imc_{1}}$.
Let $X\subseteq {\sf succ}_{r}^{\Imc_{1}}(e_{1})$ be of cardinality $n$ such that
$e\in D^{\Imc_{1}}$ for all $e\in X$. By [QForth], there exists $Y\subseteq
{\sf succ}_{r}^{\Imc_{2}}(e_{2})$ such that $S$ contains a bijection between 
$X$ and $Y$.
By induction hypothesis $e'\in D^{\Imc_{2}}$ for all $e'\in Y$. Thus
$e_{2}\in \qnrgeq n r D^{\Imc_{2}}$, as required. The reverse condition
can be proved in the same way using [QBack]. The case $\qnrleq n r D$ can
be proved similarly.

\medskip
 
Conversely, assume that $(\Imc_{1},d_{1}) \equiv_{\mathcal{ALCQ}} (\Imc_{2},d_{2})$
and $\Imc_{1},\Imc_{2}$ are $\omega$-saturated. Set 
$$
S:=\{(e_1,e_2)\in \Delta^{\Imc_1}\times \Delta^{\Imc_2} \mid 
e_1\equiv_{\ALCQ}e_2\}
$$
We show that $S$ satisfies [Atom], [QForth], and [QBack].
(Then $d_{1} \sim_{\ALCQ} d_{2}$, as required.)
As [Atom] follows directly from the definition of $S$ and
[QBack] can be proved in the same way as [QForth], we focus on
[QForth]. Assume $(e_{1},e_{2})\in S$ and 
$D_{1}\subseteq {\sf succ}_{r}^{\Imc_{1}}(e_{1})$ is finite. Take an individual
variable $x_{d}$ for every $d\in D_{1}$ and consider the set of FO-formulas
$\Gamma = \Gamma^{\not=} \cup \Gamma^{r} \cup \bigcup_{d\in D_{1}}{\sf type}(d)$, where
\begin{itemize}
\item $\Gamma^{\not=} =  \{\neg (x_{d}=x_{d'}) \mid d\not=d',d,d'\in D_{1}\}$;
\item ${\sf type}(d) = \{C^{\sharp}(x_{d}) \mid C\in t^{\Imc_{1}}_{\ALCQ}(d)\}$;
\item $\Gamma^{r} = \{ r(e_{2},x_{d}) \mid d\in D_{1}\}$.
\end{itemize}
Note that $\Gamma'$, the set $\Gamma$ with $e_{2}$ replaced by $e_{1}$,
is realizable in $\Imc_{1}$ by the assignment $a(x_{d})=d$, for $d\in D_{1}$.
Using $\omega$-saturatedness of $\Imc_{2}$
and $e_{1} \equiv_{\ALCQ} e_{2}$, it is readily check that $\Gamma$ is realizable
in $\Imc_{2}$. Assume $\Gamma$ is realizable in $\Imc_{2}$ by the variable 
assignment $a(x_{d})$, $d\in D_{1}$. Let 
$$
D_{2}= \{ a(x_{d}) \mid d\in D_{1}\}.
$$
Then $d \equiv_{\ALCQ} a(x_{d})$ for all $d\in D_{1}$ (by ${\sf type}(d)$),
$D_{2}\subseteq {\sf succ}_{r}^{\Imc_{2}}(e_{2})$ (by $\Gamma^{r}$), and
$d \mapsto a(x_{d})$ is a bijection from $D_{1}$ to $D_{2}$ (by $\Gamma^{\not=}$).
Thus [QForth] holds.

This finishes the proof for $\ALCQ$.

\medskip

\noindent
{\bf Proof for $\mathcal{EL}$}.
Assume first that $(\Imc_{1},d_{1}) \sim_{\EL} (\Imc_{2},d_{2})$.
Then $(\Imc_{1},d_{1}) \leq_{\EL} (\Imc_{2},d_{2})$ and
$(\Imc_{2},d_{2}) \leq_{\EL} (\Imc_{1},d_{1})$ and so there exists an
$\EL$-simulation $S_{1}$ between $\Imc_{1}$ and $\Imc_{2}$ with $(d_{1},d_{2})\in
S_{1}$ and an $\EL$-simulation $S_{2}$ between $\Imc_{2}$ and $\Imc_{1}$ with 
$(d_{2},d_{1})\in S_{2}$. We show the following

\begin{itemize}
\item if $(e_{1},e_{2}) \in S_{1}$ and $e_{1}\in C^{\Imc_{1}}$, then 
$e_{2}\in C^{\Imc_{2}}$, for all $\mathcal{EL}$-concepts $C$;
\item if $(e_{2},e_{1}) \in S_{2}$ and $e_{2}\in C^{\Imc_{2}}$, then 
$e_{1}\in C^{\Imc_{1}}$, for all $\mathcal{EL}$-concepts $C$.
\end{itemize}
Points~1 and 2 together and $(d_{1},d_{2})\in S_{1}$, $(d_{2},d_{1})\in S_{2}$
imply $d_{1}\equiv_{\EL} d_{2}$, as required. We provide a proof of Point~1.
The proof is by induction on the construction of $C$. For concept names,
the claim follows from [AtomR]. For $\top$ and $\bot$ the claim is trivial.
For conjunction the proof is trivial. Now assume $C = \exists r.D$,
$(e_{1},e_{2})\in S_{1}$ and $e_{1}\in C^{\Imc_{1}}$. There exists $e_{1}'$
with $(e_{1},e_{1}')\in r^{\Imc_{1}}$ such that $e_{1}'\in D^{\Imc_{2}}$.
By [Forth], there exists $e_{2}'$ with $(e_{2},e_{2}')\in r^{\Imc_{2}}$ such
that $(e_{1}',e_{2}')\in S_{1}$. By induction hypothesis, $e_{2}'\in D^{\Imc_{2}}$.
Thus, $e_{2}\in C^{\Imc_{2}}$, as required.

\medskip

Conversely, let $(\Imc_{1},d_{1}) \equiv_{\EL} (\Imc_{2},d_{2})$
and assume that $\Imc_{1},\Imc_{2}$ are $\omega$-saturated.
Let
$$
S_{1}=\{(e_1,e_2)\in \Delta^{\Imc_1}\times\Delta^{\Imc_2} \mid 
t^{\Imc_{1}}_{\EL}(e_1)\subseteq t^{\Imc_{2}}_{\EL}(e_2)\}
$$
and
$$
S_{2}=\{(e_2,e_1)\in \Delta^{\Imc_2}\times\Delta^{\Imc_1} \mid 
t^{\Imc_{2}}_{\EL}(e_2)\subseteq t^{\Imc_{1}}_{\EL}(e_1)\}.
$$
We show that $S_{1}$ is a $\EL$-simulation between $\Imc_{1}$ and $\Imc_{2}$.
The same argument shows that $S_{2}$ is a $\EL$-simulation between 
$\Imc_{2}$ and $\Imc_{1}$.
Thus, from $(d_{1},d_{2})\in S_{1}$ and $(d_{2},d_{1})\in S_{2}$, we obtain
$d_{1} \sim_{\EL} d_{2}$, as required. 

Property [AtomR] follows directly from the definition of $S_{1}$.
We consider [Forth]. Let $(e_1,e_2)\in S_{1}$ and 
$(e_{1},e_{1}')\in r^{\Imc_{1}}$. 
Take an individual variable $x$ and consider the set of FO-formulas
$\Gamma = {\sf type}(e_{1}') \cup \Gamma^{r}$, where
\begin{itemize}
\item ${\sf type}(e_{1}') = \{C^{\sharp}(x) \mid C \in t^{\Imc_{1}}_{\EL}(e_{1}')\}$;
\item $\Gamma^{r} = \{ r(e_{2},x)\}$.
\end{itemize}
Note that $\Gamma'$, the set $\Gamma$ with $e_{2}$ replaced by $e_{1}$,
is realizable in $\Imc_{1}$ by the assignment $a(x)= e_{1}'$.
Using $\omega$-saturatedness of $\Imc_{2}$
and $t^{\Imc_{1}}_{\EL}(e_1)\subseteq t^{\Imc_{2}}_{\EL}(e_2)$, it is readily check that 
$\Gamma$ is realizable in $\Imc_{2}$. Assume $\Gamma$ is realizable in 
$\Imc_{2}$ by the variable assignment $a(x)$.
Then $(e_{1}',a(x)) \in S_{1}$ (by ${\sf type}(e_{1}')$)
and $(e_{2},a(x))\in r^{\Imc_{2}}$ (by $\Gamma^{r}$).
Thus [Forth] holds.

This finishes the proof for $\EL$.

\medskip

\noindent 
{\bf Proof for DL-Lite$_{\sf horn}$}.
The proof for DL-Lite$_{\sf horn}$ is rather straightforward: no induction over concepts 
is required as there are no nestings of existential restrictions. Moreover, 
$\omega$-saturatedness is not required for the implication from 
$\equiv_{\text{DL-Lite}_{\sf horn}}$ to $\sim_{\text{DL-Lite}_{\sf horn}}$.

Assume first that $(\Imc_{1},d_{1}) \sim_{\text{DL-Lite}_{\sf horn}} (\Imc_{2},d_{2})$.
Then $(\Imc_{1},d_{1}) \leq_{\text{DL-Lite}_{\sf horn}} (\Imc_{2},d_{2})$ and
$(\Imc_{2},d_{2}) \leq_{\text{DL-Lite}_{\sf horn}} (\Imc_{1},d_{1})$ and so there exists a 
DL-Lite$_{\sf horn}$-simulation $S_{1}$ between $\Imc_{1}$ and $\Imc_{2}$ with $(d_{1},d_{2})\in
S_{1}$ and a DL-Lite$_{\sf horn}$-simulation $S_{2}$ between $\Imc_{2}$ and $\Imc_{1}$ with 
$(d_{2},d_{1})\in S_{2}$. It is straightforward to show using the conditions on
DL-Lite$_{\sf horn}$-simulations that

\begin{itemize}
\item if $(e_{1},e_{2}) \in S_{1}$ and $e_{1}\in C^{\Imc_{1}}$, then 
$e_{2}\in C^{\Imc_{2}}$, for all DL-Lite$_{\sf horn}$-concepts $C$;
\item if $(e_{2},e_{1}) \in S_{2}$ and $e_{2}\in C^{\Imc_{2}}$, then 
$e_{1}\in C^{\Imc_{1}}$, for all DL-Lite$_{\sf horn}$-concepts $C$.
\end{itemize}
Points~1 and 2 together and $(d_{1},d_{2})\in S_{1}$, $(d_{2},d_{1})\in S_{2}$
imply $d_{1}\equiv_{\text{DL-Lite}_{\sf horn}} d_{2}$, as required.

\medskip

Conversely, assume $(\Imc_{1},d_{1}) \equiv_{\text{DL-Lite}_{\sf horn}} (\Imc_{2},d_{2})$.
Let
$$
S =\{(e_1,e_2)\in \Delta^{\Imc_1}\times\Delta^{\Imc_2} \mid 
t^{\Imc_{1}}_{\text{DL-Lite}}(e_1) = t^{\Imc_{2}}_{\text{DL-Lite}}(e_2)\}.
$$
It is easily checked that $S$ is a DL-Lite$_{\sf horn}$-simulation between
$\Imc_{1}$ and $\Imc_{2}$ and that $S^{-}$ is a DL-Lite$_{\sf horn}$-simulation between 
$\Imc_{2}$ and $\Imc_{1}$.
We obtain $d_{1} \sim_{\text{DL-Lite}_{\sf horn}} d_{2}$, as required. 
This finishes the proof for DL-Lite$_{\sf horn}$.

\medskip

In the proof of Theorem~\ref{local1} we will employ
the following non-symmetric version of Theorem~\ref{equivalence}
for $\EL$ and DL-Lite$_{\sf horn}$ that follows directly from
the proof of Theorem~\ref{equivalence} above:

\begin{lemma}\label{equivalence5}
Let $\mathcal{L}\in \{\mathcal{EL},\text{DL-Lite}_{\sf horn}\}$
and let $(\mathcal{I}_{1},d_{1})$ and $(\mathcal{I}_{2},d_{2})$ be objects.
\begin{itemize}
\item If $d_{1} \leq_{\mathcal{L}} d_{2}$, then $t^{\Imc_{1}}_{\Lmc}(d_{1})
\subseteq t^{\Imc_{2}}_{\Lmc}(d_{2})$;
\item If $t^{\Imc_{1}}_{\Lmc}(d_{1})
\subseteq t^{\Imc_{2}}_{\Lmc}(d_{2})$ and both objects are $\omega$-saturated , 
then $d_{1} \leq_{\mathcal{L}} d_{2}$.
\end{itemize}
\end{lemma}

\medskip

For a set $\Gamma$ of FO-formulas and a FO-formula $\varphi$ (all 
possibly containing free variables), we write $\Gamma \models \varphi$
if for every interpretation $\Imc$ with variable assigment $a$, we have 
$\Imc\models_{a} \varphi$ whenever $\Imc\models_{a}\psi$ for all 
$\psi\in \Gamma$. $\varphi\models \psi$ stands for $\{\varphi\}\models \psi$.
 
\medskip

\noindent
{\bf Theorem~\ref{local}}
Let $\mathcal{L}\in {\sf ExpDL}$ and $\varphi(x)$ a first-order formula with free variable $x$. 
Then the following conditions are equivalent:
\begin{enumerate}
\item there exists an $\mathcal{L}$-concept $C$ such that $C^{\sharp}(x)\equiv \varphi(x)$;
\item $\varphi(x)$ is invariant under $\sim_{\mathcal{L}}$.
\end{enumerate}

\medskip

\noindent
\begin{proof}
Let $\Lmc \in {\sf ExpDL}$.

The direction $1\Rightarrow 2$ follows from the fact $\Lmc$-concepts
are invariant under $\sim_{\Lmc}$ (which has been shown in 
Theorem~\ref{equivalence}).

For the direction $2 \Rightarrow 1$ let $\varphi(x)$ be invariant under
$\sim_{\Lmc}$ but assume there is no $\mathcal L$-concept $C$ such that 
$C^\sharp(x)$ is equivalent to $\varphi(x)$. 
Let 
$$
\mathrm{cons}(\varphi(x)):=\{ C^\sharp(x) \mid \text{$C$ and $\Lmc$-concept},
\varphi(x)\models C^\sharp(x))\}.
$$
By compactness, $\mathrm{cons}(\varphi(x)) \cup \{\neg \varphi(x)\}$
is satisfiable. Let $\Imc^{-}$ be an interpretation satisfying
$\mathrm{cons}(\varphi(x)) \cup \{\neg \varphi(x)\}$ under the assignment
$a_{2}(x) = d_{2}$. We may assume that $\Imc^{-}$ is $\omega$-saturated.

\medskip

\noindent
Claim 1. $\{\varphi(x)\} \cup \{C^{\sharp}(x) \mid C\in t^{\Imc}_{\Lmc}(d_{2})\}$
is satisfiable. 

\medskip

\noindent Assume that Claim 1 does not hold. Then, by compactness, there
is a finite set $\Gamma \subseteq t^{\Imc}_{\Lmc}(d_{2})$ such that 
$\{\varphi(x)\} \cup \{C^{\sharp}(x) \mid C\in \Gamma\}$
is unsatisfiable. Thus,
$$
\models \varphi(x) \rightarrow (\neg \bigsqcap_{C\in \Gamma}C)^{\sharp}(x)
$$
which implies that $(\neg \bigsqcap_{C\in \Gamma}C)^{\sharp}(x) \in 
\mathrm{cons}(\varphi(x))$ (here we use the fact that $\Lmc$-concepts
are closed under forming negations and conjunctions) and so leads to 
a contradiction as  
$\mathrm{cons}(\varphi(x)) \subseteq \{C^{\sharp}(x) \mid C \in t^{\Imc}_{\Lmc}(d_{2})\}$.

\medskip

Take an $\omega$-saturated interpretation $\Imc^{+}$ satisfying
$\{\varphi(x)\} \cup \{C^{\sharp}(x) \mid C\in t^{\Imc}_{\Lmc}(d_{2})\}$ under the assignment
$a_{1}(x) = d_{1}$. By definition, $(\Imc_{1},d_{1})\equiv_{\Lmc}(\Imc_{2},d_{2})$.
By Theorem~\ref{equivalence}, $(\Imc_{1},d_{1})\sim_{\Lmc}(\Imc_{2},d_{2})$.
We have derived a contradiction as $\Imc_{1}\models \varphi[d_{1}]$ but
$\Imc_{2}\not\models \varphi[d_{2}]$.
\end{proof}

Before proving Theorem~\ref{local1}, we determine the behaviour of $\EL$ and
DL-Lite-concepts in direct products.
\begin{lemma}\label{prodeldl}
Let $\Lmc\in \{\EL,\text{DL-Lite}_{\sf horn}\}$, $C$ a $\Lmc$-concept,
and $(\Imc_{i},d_{i})$, $i\in I$, a family of objects. Then
$$
(d_{i})_{i\in I} \in C^{\prod_{i\in I}\Imc_{i}}\quad \Leftrightarrow \quad \forall i\in I: d_{i}\in 
C^{\Imc_{i}}
$$
\end{lemma}
\begin{proof}
Straightforward.
\end{proof}

\medskip
\noindent
{\bf Theorem~\ref{local1}}
Let $\mathcal{L}\in \{\mathcal{EL},\text{DL-Lite}_{\sf horn}\}$
and $\varphi(x)$ an FO-formula with free variable $x$. 
Then the following conditions are equivalent:
\begin{enumerate}
\item there exists an $\mathcal{L}$-concept $C$ such that 
$C^{\sharp}(x)\equiv \varphi(x)$;
\item $\varphi(x)$ is preserved under ${\mathcal{L}}$-simulation and direct
products.
\end{enumerate}

\medskip

\noindent
\begin{proof} It follows from Theorem~\ref{equivalence} and 
Lemma~\ref{prodeldl} that $\EL$ and DL-Lite$_{\sf horn}$-concepts
are preserved under the corresponding simulations and under forming direct
products. The direction $1 \Rightarrow 2$ follows.

For the direction $1\Rightarrow 2$, assume that $\varphi(x)$ is preserved
under $\Lmc$-simulations and direct products but is not equivalent to 
any $\Lmc$-concept. Let 
$$
\mathrm{cons}( \varphi(x)) = 
\{ C^\sharp(x) \mid \text{$C$ an $\Lmc$-concept}, \varphi(x)\models C^\sharp(x)\}.
$$
By compactness, $\mathrm{cons}(\varphi(x)) \cup \{\neg \varphi(x)\}$ is satisfiable.
Let $\Imc^{-}$ be an $\omega$-saturated interpretation satisfying
$\mathrm{cons}(\varphi(x)) \cup \{\neg \varphi(x)\}$ under an assignment $a_{2}(x)=d_{2}$.

Let $I$ be the set of $\Lmc$-concepts $C$ with $d_{2}\notin C^{\Imc^{-}}$.
For any $C\in I$, the set $\{\varphi(x),\lnot(C^\sharp(x))\}$ is satisfiable, 
because otherwise $\varphi(x) \models C^\sharp(x)$ and hence 
$C^\sharp(x)\in \mathrm{cons}(\varphi(x))$, a contradiction to 
$\Imc^{-}\models_{a_{2}} \mathrm{cons}(\varphi(x))$. 
Let $\Imc_C$ denote an interpretation such that for some $d_C\in \Delta^{\Imc_{C}}$ 
we have $\Imc_C\models \varphi[d_C]\wedge \neg C^\sharp[d_C]$.

Define
$$
\Imc = \prod_{C\in I} \Imc_{C}, \quad \overline{d}= (d_{C})_{C\in I}
$$
As $\varphi(x)$ is preserved under products, $\Imc\models \varphi[\overline{d}]$.
As $\mathcal L$ concepts are invariant under products (Lemma~\ref{prodeldl}), 
we have $\overline{d}\not\in C^{\Imc}$, for all $C\in I$. Thus $\overline{d}\in D^{\Imc}$
implies $d_{2}\in D^{\Imc^{-}}$, for all $\Lmc$-concepts $D$. Thus, we can take
an $\omega$-saturated interpretation $\Imc^{+}$ satisfying the same FO-sentences
as $\Imc$ and a $d_{1}\in \Delta^{\Imc^{+}}$ such that $\Imc^{+}\models \varphi[d_{1}]$ and 
$d_{1}\in D^{\Imc}$ implies $d_{2}\in D^{\Imc^{-}}$, for all $\Lmc$-concepts $D$.
It follows from Lemma~\ref{equivalence5} that
$(\Imc^{+},d_{1}) \leq_{\Lmc} (\Imc^{-},d_{2})$ and we have derived a contradiction
to the condition that $\varphi(x)$ is preserved under $\Lmc$-simulations.
\end{proof}

\section{Proofs for Section~\ref{sect:tboxdef}}

\noindent
{\bf Theorem~\ref{global}.}
  Let $\mathcal{L}\in {\sf ExpDL}$ and $\varphi$ an FO-sentence.  Then
  the following conditions are equivalent:
\begin{enumerate}

\item there exists a Boolean $\mathcal{L}$-TBox $\mathcal{T}$ such
  that $\mathcal{T}\equiv \varphi$;

\item $\varphi$ is invariant under $\sim_{\mathcal{L}}^{g}$.

\end{enumerate}

\noindent
\begin{proof} 
  For the direction $1 \Rightarrow 2$, let \Tmc be Boolean \Lmc-TBox
  \Tmc such that $\mathcal{T}\equiv \varphi$ and assume w.l.o.g.\ that
  $\Tmc = \{ \top \sqsubseteq C_\Tmc \}$.  Let $\Imc_1$ and $\Imc_2$
  be interpretations such that $\Imc_1 \models \vp$ and $\Imc_1
  \sim_{\mathcal{L}}^{g} \Imc_2$. Then $\Imc_1 \models \Tmc$, thus
  $C_\Tmc^{\Imc_1} = \Delta^{\Imc_1}$. Since $\Imc_1
  \sim_{\mathcal{L}}^{g} \Imc_2$ and by Point~1 of
  Theorem~\ref{equivalence}, this yields $C_\Tmc^{\Imc_2} =
  \Delta^{\Imc_2}$, thus $\Imc_2 \models \vp$.

  For $2 \Rightarrow 1$, let $\vp$ be invariant under
  $\sim_{\mathcal{L}}^{g}$ and consider the set $\mn{cons}(\vp)$ of
  Boolean \Lmc-TBoxes that are implied by~$\vp$.  We are done if we
  can show that $\mn{cons}(\vp) \models \vp$, because by compactness
  there then is a finite $\Gamma \subseteq \mn{cons}(\vp)$ with
  $\Gamma \models \vp$, thus $\bigwedge \Gamma$ is the desired Boolean
  \Lmc-TBox. Assume to the contrary that $\mn{cons}(\vp) \not\models
  \vp$. Our aim is to construct $\omega$-saturated interpretations
  $\Imc^-$ and $\Imc^+$ such that $\Imc^- \not\models \vp$, $\Imc^+
  \models \vp$, and $\Imc^- \equiv^g_\Lmc \Imc^+$, i.e., for all
  $d_{1}\in \Delta^{\mathcal{I}_{1}}$ there exists $ d_{2}\in
  \Delta^{\mathcal{I}_{2}}$ with
  $(\mathcal{I}_{1},d_{1})\equiv_{\mathcal{L}}
  (\mathcal{I}_{2},d_{2})$ and vice versa. By
  Theorem~\ref{equivalence}, this implies $\Imc^- \sim^g_\Lmc \Imc^+$,
  in contradiction to $\vp$ being invariant under $\sim^g_\Lmc$. We
  start with $\Imc^-$, which is any model of $\mn{cons}(\vp) \cup
  \{\neg \vp\}$. Let $\Gamma$ be the set of all \Lmc-concept literals
  true in $\Imc^-$, where a \emph{concept literal} is a concept
  inclusion or the negation thereof. We have that $\Gamma \cup \{ \vp
  \}$ is satisfiable: if this is not the case, then by compactness
  there is a finite $\Gamma_f \subseteq \Gamma$ with $\Gamma_f \cup \{
  \vp \}$ unsatisfiable, thus the Boolean TBox $\neg \bigwedge
  \Gamma_f$ is in $\mn{cons}(\vp)$, in contradiction to the existence
  of $\Imc^-$. Let $\Imc^+$ be a model of $\Gamma \cup \{ \vp \}$.  By
  Theorem~\ref{exists}, we can assume w.l.o.g.\ that $\Imc^-$ and
  $\Imc^+$ are $\omega$-saturated.

  It remains to show that $\Imc^- \equiv^g_\Lmc \Imc^+$, based on the
  fact that $\Imc^-$ and $\Imc^+$ satisfy the same \Lmc-concept
  inclusions (namely those that occur positively in $\Gamma$). Take a
  $d \in \Delta^{\Imc^-}$. We have to show that there is an $e \in
  \Delta^{\Imc^+}$ with $t^{\Imc^-}_\Lmc(d) = t^{\Imc^+}_\Lmc(e)$. For
  any finite $\Gamma_f \subseteq t^{\Imc^-}_\Lmc(d)$, there is an
  $e_{\Gamma_f} \in \Delta^{\Imc^+}$ such that $e_{\Gamma_f} \in
  (\midsqcap \Gamma_f)^{\Imc^+}$: since $\Imc^-$ does not satisfy
  $\top \models \neg \midsqcap \Gamma_f$, neither does $\Imc^+$, which
  yields the desired $e_{\Gamma_f}$. As $\Imc^+$ is
  $\omega$-saturated, the existence of the $e_{\Gamma_f}$ for all
  finite $\Gamma_f \subseteq t^{\Imc^-}_\Lmc(d)$ implies the existence
  of an $e \in \Delta^{\Imc^+}$ such that $e \in C^{\Imc^+}$ for all
  $C \in \Gamma$. It follows that $t^{\Imc^-}_\Lmc(d) =
  t^{\Imc^+}_\Lmc(e)$. The direction from $\Imc^+$ to $\Imc^-$ is
  analogous.
\end{proof}
Before we come to the proof of Theorem~\ref{ALCItoALC}, we introduce
some notation that will be used in other proofs as well.

We assume that $\ALCI$-concepts are defined using conjunction, negation, and existential
restrictions. Other connectives such as disjunction and value restrictions will be used as abbreviations.
Thus, in definitions and in inductive proofs, we only consider concepts constructed using
those three constructors.

Define the \emph{role depth} $\text{rd}(C)$ of an $\ALCI$-concept $C$ in the
usual way as the number of nestings of existential restrictions in $C$. The role
depth $\text{rd}(\Tmc)$ of a TBox $\Tmc$ is the maximum of all ${\sf rd}(C)$ such
that $C$ occurs in $\Tmc$. By $\text{sub}(\Tmc)$ we denote the closure under single
negation of the set of subconcept of concepts that occur in $\Tmc$.
A \emph{$\Tmc$-type} $t$ is a subset of $\text{sub}(\Tmc)$ such that 
\begin{itemize}
\item $C\in t$ or $\neg C\in t$ for all $\neg C\in \text{sub}(\Tmc)$;
\item $C \sqcap D \in t$ iff $C\in t$ and $D\in t$, for all 
$C \sqcap D\in \text{sub}(\Tmc)$.
\end{itemize}
By $\text{tp}$ we denote the set of all $\Tmc$-types and by
$\text{tp}(\Tmc)$ the set of all $\Tmc$-types that are satisfiable in a model of $\Tmc$.
A $t\in \text{tp}$ is \emph{realized} by an object
$(\Imc,d)$ if $C\in d^{\Imc}$ for all $C\in t$. We also set
$$
t^{\Imc}(d)= \{ C\in {\sf sub}(\Tmc)\mid d\in C^{\Imc}\}
$$
For an inverse role $r$, we denote by $r^{-}$ the role name $s$ with $r=s^{-}$.
We say that two $\Tmc$-types $t_{1},t_{2}$ are \emph{coherent for a role $r$}, in symbols
$t_{1} \leadsto_{r} t_{2}$, if $\neg \exists r.C \in t'$ implies $C\not\in t$
and $\neg \exists r^{-}.C \in t'$ implies $C\not\in t$.
Note that $t\leadsto_{r} t'$ iff $t' \leadsto_{r^{-}} t$.

\medskip

\noindent
{\bf Proof of Theorem~\ref{ALCItoALC}}
$\ALCI$-to-$\ALC$ TBox rewritability is decidable in 2-\ExpTime.

\medskip

The proof extends the type elimination method known from complexity proofs in modal logic.
Let $\Tmc$ be an $\ALCI$-TBox.
The idea is to decide non-$\ALC$-rewritability of $\Tmc$ by checking whether there is an interpretation
$\Imc_{1}$ refuting $\Tmc$ and an interpretation $\Imc_{2}$ satisfying $\Tmc$ 
such that $\Imc_{1}\sim_{\ALC}^{g}\Imc_{2}$. 
In the proof, we determine the set $Z$ of all pairs $(s,S)$ with $s\in {\sf tp}$ and $S\subseteq {\sf tp}$ 
such that there exist an object $(\Imc_{1},d)$, an interpretation $\Imc_{2}$, 
and a bisimulation $B$ between $\Imc_{1}$ and $\Imc_{2}$ such that ${\sf dom}(B)=\Delta^{\Imc_{1}}$ and 
\begin{itemize}
\item $\Imc_{2}$ is a model of $\Tmc$;
\item $s=t^{\Imc_{1}}(d)$;
\item $S = \{t^{\Imc_{2}}(d') \mid (d,d')\in B\}$.
\end{itemize} 
Clearly, $\Tmc$ is not $\ALC$-rewritable iff there exists $(s,S)\in Z$ 
such that $s\in {\sf tp}\setminus{\sf tp}(\Tmc)$. 
Denote by ${\sf Init}$ the set of all pairs $(s,S)$ such that
\begin{itemize}
\item $s\in {\sf tp}$;
\item $S\subseteq {\sf tp}(\Tmc)$;
\item for all $A\in \NC$ and $t,t'\in S\cup \{s\}$: $A\in t$ iff $A\in t'$.
\end{itemize}
We have ${\sf Init} \subseteq Z$ and ${\sf Init}$ can be determined in double exponential time.
Thus, a double exponential time algorithm computing $Z$ from ${\sf Init}$ is sufficient to prove 
the desired result. To formulate the algorithm, we have to lift the coherence relation $\leadsto_{r}$ between types 
to a coherence relation between members of ${\sf Init}$. For $r\in \NR$, set
\begin{itemize}
\item $S \leadsto_{r} S'$ if for every $t\in S$ there 
exists $t'\in S'$ with $t \leadsto_{r} t'$.
\item $(s,S) \leadsto_{r} (s',S')$ if 
$s \leadsto_{r} s'$ and $S \leadsto_{r} S'$;
\end{itemize}

Denote by ${\sf Final}$ the subset of ${\sf Init}$ that is the result of applying the
rules (r1) to (r3) from Figure~\ref{rules0} exhaustively to $Y:={\sf Init}$.
Clearly, ${\sf Final}$ is obtained from ${\sf Init}$ in at most double exponentially
many steps. Thus, we are done if we can prove the following result.

\begin{figure}[t]

\begin{center}
\begin{itemize}
\item[{(r1)}] If $(s,S)\in Y$ and $\exists r.C\in s$ 
with $r$ a role name and there does not exist $(s',S')\in Y$ with 
$C\in s'$ and $(s,S) \leadsto_{r} (s',S')$, then set $Y:=Y \setminus\{(s,S)\}$.


\item[{(r2)}] If $(s,S)\in Y$ and $\exists r.C\in s$ with
$r$ an inverse role and there does not exist $(s',S')\in Y$ with $C\in s'$
and $(s',S') \leadsto_{r^{-}} (s,S)$, then set $Y:=Y \setminus\{(s,S)\}$.

\item[{(r3)}] If $(s,S)\in Y$ and $\exists r.C\in t$ for some $t \in S$ with $r$ a role name,
and there do not exist $(s',S')\in Y$ and $t'\in S'$ with $C\in t'$,
$t \leadsto_{r} t'$, and $(s,S) \leadsto_{r} (s',S')$, then set $Y:=Y \setminus\{(s,S)\}$.
\end{itemize}
    \caption{Elimination Rules}
    \label{rules0}
  \end{center}
\end{figure}

\begin{lemma}
${\sf Final} = Z$.
\end{lemma} 
\begin{proof}
We start by proving ${\sf Final} \subseteq Z$.
To this end, we construct $\Imc_{1}$, $\Imc_{2}$ and $B$ that witness $(s,S)\in Z$ for
all $(s,S)\in {\sf Final}$. We first construct $\Imc_{1}$.
Set
\begin{itemize}
\item $\Delta^{\Imc_{1}}= {\sf Final}$;
\item For $A\in \NC$: $A^{\Imc_{1}}= \{(s,S) \in \Delta^{\Imc_{1}} \mid A\in s\}$;
\item For $r\in\NR$:
$((s,S),(s',S'))\in r^{\Imc_{1}}$ iff $(s,S) \leadsto_{r} (s',S')$.
\end{itemize}
The proof of the following claim uses non-applicability of (r1) and (r2) to members of
${\sf Final}$:

\medskip

\noindent
Claim 1. For all $C\in {\sf sub}(\Tmc)$ and $(s,S)\in {\sf Final}$: $C\in s$ iff $(s,S) \in C^{\Imc_{1}}$.

\medskip

It follows that $t^{\Imc_{1}}(s,S)=s$ for all $(s,S)\in \Delta^{\Imc_{1}}$.

We now construct $\Imc_{2}$. First define $\Jmc_{2}$ by 
\begin{itemize}
\item $\Delta^{\Jmc_{2}}= \{ (s,S,t)\mid (s,S)\in {\sf Final}, t\in S\}$,
\item For $A\in \NC$: $A^{\Jmc_{2}}= \{(s,S,t) \in \Delta^{\Jmc_{2}} \mid A\in t\}$;
\item For $r\in \NR$: $((s,S,t),(s',S',t'))\in r^{\Jmc_{2}}$ iff $t \leadsto_{r} t'$
and $(s,S) \leadsto_{r} (s',S')$.
\end{itemize}
For $e=(s,S,t)\in \Delta^{\Jmc_{2}}$, take for every
$\exists r.C\in t$ with $r$ an inverse role, an object
$(\Jmc_{e,\exists r.C},e_{\exists r.C})$ such that
$\Jmc_{e,\exists r.C}$ satisfies $\Tmc$ and $e_{\exists r.C}
\in C^{\Jmc_{e,\exists r.C}}$. Assume those interpretations are disjoint
and let $\Jmc_{e}$ be defined by taking the union
of the $\Jmc_{e,\exists r.C}$ and adding $e$ to its domain
as well as $(e_{\exists r.C},e)\in r^{\Jmc_{e}}$. We may assume that
$\Delta^{\Jmc_{2}}\cap \Delta^{\Jmc_{e}}=\{e\}$ for all $e\in \Delta^{\Jmc_{2}}$.

Define $\Imc_{2}$ as the union of $\Jmc_{2}$ and 
all $\Imc_{e}$, $e\in \Delta^{\Jmc_{2}}$.
The following claim is proved using non-applicability of (r3) to ${\sf Final}$:

\medskip

\noindent
Claim 2. For all $C\in {\sf sub}(\Tmc)$ and $(s,S,t)\in \Delta^{\Imc_{2}}$:
$C\in t$ iff $(s,S,t) \in C^{\Imc_{2}}$.

\medskip

\noindent
It follows that $t^{\Imc_{2}}(s,S,t) = t$ for all $(s,S,t)\in \Delta^{\Imc_{2}}$ and, since $t\in {\sf tp}(\Tmc)$
for all such $t$, that $\Imc_{2}$ is a model of $\Tmc$.

Define $B$ as the set of all pairs $((s,S),(s,S,t))$ with $(s,S,t)\in \Delta^{\Jmc_{2}}$.

\medskip

\noindent
Claim 3. $B$ is a bisimulation.

\medskip

\noindent
To prove the claim, first
assume $(s,S)\in \Delta^{\Imc_{1}}$, $((s,S),(s',S'))\in r^{\Imc_{1}}$,
and $((s,S),(s,S,t))\in B$.
We have $(s,S) \leadsto_{r} (s',S')$. Hence $S \leadsto_{r} S'$ and so there
exists $t'\in S'$ with $t \leadsto_{r} t'$. We have $((s,S,t),(s,S',t'))\in r^{\Imc_{2}}$
and $((s',S'),(s',S',t'))\in B$, as required.

Now assume $(s,S,t)\in \Delta^{\Imc_{2}}$, $((s,S,t),(s',S',t'))\in r^{\Imc_{2}}$,
and $((s,S),(s,S,t))\in B$.
Then $((s,S),(s',S'))\in r^{\Imc_{2}}$
and $((s',S'),(s',S',t'))\in B$, as required.

\medskip

Using Claims~1 to 3 one can now use $\Imc_{1},\Imc_{2}$, and $B$ to show that 
${\sf Final}\subseteq Z$.

\medskip

We come to $Z \subseteq {\sf Final}$. Clearly, ${\sf Init}\supseteq Z$. Thus,
to prove that ${\sf Final} \supseteq Z$ it is sufficient to show that if $Y\supseteq Z$
and $Y'$ is the result of applying one of the rules (r1) to (r3) to $Y$, then $Y'\supseteq Z$.
We show this for (r2), the other rules are considered similarly.

Consider an application of (r2) that eliminates $(s,S)\in Y$ triggered by $\exists r.C\in s$.
Assume to the contrary of what has to be shown that $(s,S)\in Z$. Take interpretations $\Imc_{1}$, $\Imc_{2}$,
$d\in \Imc_{1}$ and a bisimulation $B$ between $\Imc_{1}$ and $\Imc_{2}$ 
with ${\sf dom}(B)=\Delta^{\Imc_{1}}$ that are a witness for this. 
As $s=t^{\Imc_{1}}(d)$, there exists $d'$ with $(d,d')\in r^{\Imc_{1}}$ and $C\in t^{\Imc_{1}}(d')$. 
Let $s'=t^{\Imc_{1}}(d')$ and $S' = \{t^{\Imc_{2}}(e') \mid (d',e')\in B\}$.
We have $(s',S')\in Z$, and so $(s',S')\in Y$.  
We show $(s',S') \leadsto_{r^{-}} (s,S)$ which is a contradiction to the applicability of (r2).
$s' \leadsto_{r^{-}} s$ is clear from $(d,d')\in r^{\Imc_{1}}$.  
$B$ is a bisimulation and $s:=r^{-}$ a role name. Thus, for every $(d',e')\in B$ there exists $e$ with 
$(e',e)\in s^{\Imc_{2}}$ such that $(d,e)\in B$. Thus, for every $t'\in S'$ there exists 
$t\in S$ with $t\leadsto_{s} t'$. We obtain $S' \leadsto_{r^{-}} S$, as required.
\end{proof}

\noindent
{\bf Theorem~\ref{thm:charnom}.}
Let $\Lmc\in \{\mathcal{ALCIO},\mathcal{ALCQIO}\}$ and $\varphi$ be an FO-sentence.
Then the following conditions are equivalent:
\begin{enumerate}

\item there exists an $\Lmc$-TBox $\Tmc$ such that $\Tmc \equiv \varphi$;

\item $\varphi$ is invariant under $\sim_{\Lmc}^{g}$ and nominal disjoint unions.

\end{enumerate}

\noindent
\begin{proof}
The proof of $1 \Rightarrow 2$ is straightforward and left to the reader.
Conversely, assume $\varphi$ is invariant under $\sim_{\Lmc}^{g}$ and under 
nominal disjoint unions but not equivalent to any $\Lmc$-TBox. Our proof strategy is
similar to the previous proofs. Let
$$
{\sf cons}(\varphi) = \{ C \sqsubseteq D \mid \varphi \models C \sqsubseteq D \text{ and $C,D$ are $\Lmc$-concepts}\}
$$
As in previous proofs, by compactness, ${\sf cons}(\varphi) \not\models \varphi$.
We now construct, using invariance under nominal disjoint unions,
interpretations $\Imc^{-}$ not satisfying $\varphi$ and $\Imc^{+}$
satisfying $\varphi$ such that $\Imc_{1}^{-} \equiv_{\Lmc}^{g} \Imc_{2}^{+}$. Assuming
$\omega$-saturatedness, we obtain $\Imc_{1}^{-} \sim_{\Lmc}^{g} \Imc_{2}^{+}$, and have derived
a contradiction. 
We start with the construction of $\Imc^{-}$.

For an interpretation $\Imc$ and $e,f\in \Delta^{\Imc}$, we set 
$e \sim^{R}_{\Imc} f$ iff there exists a (possibly empty) sequence 
$r_{1},\ldots,r_{n}$ of roles and $d_{0},\ldots,d_{n}$ such that 
$d_{0}=e$, $d_{n}=f$, and $(d_{i},d_{i+1})\in r^{\Imc}$ for all $i <n$.
 
Let $\Imc$ be an interpretation satisfying ${\sf cons}(\varphi)$ and refuting $\varphi$.
Assume, for simplicity, that $a^{\Imc}=b^{\Imc}$ for all $a,b$ that do not occur in $\varphi$.
Let $N$ denote the set of concepts all of the form
$$
\forall r_{1}.\cdots.\forall r_{n}.\neg\{a\},
$$
where $r_{1},\ldots,r_{n}$ are roles, $n\geq 0$ (thus the sequence can be empty),
and $a\in \NI$. Let $\Gamma$ denote the set of $\Lmc$-concepts $C$ such that
${\sf cons}(\varphi) \cup \{C^{\sharp}(x)\} \cup \{F^{\sharp}(x) \mid F\in N\}$ is 
satisfiable. Note that $\Gamma$ consists of exactly those $\Lmc$-concepts $C$ 
for which there exists an interpretation $\Jmc$ satisfying ${\sf cons}(\varphi)$ 
and a $d\in \Delta^{\Jmc}$ such that $d\in C^{\Jmc}$ and no nominal is interpreted
in the connected component generated by $d$.
 
Take for any $C\in \Gamma$ an interpretation $\Imc_{C}$ satisfying 
${\sf cons}(\varphi) \cup \{C^{\sharp}(x)\} \cup \{F^{\sharp}(x) \mid F\in N\}$.
Let $\Jmc_{C}$ denote the maximal component of $\Imc_{C}$ with 
${\sf Nom}(\Jmc_{C})=\emptyset$. Observe that $C$ is satisfied in $\Jmc_{C}$. 
Let $I= \Gamma \cup \{0\}$ and $\Jmc_{0}=\Imc_{0}=\Imc$. 
We can form the nominal disjoint union $\Imc^{-}= \sum_{i\in I}^{\sf nom}\Jmc_{i}$. Then 
\begin{itemize}
\item $\Imc^{-}$ refutes $\varphi$ (by condition (b));
\item $\Imc^{-}$ satisfies ${\sf cons}(\varphi)$;
\item for all $C\in \Gamma$, $C^{\Imc^{-}}\not=\emptyset$.
\end{itemize}
We can assume that $\Imc^{-}$ is $\omega$-saturated.

\medskip

\noindent
Claim 1. $\Gamma$ coincides with the set of concepts $C$ such that 
$\{\varphi,C^{\sharp}(x)\} \cup 
\{F^{\sharp}(x) \mid F\in N\}$ is satisfiable.

\medskip

To prove the claim assume there exists $C$ such that $\{\varphi,C^{\sharp}(x)\} \cup 
\{F^{\sharp}(x) \mid F\in N\}$ is not satisfiable, but ${\sf cons}(\varphi) \cup \{C^{\sharp}(c)\}\cup
\{F^{\sharp}(x) \mid F\in N\}$ is satisfiable. By compactness, 
$$
\varphi \models \bigsqcap_{F\in N'}F \sqsubseteq \neg C,
$$
for some finite subset $N'$ of $N$. 
But then $(\bigsqcap_{F\in N'}F \sqsubseteq \neg C)^{\sharp} \in {\sf cons}(\varphi)$ 
and we obtain a contradiction.

Let $X\subseteq \NI$ be a maximal set of individual names such that
\begin{itemize}
\item $a^{\Imc}\not\sim^{R}_{\Imc} b^{\Imc}$, for any two distinct $a,b\in X$;
\item for every $a\in \NI$ there is a $b\in X$ 
such that $a\sim^{R}_{\Imc} b$.
\end{itemize}
Note that $X$ is finite since $a^{\Imc}=b^{\Imc}$ for all $a,b$ that do not occur in 
$\varphi$.

\medskip

\noindent
Claim 2. For all $a\in X$, $\{\varphi\} \cup \{ C^{\sharp}(x) \mid C \in t^{\Imc}_{\Lmc}
(a^{\Imc})\}$
is satisfiable.

\medskip

Claim 2 follows from the fact that 
${\sf cons}(\varphi)\not\models C \sqsubseteq \bot$
for any $C\in t^{\Imc}_{\Lmc}(a^{\Imc})$, $a\in X$.

By Claim 1, we can take for every $C\in \Gamma$ an interpretation 
$\Imc_{C}'$ satisfying $\{\varphi,C^{\sharp}(x)\} \cup 
\{F^{\sharp}(x) \mid F\in N\}$. 
By Claim 2, we can take for every $a\in X$ an interpretation $\Imc_{a}$
satisfying $\{\varphi\} \cup \{ C^{\sharp}(x) \mid C \in t^{\Imc}_{\Lmc}(a^{\Imc})\}$. 

For $C\in \Gamma$, let $\Jmc_{C}'$ denote the maximal component of 
$\Imc_{C}'$ with ${\sf Nom}(\Jmc_{C}')=\emptyset$. 
Observe that $C$ is satisfied in $\Jmc_{C}'$. 

For $a\in X$, let $\Jmc_{a}$ denote the minimal component of $\Imc_{a}$ containing 
$a^{\Imc_{a}}$. Let $J= \Gamma \cup X$ and consider $\Imc^{+}=\sum^{\sf nom}_{j\in J}\Jmc_{j}$.
As $\varphi$ is preserved under nominal disjoint unions, $\Imc^{+}\models \varphi$.
We may assume that $\Imc^{+}$ is $\omega$-saturated. By definition,
\begin{itemize}
\item $t^{\Imc^{+}}_{\Lmc}(a^{\Imc^{+}}) = t^{\Imc^{-}}_{\Lmc}(a^{\Imc^{-}})$, for all $a\in \NI$;
\item for all $C\in \Gamma$, $C^{\Imc^{+}} \not=\emptyset$.
\end{itemize}
It follows that $\Imc^{-} \equiv_{\Lmc}^{g} \Imc^{+}$. Thus, $\Imc^{-} \sim_{\Lmc}^{g} \Imc^{+}$, 
and we have obtained a contradiction.
\end{proof}

\medskip

\noindent
{\bf Theorem~\ref{Boolnom}}
For Boolean $\mathcal{ALCIO}$-TBoxes, it is \ExpTime-complete to decide whether they
are equivalent to $\mathcal{ALCIO}$-TBoxes. This problem is co\NExpTime-complete for
Boolean $\mathcal{ALCQIO}$-TBoxes.

\medskip

\begin{proof} 
The lower bounds can be proved by a straightforward reduction from the \ExpTime-complete validity problem
for Boolean $\ALCIO$-TBoxes and the co-\NExpTime-complete validity problem for Boolean $\ALCQIO$-TBoxes, respectively.

Let $\Lmc\in \{\mathcal{ALCIO},\mathcal{ALCQIO}\}$.
The upper bound for $\Lmc$ is proved by a reduction to the validity problem for Boolean $\Lmc$-TBoxes.
Let $\vp$ be a Boolean $\Lmc$-TBox and let $X$ denote the set of nominals in $\vp$.
We may assume that $X\not=\emptyset$.
We reduce the problem of checking invariance under nominal disjoint unions of $\vp$.
Note that one can show by induction that it is sufficient to consider condition (a) for nominal 
disjoint unions of families $(\Imc_{i},\Jmc_{i})_{i\in I}$ in which ${\sf Nom}(\Jmc_{i}) \cap X=\emptyset$ for at most one
$i\in I$. Similarly, it is sufficient to consider condition (b) for nominal 
disjoint unions of families $(\Imc_{i},\Jmc_{i})_{i\in I}$ with $I$ of cardinality $2$.

With any partition $\Xi = \{X_{1},\ldots,X_{n}\}$ of $X$ (in which one $X_{i}$ can be the empty set)
we associate 
\begin{itemize}
\item a Boolean $\Lmc$-TBox $\vp_{\Xi}^{1}$ such that condition~(a) for invariance under nominal disjoint
unions holds for $\vp$ iff $\vp_{\Xi}^{1}$
is valid for all $\Xi$;
\item a Boolean $\Lmc$-TBox $\vp_{\Xi}^{2}$ such that condition~(b) for invariance under nominal disjoint unions 
holds for $\vp$ iff $\vp_{\Xi}^{2}$
is valid for all $\Xi$.
\end{itemize}
Assume $\Xi=\{X_{1},\ldots,X_{n}\}$ is given. 

To construct $\vp_{\Xi}^{1}$, choose concepts names $A_{1},\ldots,A_{n}$ and $B_{1},\ldots,B_{n}$. 
Denote by $\vp_{C}$ the relativization of $\varphi$ to $C$; i.e.,
the Boolean TBox such that any interpretation $\Imc$ is a model of $\vp_{C}$ iff the restriction of $\Imc$ to 
$C^{\Imc}$ is a model of $\vp$. Now let 
$$
\vp_{\Xi}^{1}= ((\chi \wedge (\bigwedge_{1\leq i \leq n}\vp_{B_{i}}) \rightarrow \vp_{C}),
$$
where $C = \bigsqcup_{1\leq i \leq n}A_{i}$ and $\chi$ is the conjunction of
\begin{itemize}
\item $A_{i} \sqsubseteq B_{i}$, for $1\leq i \leq n$;
\item $A_{i} \sqsubseteq \forall r.A_{i}$ for all roles $r$ in $\vp$ and $1\leq i \leq n$;
\item $\{a\} \sqsubseteq A_{i}$, for all $a\in X_{i}$ and $1\leq i \leq n$;
\item $B_{i} \sqcap B_{j}\sqsubseteq \bot$, for $1\leq i <j \leq n$;
\item $\neg (A_{i}\sqsubseteq \bot)$ for $1\leq i \leq N$;
\item $B_{i}\sqsubseteq \forall r.B_{i}$ for all roles $r$ in $\vp$ and $1\leq i \leq n$.
\end{itemize}
To prove our claim, observe that in any interpretation $\Imc$ satisfying $\chi$,
the interpretations $\Jmc_{i}$, $1\leq i\leq n$, induced by $A_{i}^{\Imc}$ and 
$\Imc_{i}$, $1\leq i\leq n$, induced by
$B_{i}^{\Imc_{i}}$ satisfy the conditions for nominal disjoint unions. 

To construct $\vp_{\Xi}^{2}$, choose concept names $A_{1},A_{2}$, and $B_{1},B_{2}$.
Then let
$$
\vp_{\Xi}^{2} = ((\chi \wedge \vp_{A_{1}\sqcup A_{2}}) \rightarrow \vp_{A_{1}}),
$$
where $\chi$ is the conjunction of
\begin{itemize}
\item $A_{1} \equiv B_{1}, A_{2}\sqsubseteq B_{2}$;
\item $A_{i} \sqsubseteq \forall r.A_{i}$ for all roles $r$ in $\vp$ and $i=1,2$;
\item $\{a\} \sqsubseteq A_{1}$, for all $a\in X$;
\item $B_{1} \sqcap B_{2}\sqsubseteq \bot$;
\item $\neg (A_{i}\sqsubseteq \bot)$ for $i=1,2$;
\item $B_{2}\sqsubseteq \forall r.B_{2}$ for all roles $r$ in $\vp$.
\end{itemize}
\end{proof}

\section{Proofs for Section~\ref{LightweightDLs}}

\medskip

\noindent
{\bf Theorem~\ref{DLLite-TBox}.}
Let $\Lmc\in \{\EL,\text{DL-Lite}_{horn}\}$ and let $\varphi$ be a first-order sentence.
The following conditions are equivalent:
\begin{enumerate}
\item $\varphi$ is equivalent to an $\Lmc$-TBox;
\item $\varphi$ is invariant under $\sim_{\Lmc}^{g}$ and disjoint unions, and preserved under
products.
\end{enumerate}

\noindent
\begin{proof}
  The proof of 1 $\Rightarrow$ 2 is straightforward. For the converse
  direction, in principle we follow the strategy of the proof of
  Theorem~\ref{thm:withdisjunion}. A problem is posed by the fact
  that, unlike in the case of expressive DLs, two $\omega$-saturated
  interpretations $\Imc^-$ and $\Imc^+$ that satisfy the same
  $\Lmc$-CIs need not satisfy $\Imc^- \equiv^g_\Lmc \Imc^+$ (e.g.\
  when $\Imc^-$ consists of two elements that satisfy $A$ and $B$,
  respectively, and $\Imc^+$ consists of two elements that satisfy no
  concept name and $A,B$, respectively).  To deal with this, we ensure
  that $\Imc^-$ and $\Imc^+$ satisfy the same \emph{disjunctive
    \Lmc-CIs}, i.e., CIs of the form $C \sqsubseteq D_1 \sqcup \cdots
  \sqcup D_n$ with $C, D_1,\dots,D_n$ \Lmc-concepts; this suffices to
  prove $\Imc^- \equiv_g \Imc^+$ as required.

  Let ${\sf cons}(\varphi)$ be the set of all \Lmc-CIs that are a
  consequence of $\vp$ and ${\sf cons}^\sqcup(\varphi)$ set of all
  disjunctive \Lmc-CIs that are a consequence of $\mn{cons}(\vp)$. As
  before, we are done when ${\sf cons}(\varphi) \models \varphi$, thus
  assume the opposite and derive a contradiction.

  Our aim is to construct interpretations $\Imc^-$ and $\Imc^+$ such
  that $\Imc^- \not\models \vp$, $\Imc^+ \models \vp$, and both
  $\Imc^-$ and $\Imc^+$ satisfy precisely those disjunctive \Lmc-CIs
  that are in ${\sf cons}^\sqcup(\varphi)$.

  $\Imc^-$ is constructed as follows.  For every disjunctive \Lmc-CI
  $C \sqsubseteq D_1 \sqcup \cdots \sqcup D_n \notin {\sf
    cons}^\sqcup(\varphi)$, take a model $\Imc_{C\not\sqsubseteq D_1
    \sqcup \cdots \sqcup D_n}$ of ${\sf cons}(\varphi)$ that violates
  \mbox{$C \sqsubseteq D_1 \sqcup \cdots \sqcup D_n$}. Then $\Imc^-$
  is the disjoint union of all $\Imc_{C\not\sqsubseteq D_1 \sqcup
    \cdots \sqcup D_n}$ and a model of \mbox{$\mn{cons}(\vp) \cup \{
    \neg \vp \}$}.  Clearly, $\Imc^-$ satisfies the desired properties.

  To construct~$\Imc^+$,
  first take for every \Lmc-CI $C \sqsubseteq D \notin {\sf
    cons}(\varphi)$ a model $\Imc_{C\not\sqsubseteq D}$ of $\varphi$
  that violates $C \sqsubseteq D$. Second, take for every disjunctive
  \Lmc-CI $C \sqsubseteq D_1 \sqcup \cdots \sqcup D_n \notin {\sf
    cons}^\sqcup(\varphi)$ the product
$$
\Jmc_{C \not\sqsubseteq D_{1}\sqcup \cdots \sqcup D_{n}} = \prod_{1\leq i \leq n}\Imc_{C \not\sqsubseteq D_{i}}
$$
Since $\vp$ is preserved under products and by Lemma~\ref{prodeldl},
each $\Jmc_{C \not\sqsubseteq (D_{1}\sqcup \cdots \sqcup D_{n})}$ is a
model of $\vp$ that violates $C \sqsubseteq D_{1}\sqcup \cdots \sqcup
D_{n}$. By defining $\Imc^+$ as the disjoint union of all $\Jmc_{C
  \not\sqsubseteq D_{1}\sqcup \cdots \sqcup D_{n}}$, we clearly attain
the properties desired for~$\Imc^+$.

It remains to show that $\Imc^- \equiv^g_\Lmc \Imc^+$, as then
Theorem~\ref{equivalence} implies $\Imc^- \sim^g_\Lmc \Imc^+$, in
contradiction to $\vp$ being invariant under $\sim^g_\Lmc$. We can
assume w.l.o.g.\ that $\Imc^-$ and $\Imc^+$ are $\omega$-saturated.
Take a $d \in \Delta^{\Imc^-}$. We have to show that there is an $e
\in \Delta^{\Imc^+}$ with $t^{\Imc^-}_\Lmc(d) =
t^{\Imc^+}_\Lmc(e)$. Let $\Gamma^+=t^{\Imc^-}_\Lmc(d)$ be the set of
\Lmc-concepts satisfied by $d$ in $\Imc^-$ and $\Gamma^-$ the set of
\Lmc-concepts not satisfied by $d$ in $\Imc^-$.  For any finite
$\Gamma^-_f \subseteq \Gamma^-$ and $\Gamma^+_f \subseteq \Gamma^+$,
there is an $e_{\Gamma^-_f,\Gamma^+_f} \in \Delta^{\Imc^+}$ such that
$e_{\Gamma^-_f,\Gamma^+_f} \in (\midsqcap \Gamma^+_f \sqcap \midsqcap
\Gamma^-_f)^{\Imc^+}$: since $\Imc^-$ does not satisfy $\midsqcap
\Gamma^+_f \sqsubseteq \midsqcup \Gamma^-_f$ neither does $\Imc^+$,
which yields the desired $e_{\Gamma^-_f,\Gamma^+_f}$. As $\Imc^+$ is
$\omega$-saturated, the existence of the $e_{\Gamma^-_f,\Gamma^+_f}$
implies the existence of an $e \in \Delta^{\Imc^+}$ such that $e \in
C^{\Imc^+}$ for all $C \in \Gamma^+$ and $e \notin C^{\Imc^+}$ for all
$C \in \Gamma^-$, i.e., $t^{\Imc^-}_\Lmc(d) =
t^{\Imc^+}_\Lmc(e)$. The direction from $\Imc^+$ to $\Imc^-$ is
analogous.

\end{proof}

We devide the proof of Theorem~\ref{ALCdown} into two parts and
reserve a subsection for each part.

\subsection{Proof of Theorem~\ref{ALCdown}: Invariance under $\sim^{g}_{\EL}$}

In this subsection, we prove the following result:

\begin{theorem}
The problem of deciding whether an $\mathcal{ALC}$-TBox $\Tmc$ is 
invariant under $\sim_{\EL}^{g}$ is \ExpTime-complete.
\end{theorem}

The lower bound proof is straightforward by a reduction of the
\ExpTime-hard satisfiability problem for $\ALC$-TBoxes:

\begin{lemma}
Let $\Tmc$ be an $\ALC$-TBox. The following conditions are equivalent
\begin{enumerate}
\item $\Tmc$ is satisfiable;
\item $\Tmc'=\Tmc\cup \{A \sqsubseteq \forall r.B\}$ is not
invariant under $\sim_{\EL}^{g}$ (where $A$, $B$, and $r$ are fresh).
\end{enumerate}
\end{lemma}
\begin{proof}
The direction $2\Rightarrow 1$ is trivial.
For the direction $1\Rightarrow 2$, assume that $\Tmc$ is satisfiable.
Let $\Imc$ be a model of $\Tmc$ such that $\Delta^{\Imc}$ has at least four 
elements, $d_{1},\ldots,d_{4}$ (such a model exists by invariance of $\ALC$-TBoxes 
under disjoint unions). Expand $\Imc$ to $\Imc_{1}$ and $\Imc_{2}$ by setting 
\begin{itemize}
\item $A^{\Imc_{1}}= A^{\Imc_{2}}=\{d_{1}\}$, 
\item $r^{\Imc_{1}}= r^{\Imc_{2}}=\{(d_{1},d_{2}),(d_{1},d_{3})\}$;
\item $B^{\Imc_{1}}= \{d_{2},d_{3}\}$, $B^{\Imc_{2}}=\{d_{2},d_{4}\}$, 
\end{itemize}
Clearly $\Imc_{1}$ is a model of $\Tmc'$, but $\Imc_{2}$ is not.
On the other hand, $\Imc_{1}\sim_{\EL}^{g} \Imc_{2}$. We show that 
$(\Imc_{1},d_{1}) \sim_{\EL} (\Imc_{2},d_{1})$, equi-simulations for the
remaining domain elements are straightforward. Now, 
$$
S_{1}= \{(d_{1},d_{1}),(d_{2},d_{2}),(d_{3},d_{2})\}
$$
is a $\EL$-simulation between $(\Imc_{1},d_{1})$ and $(\Imc_{2},d_{2})$.
Conversely,
$$
S_{2}= \{(d_{1},d_{1}),(d_{2},d_{2}),(d_{3},d_{3})\}
$$
is a $\EL$-simulation between $(\Imc_{2},d_{2})$ and $(\Imc_{1},d_{1})$.
\end{proof}
 
The upper bound proof is more involved. Firstly, we require
the following result about $\EL$-simulations:

\begin{lemma}\label{quick}
Let $(\Imc_1,d_1)\sim_{\EL}(\Imc_2,d_2)$ and let $\Imc_1,\Imc_{2}$ be 
$\omega$-saturated. Let $(d_{1},d_{1}')\in r^{\Imc_{1}}$. Then there exist
$d_{1}''$ and $d_{2}''$ with $(d_{1},d_{1}'')\in r^{\Imc_{1}}$ and
$(d_{2},d_{2}'')\in r^{\Imc_{2}}$ such that 
$$
d_1' \leq_{\EL} d_{1}'' \sim_{\EL} d_{2}''.
$$
\end{lemma}
\begin{proof} Let
$$
X = {\sf succ}^{\Imc_{1}}_{r}(d_{1}) \cap \{ d \mid (\Imc_{1},d_{1}')
\leq_{\EL} (\Imc_{1},d)\}
$$
We have $d_1'\in X$. $X$ is ordered by the simulation relation $\leq_{\EL}$.
Recall that, by Lemma~\ref{equivalence5}, for all $d,d'\in X$,
$d\leq_{\EL} d'$ iff $t^{\Imc_{1}}_{\EL}(d) \subseteq 
t^{\Imc_{1}}_{\EL}(d')$ since $\Imc_{1}$ is $\omega$-saturated.
\medskip

\noindent
Claim 1. $X$ contains a $\leq_{\EL}$-maximal element.

\medskip
To prove Claim 1 it is sufficient to show that for every $\leq_{\EL}$-ascending
chain $(e_i)_{i\in I}$ in $X$ there exists $e\in X$ such that $e_{i}\leq_{\EL}e$
for all $i\in I$. Consider the set of FO-formulas 
$$
\Gamma = \{ r(d_{1},x)\} \cup \{C^{\sharp}(x) \mid C \in \bigcup_{i\in I}t^{\Imc_{1}}_{\EL}(e_i)
\}.
$$
Clearly $\Gamma$ is finitely realizable in $\Imc_{1}$. By $\omega$-saturatedness,
$\Gamma$ is realizable in $\Imc_{1}$ for an assignment $a(x)\in X$. Let $e=a(x)$. Then
$e\in X$ and $e_{i}\leq_{\EL} e$ for all $i\in I$, as required.

Let $d_1''$ be a $\leq_{\EL}$-maximal element of $X$. 
Since $d_{1}\leq_{\EL} d_{2}$, there exists $d_{2}''\in {\sf succ}^{\Imc_{2}}_{r}(d_{2})$
such that $d_{1}''\leq_{\EL} d_{2}''$. Now $d_2'' \leq_{\EL} d_{1}''$ holds as well
because there exists $e\in X$ such that $d_{2}''\leq_{\EL} e$ and so 
$d_{1}''\leq_{\EL}d_{2}''$ implies $d_{2}''=e$ by $\leq_{\EL}$-maximality of $d_{1}''$ in 
$X$. We obtain $d_{1}''\sim_{\EL} d_{2}''$, as required.
\end{proof}

We are now in the position to prove the \ExpTime upper bound.
It is proved by means of a generalization of the type elimination
method to sequences of types rather than single types.
Given an $\ALC$-TBox $\Tmc$, by exponential time type elimination, we want to 
determine the set $P$ of all pairs $(t,s)$ of $\Tmc$-types such that there exist
$(\Imc,d)$ and $(\Jmc,d')$ with $t$ realized in $d$, $s$
realized in $d'$, and such that $\Jmc$ is a model of $\Tmc$,
$d \sim_{\EL} d'$ and $\Imc \sim_{\EL}^{g} \Jmc$. If $P$
contains a pair $(t,s)$ in which $t\in {\sf tp}\setminus{\sf tp}(\Tmc)$,
then $\Tmc$ is not preserved under $\sim_{\EL}^{g}$. If $P$ does not
contain any such pair, then $\Tmc$ is preserved under $\sim_{\EL}^{g}$. 
The straightforward idea of a recursive procedure that computes
$P$ by eliminating pairs from the set of all pairs $(t,s)$ with $t\in {\sf tp}$ and $s\in {\sf tp}(\Tmc)$
for which no appropriate witnesses for existential restrictions exist does not work:
the length of the sequences of types required as witnesses for existential restrictions 
grows. However, as in the interpretation
$\Imc$ we do not have to satisfy a fixed TBox, the role depth of the types
to be realized in $\Imc$ decreases and, therefore, the length of the sequences of types one has to consider
stabilizes after ${\sf rd}(\Tmc)$ man steps.
We now give a detailed proof.

For $m\geq 0$, by $\text{tp}^{m}$ we denote the set of all 
$t'\subseteq \text{sub}(\Tmc)$ 
such that there exists $t\in \text{tp}$ with
$$
t' = \{ C\in t \mid \text{rd}(C) \leq m\}.
$$
A $t\in \text{tp}^{m}$ is \emph{realized} by an object
$(\Imc,d)$ if $C\in d^{\Imc}$ for all $C\in t$.
Let $k$ be the role depth of the $\mathcal{ALC}$-TBox $\Tmc$.
For $m=0$ we set $m-1:=0$.

For $m,l\geq 0$ with $m+l\leq k$, we define $X_{l}^{m}$ as the set of all tuples
$$
(t,s,s_{0},t_{1},s_{1},\ldots,t_{l},s_{l}),
$$
such that 
\begin{itemize}
\item $t,t_{1},\ldots,t_{l}\in \text{tp}^{m}$,
\item $s,s_{0},s_{1},\ldots,s_{l} \in \text{tp}(\Tmc)$,
\end{itemize}
and there exist objects 
$(\Imc,d)$, $(\Imc_{1},d_{1})\ldots,(\Imc_{l},d_{l})$ and
$(\Jmc,d')$, $(\Jmc_{0},d_{0}'),\ldots,(\Jmc_{l},d_{l}')$
such that 
\begin{enumerate}
\item $(\Imc,d)$ realizes $t$ and $(\Imc_{i},d_{i})$ realizes $t_{i}$ for $1\leq i \leq l$;
\item $(\Jmc,d')$ realizes $s$ and $(\Jmc_{i},d_{i}')$ realizes $s_{i}$ for $0\leq i \leq l$;
\item $\Jmc$ and $\Jmc_{i}$ satisfy $\Tmc$, for $0 \leq i \leq l$;
\item $(\Imc,d)\leq_{\EL} (\Imc_{i},d_{i})$ for $1 \leq i \leq l$;
\item $(\Imc_{i},d_{i}) \sim_{\EL} (\Jmc_{i},d_{i}')$ and
$\Imc_{i} \sim_{\EL}^{g} \Jmc_{i}$ for $1\leq i \leq l$;
\item $(\Imc,d) \sim_{\EL} (\Jmc,d')$ and
$\Imc \sim_{\EL}^{g} \Jmc$;
\item $(\Jmc_{0},d_{0}') \leq (\Jmc,d')$. 
\end{enumerate}
\begin{lemma}
$\Tmc$ is not invariant under $\sim^{g}_{\EL}$ 
iff there exist $t\in {\sf tp} \setminus {\sf tp}(\Tmc)$,
and $s=s_{0}\in {\sf tp}(\Tmc)$ such that $(t,s,s_{0})\in X^{k}_{0}$.
\end{lemma}

Thus, the \ExpTime upper bound follows if one can compute $X^{k}_{0}$ in
exponential time. To this end, we will give an exponential time elimination algorithm 
that determines \emph{all} sets $X^{m}_{l}$, $0\leq m,l$ and $m+l\leq k$.

First compute the sets $\text{Init}_{l}^{m}$, $0\leq l,m$ and $l+m\leq k$,
consisting of all
$$
(t,s,s_{0},t_{1},s_{1},\ldots,t_{l},s_{l}),
$$
where $t,t_{1},\ldots,t_{l}\in \text{tp}^{m}$,
$s,s_{0},s_{1},\ldots,s_{l} \in \text{tp}(\Tmc)$
and for all $A\in \NC$: 
\begin{itemize}
\item $A \in t$ implies $A\in t_{i}$ for $1 \leq i \leq l$;
\item $A \in t_{i}$ iff $A \in s_{i}$ for $1\leq i \leq l$;
\item $A \in t$ iff $A \in s$;
\item $A \in s_{0}$ implies $A \in s$. 
\end{itemize}
Note that $\text{Init}_{l}^{m}$ can be computed in
exponential time since ${\sf tp}(\Tmc)$ can be computed in
exponential time.

Now apply exhaustively the rules from Figure~\ref{rules} to
the sets $Y^{m}_{l}:=\text{Init}^{m}_{l}$ and denote the resulting sets 
of tuples by $\text{Final}^{m}_{l}$.

\begin{figure}[t]
Let $(t,s,s_{0},t_{1},s_{1},\ldots,t_{l},s_{l})\in Y_{l}^{m}$.
\begin{center}
\begin{itemize}
\item[{(r1)}] if $m>0$ and there exists $\exists r.C \in t$ and 
such that there does not exist
$(t',s',s_{0}',t_{1}',s_{1}',\ldots,t_{l}',s_{l}',t_{l+1}',s_{l+1}')\in 
Y_{l+1}^{m-1}$ with $C\in t'$ and $t \leadsto_{r} t'$, 
$t\leadsto_{r} t_{l+1}'$, $s\leadsto_{r}s_{l+1}'$, and, for $1 \leq i \leq l$:
$t_{i} \leadsto_{r} t_{i}'$, $s_{i} \leadsto_{r} s_{i}'$, then set
$$
Y_{l}^{m}:= Y_{l}^{m}\setminus\{(t,s,s_{0},t_{1},s_{1},\ldots,t_{l},s_{l})\}
$$
\item[{(r2)}] if there exists $\exists r.C \in s$ and there does not exist
$(t',s',s_{0}',t_{1}',s_{1}',\ldots,t_{l}',s_{l}')\in Y_{l}^{m-1}$ with
$C\in s_{0}'$ and $s \leadsto_{r} s_{0}'$, $s \leadsto_{r} s'$,
$t \leadsto_{r} t'$, and, for $1 \leq i \leq l$:
$t_{i} \leadsto_{r} t_{i}'$, $s_{i} \leadsto_{r} s_{i}'$, then set
$$
Y_{l}^{m}:= Y_{l}^{m}\setminus\{(t,s,s_{0},t_{1},s_{1},\ldots,t_{l},s_{l})\}
$$
\item[{(r3)}] if there exists $\exists r.C \in s_{0}$ and there does not exist
$(t',s',s_{0}',t_{1}',s_{1}',\ldots,t_{l}',s_{l}')\in Y_{l}^{m-1}$ with
$C\in s_{0}'$ and $s_{0} \leadsto_{r} s_{0}'$, $s\leadsto_{r} 
s'$,
$t \leadsto_{r} t'$, and, for $1 \leq i \leq l$:
$t_{i} \leadsto_{r} t_{i}'$, $s_{i} \leadsto_{r} s_{i}'$, then set
$$
Y_{l}^{m}:= Y_{l}^{m}\setminus\{(t,s,s_{0},t_{1},s_{1},\ldots,t_{l},s_{l})\}
$$
\item[{(r4)}] if there exist $1 \leq i \leq l$ and $\exists r.C \in t_{i}$ 
such that there does not exist
$(t',s',s_{0}',t_{1}',s'_{1}) \in Y^{m-1}_{1}$ with
$C\in t'$ and $t_{i} \leadsto_{r} t'$, $t_{i}\leadsto_{r} t_{1}'$, 
$s_{i}\leadsto_{r} s_{1}'$, then set
$$
Y_{l}^{m}:= Y_{l}^{m}\setminus\{(t,s,s_{0},t_{1},s_{1},\ldots,t_{l},s_{l})\}
$$
\item[{(r5)}] if there exist $1 \leq i \leq l$ and $\exists r.C \in s_{i}$ 
such that there does not exist
$(t',s',s_{0}') \in Y^{m-1}_{0}$ with
$C\in s_{0}'$ and $s_{i} \leadsto_{r} s_{0}'$, 
$s_{i}\leadsto_{r} s'$, $t_{i}\leadsto_{r} t'$, then set
$$
Y_{l}^{m}:= Y_{l}^{m}\setminus\{(t,s,s_{0},t_{1},s_{1},\ldots,t_{l},s_{l})\}
$$
\end{itemize}
    \caption{Elimination Rules}
    \label{rules}
  \end{center}
\end{figure}

It should be clear that the elimination algorithm terminates after
at most exponentially many steps. Thus, the lower bound follows from
the following lemma:

\begin{lemma}
For all $m,l\geq 0$ with $m+l\leq k$, we have $X^{m}_{l}=\text{Final}^{m}_{l}$.
\end{lemma}
\begin{proof}
We start with the proof of the inclusion $X^{m}_{l}\subseteq 
\text{Final}^{m}_{l}$.
To this end, it is sufficient to observe that $X^{m}_{l}\subseteq 
\text{Init}^{m}_{l}$ and that the following holds for $1\leq i \leq 5$:

\medskip

\noindent
Claim 1. If $X^{m}_{l}\subseteq Y^{m}_{l}$ for all $0\leq m,l$ with $m+l\leq k$,
and $\vec{x}$ is removed from $Y^{m_{0}}_{l_{0}}$ by an application of the rule 
(ri), then $\vec{x}\not\in X^{m_{0}}_{l_{0}}$.

\medskip

To prove the claim, first let $i=1$. Assume that, in contrast to
what has to be shown, 
there are $(t,s,s_{0},t_{1},s_{1},\ldots,t_{l_{0}},s_{l_{0}})\in Y^{m_{0}}_{l_{0}}$
with $m_{0}>0$ and $\exists r.C \in t$ such that 
\begin{itemize}
\item (r1) is applicable: there does not exist
$(t',s',s_{0}',t_{1}',s_{1}',\ldots,t_{l_{0}}',s_{l_{0}}',t_{l_{0}+1}',s_{l_{0}+1}')\in 
Y_{l_{0}+1}^{m_{0}-1}$ with ($\ast$) $C\in t'$ and $t \leadsto_{r} t'$, 
$t\leadsto_{r} t_{l_{0}+1}'$, $s\leadsto_{r}s_{l_{0}+1}'$, and, for 
$1 \leq i \leq l_{0}$:
$t_{i} \leadsto_{r} t_{i}'$, $s_{i} \leadsto_{r} s_{i}'$;
\item $(t,s,s_{0},t_{1},s_{1},\ldots,t_{l_{0}},s_{l_{0}})\in X^{m_{0}}_{l_{0}}$.
\end{itemize}
By Point~2, we can take objects 
$(\Imc,d),(\Imc_{1},d_{1}),\ldots,(\Imc_{l_{0}},d_{l_{0}})$ and
$(\Jmc,d'),(\Jmc_{1},d_{1}'),\ldots,(\Jmc_{l_{0}},d_{l_{0}}')$ with the
properties 1--7. We may assume that those objects are $\omega$-saturated.
We find $e$ with $(d,e)\in r^{\Imc}$ such that $e\in C^{\Imc}$.
Let $1\leq i \leq l_{0}$. We find $f_{i}\in \Imc_{i}$ 
with $(d_{i},f_{i})\in r^{\Imc_{i}}$ and 
$(\Imc,e) \leq_{\EL} (\Imc_{i},f_{i})$. By Lemma~\ref{quick}, we find
$e_{i}$ and $e_{i}'$ with $(d_{i},e_{i})\in r^{\Imc_{i}}$ and
$(d_{i}',e_{i}')\in r^{\Jmc_{i}}$ such that
$$
(\Imc,f_{i}) \leq_{\EL} (\Imc_{i},e_{i}) \sim_{\EL} (\Jmc_{i},e_{i}')
$$
We also have $(\Imc,e)\leq_{\EL} (\Imc_{i},e_{i})$.
Also, by Lemma~\ref{quick}, we find $e_{l_{0}+1}$ and $e_{l_{0}+1}'$
with $(d,e_{l_{0}+1})\in r^{\Imc}$ and $(e,e_{l_{0}+1}')\in r^{\Jmc}$
such that 
$$
(\Imc,e)\leq_{\EL} (\Imc,e_{l_{0}+1}) \sim_{\EL} (\Jmc,e_{l_{0}+1}')
$$
Set $\Imc_{l_{0}+1}= \Imc$ and $\Jmc_{l_{0}+1}=\Jmc$.
Now let
\begin{itemize}
\item $t'$ be the type in $\text{tp}^{m_{0}-1}$ realized by $(\Imc,e)$;
\item $t_{i}'$, $1\leq i \leq l_{0}+1$, be the type in $\text{tp}^{m_{0}-1}$ 
realized
by $(\Imc_{i},e_{i})$;
\item $s_{i}'$, $1\leq i \leq l_{0}+1$, be the $\Tmc$-types realized by 
$(\Jmc_{i},e_{i}')$;
\item $s'=s_{0}'$ be the $\Tmc$-type realized by some $(\Kmc,f)$
such that $\Kmc$ is a model of $\Tmc$,
$(\Imc,e)\sim_{\EL}(\Kmc,f)$, and $\Imc\sim_{\EL}^{g}\Kmc$.
\end{itemize}
Let $\vec{x}=(t',s',s_{0}',t_{1}',s_{1}',\ldots,t_{l_{0}}',s_{l_{0}}',
t_{l_{0}+1}',s_{l_{0}+1}')$. Clearly $\vec{x}\in X^{m_{0}-1}_{l_{0}+1}$ and so 
$\vec{x}\in Y^{m_{0}-1}_{l_{0}+1}$. Moreover $\vec{x}$ satisfies ($\ast$).
Thus, we have derived a contradiction.

The rules (r2)--(r5) are considered similarly.

\medskip

We now come to the inclusion $X^{m}_{l}\supseteq Z^{m}_{l}$.
Denote the nth entry of $\vec{x}\in Z^{m}_{l}$ by $\vec{x}(n)$;
$l(\vec{x})$ denotes the length on $\vec{x}$.

Now define an interpretation $\Imc$ by setting
$$
\Delta^{\Imc} = \{(n,\vec{x}) \mid \vec{x}\in Z^{m}_{l},l(\vec{x})\geq n\}.
$$
and, for $A\in \NC$, 
$$
A^{\Imc} = \{ (n,\vec{x})\in \Delta^{\Imc} \mid A\in \vec{x}(n)\}.
$$
Finally, for $r\in \NR$,
$$
\vec{x} = (t,s,s_{0},
t_{1},s_{1},\ldots,t_{l},s_{l})\in Z^{m}_{l},
$$
and $(n,\vec{y})\in \Delta^{\Imc}$ we set 
$((n,\vec{x}),(m,\vec{y}))\in r^{\Imc}$ if
\begin{itemize}
\item $\vec{y}= 
(t',s',s_{0}',t_{1}',s_{1}',\ldots,t_{l}',s_{l}',t_{l+1}',s_{l+1}')\in 
Z_{l+1}^{m-1}$, 
$(\vec{x}(n),\vec{y}(m))$ is one of the pairs
$(t,t')$, $(t,t_{l+1}')$, $(s,s_{l+1}')$, 
$(t_{i},t_{i}'),(s_{i},s_{i}')$, for $1 \leq i \leq l$,
and 
$t \leadsto_{r} t'$, 
$t\leadsto_{r} t_{l+1}'$, $s\leadsto_{r}
s_{l+1}'$, and, for $1 \leq i \leq l$:
$t_{i} \leadsto_{r} t_{i}'$, $s_{i} \leadsto_{r} s_{i}'$.
\item $\vec{y} = (t',s',s_{0}',t_{1}',s_{1}',\ldots,t_{l}',s_{l}')
\in Z_{l}^{m-1}$,
$(\vec{x}(n),\vec{y}(m))$ is one of the pairs
$(s,s_{0}')$, $(s,s')$, $(t,t')$,
$(t_{i},t_{i}'),(s_{i},s_{i}')$, for $1 \leq i \leq l$,
and $s \leadsto_{r} s_{0}'$, $s\leadsto_{r} 
s'$,
$t \leadsto_{r} t'$, and, for $1 \leq i \leq l$:
$t_{i} \leadsto_{r} t_{i}'$, $s_{i} \leadsto_{r} s_{i}'$.
\item $\vec{y}=(t',s',s_{0}',t_{1}',s_{1}',\ldots,t_{l}',s_{l}')\in Z_{l}^{m-1}$,
$(\vec{x}(n),\vec{y}(m))$ is one of the pairs
$(s_{0},s_{0}')$, $(s,s')$, $(t,t')$
$(t_{i},t_{i}'),(s_{i},s_{i'})$, for $1 \leq i \leq l$,
and $s_{0} \leadsto_{r} s_{0}'$, $s\leadsto_{r} 
s'$,
$t \leadsto_{r} t'$, and, for $1 \leq i \leq l$:
$t_{i} \leadsto_{r} t_{i}'$, $s_{i} \leadsto_{r} s_{i}'$.
\item there exists $1 \leq i \leq l$ such that such that 
$\vec{y}= (t',s',s_{0}',t_{1}',s'_{1}) \in Z^{m-1}_{1}$,
$$
(\vec{x}(n),\vec{y}(m))\in \{(t_{i},t'),(t_{i},t_{1}'),(s_{i},s_{1}')\}
$$
and $t_{i} \leadsto_{r} t'$, $t_{i}\leadsto_{r} t_{1}'$, 
$s_{i}\leadsto_{r} s_{1}'$.
\item there exists $1 \leq i \leq l$ such that 
$\vec{y}= (t',s',s_{0}') \in Z^{m-1}_{0}$,
$$
(\vec{x}(n),\vec{y}(m))\in \{(s_{i},s_{0}'),(s_{i},s'),(t_{i},t')\},
$$
and $s_{i} \leadsto_{r} s_{0}'$, 
$s_{i}\leadsto_{r} s'$, $t_{i}\leadsto_{r} t'$.
\end{itemize}

The following can be proved by induction:

\medskip
\noindent Claim 1. For all $(n,\vec{x})\in \Delta^{\Imc}$,
if $\vec{x}(n)\in \text{tp}^{m}$ for some $m\leq k$
and $C\in \text{sub}(\Tmc)$ has role depth $\leq m$, then 
$$
C\in \vec{x}(n)\quad \Leftrightarrow \quad (n,\vec{x})\in C^{\Imc}.
$$

\medskip
\noindent Claim 2. If $\vec{x} = (t,s,s_{0},
t_{1},s_{1},\ldots,t_{l},s_{l})$ and $(n,\vec{x}),(m,\vec{x}) \in \Delta^{\Imc}$.
Then $(n,\vec{x}) \leq_{\EL} (m,\vec{x})$ whenever 
$$
(\vec{x}(n),\vec{x}(m))\in \{(s_{0},s)\} \cup
\{(t,t_{i})\mid 1 \leq i \leq l\}
$$
and $(n,\vec{x}) \sim_{\EL} (m,\vec{x})$ whenever 
$$
(\vec{x}(n),\vec{x}(m))\in \{(t,s)\} \cup
\{(t_{i},s_{i})\mid 1 \leq i \leq l\}
$$

Now let $\Imc_{s}$ be the interpretation induced by $\Imc$ on
the set of all $(n,\vec{x})\in \Delta^{\Imc}$ such that $\vec{x}(n)
\in \{s,s_{0},\ldots,s_{l}\}$ for 
$\vec{x} = (t,s,s_{0},t_{1},s_{1},\ldots,t_{l},s_{l})$. Let $\Imc_{t}$
be the interpretation induced by $\Imc$ on $\Delta^{\Imc}\setminus \Delta^{\Imc_{s}}$.
Observe that $\Imc$ is the \emph{disjoint union} of $\Imc_{s}$ and $\Imc_{t}$.

Now assume that $\vec{x}=(t,s,s_{0},t_{1},s_{1},\ldots,t_{l},s_{l})\in Z^{m}_{l}$
is given. We set 
\begin{itemize}
\item $\Imc=\Imc_{1}=\cdots = \Imc_{l}:=\Imc_{t}$;
\item $d= (1,\vec{x})$ and, for $1\leq i \leq l$, $d_{i}=(2+2i,\vec{x})$;
\item $\Jmc=\Jmc_{0}=\cdots = \Jmc_{l}:=\Imc_{t}$;
\item $d'=(2,\vec{x})$ and, for $1\leq i \leq l$, $d_{i}'=(3+2i,\vec{x})$.
\end{itemize}
It follows from Claims 1 and 2 that the defined objects satisfy the conditions
1--7. Thus, $\vec{x}\in X^{m}_{l}$, as required.
\end{proof}

\subsection{Proof of Theorem~\ref{ALCdown}: preservation under products}

The aim of this subsection is to prove the following result:

\begin{theorem}\label{prodnew}
  It is co-\NExpTime-complete to decide whether an \ALC-TBox is preserved
  under products.
\end{theorem}

We start with the upper bound proof.
An interpretation $\Imc$ is a \emph{tree interpretation} if the directed
graph $(\Delta^{\Imc},\bigcup_{r\in \NR}r^{\Imc})$ is a tree and $r^{\Imc}\cap s^{\Imc}=\emptyset$
for any two distinct $r,s\in \NR$.


\begin{lemma}
\label{lem:productunravel}
  If an \ALC-TBox \Tmc is not preserved under products, then there 
  are tree-models $\Imc_1$ and $\Imc_2$ of \Tmc with out-degree
  at most $2^{n^2+n+1}$ such that $\Imc_1 \times \Imc_2$ is not a
  model of \Tmc.
\end{lemma}
\begin{proof}
  Assume that $\Tmc= \{ \top \sqsubseteq C_\Tmc \}$ is not preserved
  under products. Then there are models $\Imc_1$ and $\Imc_2$ of \Tmc
  such that $\Imc_1 \times \Imc_2$ is not a model of \Tmc. Thus, there
  is a $(\hat d_1,\hat d_2) \in \Delta^{\Imc_1 \times \Imc_2}$ with
  $(\hat d_1,\hat d_2) \notin C_\Tmc^{\Imc_1 \times \Imc_2}$.  We
  proceed in two steps: first unravel $\Imc_1$ and $\Imc_2$ into
  tree-interpretations, then restrict their out-degree.  An
  \emph{$i$-path}, $i \in \{1,2\}$, is a sequence $d_0r_0d_1r_1 \cdots
  r_{k-1}d_k$, $k \geq 0$, alternating between elements of
  $\Delta^{\Imc_i}$ and role names that occur in \Tmc such that
  $d_0=\hat d_i$ and for all $i < k$, we have $(d_i,d_{i+1}) \in
  r_i^\Imc$. Define new interpretations $\Imc'_1$ and $\Imc'_2$ as
  follows:
  $$
  \begin{array}{r@{\,}c@{\,}l}
    \Delta^{\Imc'_i} &=& \text{the set of $i$-paths} \\[1mm]
    A^{\Imc'_i} &=& \{ d_0 \cdots d_k \in \Delta^{\Imc'_i} \mid d_k \in A^{\Imc_i}  \} \\[1mm]
    r^{\Imc'_i} &=& \{ (d_0 \cdots d_k,d_0 \cdots d_k r d_{k+1}) \mid d_0 \cdots d_k r d_{k+1} \in \Delta^{\Imc'_i} \}.
  \end{array}
  $$
  It can be proved by a straightforward induction that for all $C \in \mn{sub}(\Tmc)$ and $d_0
  \cdots d_k \in \Delta^{\Imc'_i}$, $i \in \{1,2\}$, we have $d_k \in
  C^{\Imc_i}$ iff $d_0 \cdots d_k \in C^{\Imc'_i}$. It follows that
  $\Imc'_1$ and $\Imc'_2$ are models of \Tmc. To show that $\Imc'_1 \times
  \Imc'_2$ is not a model of \Tmc, it suffices to establish the following
  claim, which yields $(\hat d_1, \hat d_2) \in (\neg C_\Tmc)^{\Imc'_1 \times \Imc'_2}$:
  \\[2mm]
  {\bf Claim}. For all $C \in \mn{sub}(\Tmc)$, $p_1=d^1_0 \cdots d^1_{k_1}
  \in \Delta^{\Imc'_1}$, and $p_2=d^2_0 \cdots d^2_{k_2} \in
  \Delta^{\Imc'_2}$, we have $(d^1_{k_1},d^2_{k_2}) \in C^{\Imc_1 \times
    \Imc_2}$ iff $(p_1,p_2) \in C^{\Imc'_1 \times \Imc'_2}$.
  \\[2mm]
  The proof is by induction on the structure of $C$, where the only interesting
  case is $C=\exists r . D$.

  \smallskip First let $(d^1_{k_1},d^2_{k_2}) \in (\exists r
  .D)^{\Imc_1 \times \Imc_2}$. Then there is a $(d_1,d_2) \in D^{\Imc_1
    \times \Imc_2}$ with $((d^1_{k_1},d^2_{k_2}),(d_1,d_2)) \in r^{\Imc_1
    \times \Imc_2}$.  It follows that $(d^1_{k_1},d_1) \in r^{\Imc_1}$
  and $(d^2_{k_2},d_2) \in r^{\Imc_2}$. Thus, $p_1 r d_1$ is a 1-path and
  $p_2 r d_2$ is a 2-path. Then $(p_1,p_1rd_1) \in r^{\Imc_1}$ and
  $(p_2,p_2rd_2) \in r^{\Imc_2}$ and $((p_1,p_2),(p_1rd_1,p_2rd)) \in
  r^{\Imc'_1 \times \Imc'_2}$. By IH, $(p_1rd_1,p_2rd_2) \in D^{\Imc'_1
    \times \Imc'_2}$ and we are done.
  
  Now let $(p_1,p_2) \in (\exists r . D)^{\Imc'_1 \times
    \Imc'_2}$. Then there are $(q_1,q_2) \in D^{\Imc'_1 \times
    \Imc'_2}$ such that $((p_1,p_2),(q_1,q_2)) \in r^{\Imc'_1 \times
    \Imc'_2}$. By definition of products and $\Imc'_1$ and $\Imc'_2$,
  we have $q_1 = p_1rd_1$ and $p_2 = p_2 r d_2$ for some $d_1 \in
  \Delta^{\Imc_1}$ and $d_2 \in \Delta^{\Imc_2}$. Since $q_1$ is a
  1-path and $q_2$ a 2-path, we have $(d^1_{k_1},d_1) \in r^{\Imc_1}$
  and $(d^2_{k_2},d_2) \in r^{\Imc_2}$, thus
  $((d^1_{k_1},d^2_{k_2}),(d_1,d_2)) \in r^{\Imc_1 \times \Imc_2}$.
  By IH, $(d_1,d_2) \in D^{\Imc_1 \times \Imc_2}$ and we are done.
  
  \medskip We now define interpretations $\Imc''_1$ from $\Imc'_1$ and
  $\Imc''_2$ from $\Imc'_2$ by dropping `unnecessary' subtrees, which
  results in a reduction of the maximum out-degree to $2^{n^2+n+1}$.  To select
  the subtrees in $\Imc'_1$ and $\Imc'_2$ that must not be dropped, we
  first need a notion of distance in the product interpretation
  $\Imc'_1 \times \Imc'_2$: for all $(p_1,p_2) \in \Imc'_1 \times
  \Imc'_2$, let $\delta_{12}(p_1,p_2)$ denote the length of the path
  from $(\hat d_1, \hat d_2)$ to $(p_1,p_2)$ in $\Imc'_1 \times
  \Imc'_2$, if such a path exists (note that the path is unique if it
  exists); otherwise, $\delta_{12}(p_1,p_2)$ is undefined.  Now choose
  for each $i \in \{1,2\}$, a smallest set $\Gamma_i \subseteq
  \Delta^{\Imc'_i}$ such that the following conditions are satisfied:
  \begin{enumerate}

  \item[(a)] $\hat d_i \in \Gamma_i$;

  \item[(b)] whenever $p \in \Gamma_i$ and
    $\exists r . C \in \mn{sub}(\Tmc)$ with $p \in (\exists r
    . C)^{\Imc'_i}$, then there is a $p' \in \Gamma_i$ such that
    $(p,p') \in r^{\Imc'_i}$ and $p' \in C^{\Imc'_i}$;

  \item[(c)] whenever $p_1 \in \Gamma_1$ and $p_2 \in \Gamma_2$ and
    $\exists r . C \in \mn{sub}(\Tmc)$ with $(p_1,p_2) \in (\exists r
    . C)^{\Imc'_1 \times \Imc'_2}$, $\delta_{12}(p_1,p_2)$ is defined,
    and $\mn{rd}(\exists r . C) \leq |\Tmc|-\delta_{12}(p_1,p_2)$,
    then there is a $(p'_1,p'_2) \in C^{\Imc'_1 \times \Imc'_2}$ such
    that $((p_1,p_2),(p'_1,p'_2)) \in r^{\Imc'_1 \times \Imc'_2}$,
    $p_1 \in \Gamma_1$, and $p_2 \in \Gamma_2$.

  \end{enumerate}
  Now let $\Imc''_i$ be obtained from $\Imc'_i$ by dropping all
  subtrees whose root is not in $\Gamma_i$, for $i \in \{1,2\}$. The
  following can be proved by a straightforward structural induction.
  \\[2mm]
  {\bf Claim}. For all $C \in \mn{sub}(\Tmc)$, $p_1 \in \Delta^{\Imc''_1}$,
  $p_2 \in \Delta^{\Imc''_2}$, and $i \in \{1,2\}$, we have
  \begin{enumerate}

  \item $p_i \in C^{\Imc'_i}$ iff $p_i \in C^{\Imc''_i}$;

  \item $(p_1,p_2) \in C^{\Imc'_1 \times \Imc'_2}$ iff $(p_1,p_2) \in C^{\Imc''_1 \times \Imc''_2}$
    whenever $\delta_{12}(p_1,p_2)$ is defined and $\mn{rd}(C) \leq |\Tmc|-\delta_{12}(p_1,p_2)$.

  \end{enumerate}
  It follows that $\Imc''_1$ and $\Imc''_2$ are still models of \Tmc,
  and that $(\hat d_1, \hat d_2) \notin C_\Tmc^{\Imc''_1 \times
    \Imc''_2}$, thus $\Imc''_1 \times \Imc''_2$ is not a model of
  \Tmc. It remains to verify that the out-degree of $\Imc''_1$ and
  $\Imc''_2$ is bounded by $2^{n^2+n+1}$. First define distance functions $\delta_1$
  and $\delta_2$ in $\Imc'_1$ and $\Imc'_2$, analogously to the definition
  of $\delta_{12}$. Let $|\Tmc|=n$, $f(0)=2n$ and $f(i)=n + n \cdot f(i-1)$
  for all $i > 0$. We establish the following
  \\[2mm]
  {\bf  Claim.} For all $i \in \{1,2\}$ and $p \in \Delta^{\Imc''}$, 
  \begin{enumerate}

  \item $p$ has at most $f(\delta_i(p))$ successors;

  \item $p$ has at most $n$ successors if $\delta_i(p) \geq n$.

  \end{enumerate}
  Since Point~2 is obvious by our use of $\delta_{12}$ in
  Condition~(c) of the definition of $\Gamma_1$ and $\Gamma_2$, we
  concentrate on Point~1 of the claim. It is proved by induction on
  $\delta_i(p)$. For the induction start, let $\delta_i(p)=0$, i.e.,
  $p=\hat d_i$. Then $p$ has at most $n$ successors selected due to
  Condition~(b) of the definition of $\Gamma_1$ and $\Gamma_2$ and at
  most $n$ successors selected due to Condition~(c). It remains to
  remind that $f(0)=2n$. For the induction step, we concentrate on the
  case $i=1$; the case $i=2$ is symmetric. Thus, let $\delta_1(p)>0$.
  Again, at most $n$ successors are selected due to Condition~(b) of
  the definition of $\Gamma_1$ and $\Gamma_2$. In $\Imc_1' \times
  \Imc'_2$, the number of elements $(p,q)$ for which
  $\delta_{12}(p,q)$ is defined is bounded by the maximal number of
  successors of elements $q' \in \Delta^{\Imc'_2}$ with $\delta_2(p')
  = \delta_1(p) -1$; the reason is that $\delta_{12}(p,q)$ is defined
  only if the predecessor $(p',q')$ of $(p,q)$ satisfies the
  following properties:
  \begin{enumerate}

  \item $p'$ is the unique predecessor of $p$ in $\Imc_1$;

  \item $q$ is a successor of $q'$ in $\Imc_2$.

  \end{enumerate}
  By IH, there are thus at most $f(\delta_1(p)-1)$ such elements
  $(p,q)$. For each such $(p,q)$, at most $n$ successors of $p$
  are selected in Condition~(c) of the definition of $\Gamma_1$ and
  $\Gamma_2$.  Thus, the maximum number of successors of $p$ is 
  $$n+
  n*f(\delta_1(p)-1) = f(\delta_1(p)).
  $$
  This finishes the proof of the claim.
  Now, an easy analysis of the recurrence in Point~1 of the above claim
  yields a maximum out-degree of $2^{n^2+n+1}$.
\end{proof}
%
For an interpretation \Imc and a $d \in \Delta^\Imc$, we use
$\mn{tp}_\Tmc^\Imc$ to denote the \emph{semantic \Tmc-type} of $d$ in
\Imc, i.e., $\mn{tp}_\Tmc^\Imc(d) = \{ C \in \mn{sub}(\Tmc) \mid d \in
C^\Imc\}$. The set of all semantic \Tmc-types is
$$
  \Tmf = \{ \mn{tp}^\Imc(d) \mid \Imc \text{ a model of } \Tmc,\ d \in \Delta^\Imc \}.
$$
For $t_1,t_2 \in \Tmf$, set $t_{1} \leadsto_{r} t_{2}$ if we have
$\exists r .C \in t_1$ iff $C \in t_2$, for all $\exists r . C \in
\mn{sub}(\Tmc)$. For $k \geq 0$, a \emph{$k$-initial interpretation
  tree} is a triple $(\Imc,\rho^\Imc,t^\Imc)$, where \Imc is a
tree-shaped interpretation of depth at most $k$ and with root
$\rho^\Imc$ and $t^\Imc: \Delta^\Imc \rightarrow \Tmf$. For $d \in
\Delta^\Imc$, we use $\delta_\Imc(d)$ to denote the distance of $d$
from the root of \Imc. We require that the following conditions are
satisfied, for all $d,e \in \Delta^\Imc$:
\begin{enumerate}

\item $d \in A^\Imc$ iff $A \in t^\Imc(d)$ for all $A \in \NC$;

\item if $\exists r . C \in t^\Imc(d)$ and $\delta_\Imc(d) < k$, then
  there is an $e \in \Delta^\Imc$ with $(d,e) \in r^\Imc$ and
  $C \in t^\Imc(e)$;

\item if $(d,e) \in r^\Imc$, then $t^\Imc(d) \leadsto t^\Imc(e)$.

\end{enumerate}
When we speak about the \emph{product} $\Imc_1 \times \Imc_2$ of two
$k$-initial interpretation trees $\Imc_1$ and~$\Imc_2$, we simply mean
the product of the interpretations $(\Delta^{\Imc_1},\cdot^{\Imc_1})$
and $(\Delta^{\Imc_2},\cdot^{\Imc_2})$, i.e., the annotating
components $\rho^{\Imc_i}$ and $t^{\Imc_i}$ are dropped before forming
the product.
\begin{lemma}
\label{lem:charproducts}
An \ALC-TBox $\Tmc = \{ \top \sqsubseteq C_\Tmc \}$ is not preserved
under products iff there are $n$-initial interpretation trees $\Imc_1$
and $\Imc_2$ of maximum out-degree $2^{n^2+n+1}$ such that
$(\rho_1,\rho_2) \notin C_\Tmc^{\Imc_1 \times \Imc_2}$, where
$n=|\Tmc|$.
\end{lemma}
\begin{proof}
  First assume that \Tmc is not preserved under products. By
  Lemma~\ref{lem:productunravel}, there are tree-shaped models
  $\Jmc_1$ and $\Jmc_2$ of \Tmc of maximum out-degree $2^{n^2+2}$ such
  that $\Jmc_1 \times \Jmc_2$ is not a model of \Tmc. Let $\rho_i$ be
  the root of $\Jmc_i$, for $i \in \{1,2\}$.  W.l.o.g., we can assume
  that $(\rho_1,\rho_2) \notin C_\Tmc^{\Jmc_1 \times \Jmc_2}$ (if this
  is not the case, replace $\Jmc_1$ and $\Jmc_2$ by suitable subtrees
  of these models).  Define $n$-initial interpretation trees $\Imc_1$
  and $\Imc_2$ by starting with $\Jmc_1$ and $\Jmc_2$, removing all
  nodes of depth exceeding $n$, and adding the annotations
  $\rho^{\Imc_i}=\rho_i$ and $t^{\Imc_i}$, where the latter is defined
  by setting $t^{\Imc_i}(d)= \mn{tp}_\Tmc^{\Jmc_i}(d)$ for all $d \in
  \Delta^{\Imc_i}$.  It remains to show that
  $(\rho^{\Imc_1},\rho^{\Imc_2}) \notin C_\Tmc^{\Imc_1 \times
    \Imc_2}$, which is an immediate consequence of $(\rho_1,\rho_2)
  \notin C_\Tmc^{\Jmc_1 \times \Jmc_2}$ and the following claim, whose
  proof is left to the reader. For all $(d_1,d_2) \in
  \Delta^{\Imc_1 \times \Imc_2}$, we use $\delta_{12}(d_1,d_2)$ to
  denote the length of the unique path from
  $(\rho^{\Imc_1},\rho^{\Imc_2})$ to $(d_1,d_2)$ in $\Imc_1 \times
  \Imc_2$ if such a path exists; otherwise, $\delta_{12}(d_1,d_2)$ is
  undefined. 
  \\[2mm]
  {\bf Claim}. For all $(d_1,d_2) \in \Delta^{\Imc_1 \times \Imc_2}$
  with $\delta_{12}(d_1,d_2)$ defined and $C \in \mn{sub}(\Tmc)$ with
  $\mn{rd}(C) \leq n-\delta_{12}(d_1,d_2)$, we have $(d_1,d_2)
  \in C^{\Jmc_1 \times \Jmc_2}$ iff $(d_1,d_2) \in C^{\Imc_1 \times
    \Imc_2}$.
  \\[2mm]
  Conversely, assume that there are $n$-initial interpretation trees $\Imc_1$
  and $\Imc_2$ as stated in the lemma. W.l.o.g., we assume that
  $\Delta^{\Imc_1} \cap \Delta^{\Imc_2} = \emptyset$. For $i \in
  \{1,2\}$, let $F_i= \{ d \in \Delta^{\Imc_i} \mid
  \delta_{\Imc_i}(d)=n \}$.  For each $d \in F_i$, choose a model
  $\Imc_d$ of \Tmc and a $\rho_d \in \Delta^{\Imc_d}$ such that
  $\mn{tp}_\Tmc^{\Imc_d}(\rho_d)=t^{\Imc_i}(d)$ (which exists by definition
  of \Tmf and initial interpretation trees). W.l.o.g., assume that
  $\Delta^{\Imc_d} \cap \Delta^{\Imc_e} = \emptyset$ whenever $d \neq
  e$ and $\Delta^{\Imc_i} \cap \Delta^{\Imc_d} = d$ for all $d \in
  F_i$, $i \in \{1,2\}$. Now let $\Jmc_i$ be the interpretation obtained
  by taking the union of $\Imc_i$ and all $\Imc^d$, $d \in F_i$. In detail:
  $$
  \begin{array}{rcl}
    \Delta^{\Jmc_i} &=& \Delta^{\Imc_i} \cup \bigcup_{d \in F_i} \Delta^{\Imc_d} \\[1mm]
    A^{\Jmc_i} &=& A^{\Imc_i} \cup \bigcup_{d \in F_i} A^{\Imc_d} \\[1mm]
    r^{\Jmc_i} &=& r^{\Imc_i} \cup \bigcup_{d \in F_i} r^{\Imc_d} 
  \end{array}
  $$
  The following claim can be proved by a straightforward induction. It
  implies that $\Jmc_1$ and $\Jmc_2$ are models of \Tmc, but $\Jmc_1
  \times \Jmc_2$ is not, whence $\Tmc$ is not closed under
  products. Details are left to the reader. For $(d_1,d_2) \in
  \Delta^{\Imc_1 \times \Imc_2}$, we use $\delta_{1,2}(d_1,d_2)$ to
  denote the length of the path $(d_1,d_2)$ from
  $(\rho^{\Imc_1},\rho^{\Imc_2})$; such a path need not exist (then 
  $\delta_{1,2}(d_1,d_2)$ is undefined), but it is unique if it exists.
  \\[2mm]
  {\bf Claim}. For all $C \in \mn{sub}(\Tmc)$, $d_1 \in
  \Delta^{\Imc_1}$, $d_2 \in \Delta^{\Imc_2}$, $e_1 \in
  \Delta^{\Imc_{d_1}}$, and $e_2 \in \Delta^{\Imc_{d_2}}$,
  we have
  \begin{enumerate}

  \item $C^{\Imc_i} \in t^{\Imc_i}(d_i)$ iff $d_i \in C^{\Jmc_i}$;

  \item $e_i \in C^{\Imc_{d_i}}$ iff $e_i \in C^{\Jmc_i}$;

  \item $(d_1,d_2) \in C^{\Imc_1 \times \Imc_2}$ iff $(d_1,d_2) \in C^{\Jmc_1 \times \Jmc_2}$
    whenever $\delta_{12}(d_1,d_2)$ is defined and $\mn{rd}(C) \leq n-\delta_{12}(d_1,d_2)$.

  \end{enumerate}
\end{proof}
By Lemma~\ref{lem:charproducts}, to decide whether a given \ALC-TBox
\Tmc is not preserved under products, it suffices to guess two initial
interpretation trees $\Imc_1$ and $\Imc_2$ whose size is bounded
exponentially in that of $|\Tmc|$, and then verifying that
$(\rho_1,\rho_2) \notin C_\Tmc^{\Imc_1 \times \Imc_2}$. It is not hard
to see that the latter can be done in time polynomial in the size of
$\Imc_1$ and $\Imc_2$, by explicitly forming the product and then
applying model checking. We have proved the upper bound stated in
Theorem~\ref{prodnew}.

The lower bound stated in Theorem~\ref{prodnew}
is proved by reduction of the $2^{n+1} \times 2^{n+1}$-tiling problem. 
\begin{definition}[Tiling System]\label{domsys}
  A \emph{tiling system} $\Tmf$ is a triple $(T,H,V)$, where $T =
  \{0,\dots,k-1\}$, $k \geq 0$, is a finite set of \emph{tile types}
  and $H,V \subseteq T \times T$ represent the \emph{horizontal and
    vertical matching conditions}.  Let \Tmf be a tiling system and $c
  = c_0,\dots,c_{n-1}$ an \emph{initial condition}, i.e.\ an $n$-tuple
  of tile types.  A mapping $\tau: \{0,\dots,2^{n+1}-1\} \times
  \{0,\dots,2^{n+1}-1\} \to T$ is a \emph{solution} for \Tmf and $c$
  iff for all $x,y < 2^{n+1}$, the following holds (where $\oplus_i$
  denotes addition modulo $i$):
  \begin{itemize}
  \item if $\tau(x,y) = t$ and $\tau(x \oplus_{2^{n+1}} 1,y) =
    t'$, then $(t,t') \in H$
  \item if $\tau(x,y) = t$ and $\tau(x,y \oplus_{2^{n+1}} 1) =
    t'$, then $(t,t') \in V$
  \item $\tau(i,0) = c_i$ for $i < n$.
  \end{itemize}
\end{definition}
To represent grid positions, we use a binary counter that is
implemented through concept names $X_0,\dots,X_{2(n+1)}$ and
$\overline{X}_0,\dots,\overline{X}_{2(n+1)}$ where truth of $X_i$
indicates that bit $i$ is set, truth of $\overline{X}_i$ indicates
that bit~$i$ is not set, the first $n+1$ bits represent the horizontal
value of the grid position, and the remaining $n+1$ bits the vertical
value. To reflect the latter, we will sometimes write $Y_0,\dots,Y_{n}$
instead of $X_{n+1},\dots,X_{2(n+1)}$, and likewise for
$\overline{Y}_0,\dots,\overline{Y}_{n}$.  Define
\newcommand{\ol}{\overline}
$$
\begin{array}{r@{\;}c@{\;}l}
  \mn{tree}_n &=& \midsqcap_{i \leq 2(n+1)} \!\forall r^i . ( \exists r . X_i \sqcap \exists r . \ol{X}_i) \, \sqcap \\[4mm]
  &&\midsqcap_{i < j < 2(n+1)} \!\!\!\!\!\!\! \forall r^j . \big ( \; ( X_i \rightarrow  \forall r . X_i) 
  \sqcap ( \ol X_i \rightarrow  \forall r . \ol X_i) \; \big ) \, \sqcap\\[4mm]
  && \forall r^{2(n+1)} . (\exists r . P  \sqcap \exists r . R \sqcap \exists r . U) \, \sqcap \\[1mm]
  && \forall r^{2(n+1)} . (\forall r . (P \rightarrow X{=}{=})  \sqcap 
\forall r . (P \rightarrow Y{=}{=})) \,  \sqcap  \\[1mm]
  && \forall r^{2(n+1)} . (\forall r . (R \rightarrow X{+}{+})  \sqcap 
\forall r . (R \rightarrow Y{=}{=})) \,  \sqcap  \\[1mm]
  && \forall r^{2(n+1)} . (\forall r . (U \rightarrow X{=}{=})  \sqcap 
\forall r . (U \rightarrow Y{+}{+}))  \, \sqcap \\[1mm]
  && \midsqcap_{j \leq 2(n+1)+1 \atop i \leq 2(n+1)} \forall r^j . \neg ( X_i \sqcap \ol X_i) 
\end{array}
$$
where $\forall r . (P \rightarrow X{=}{=})$ is a concept which
expresses that the horizontal value of all $r$-successors that satisfy
$P$ is identical to the horizontal value at the current node, $\forall
r . (U \rightarrow Y{+}{+})$ expresses that the vertical value of all
$r$-successors that satisfy $P$ can be obtained from the vertical
value at the current node by incrementation, and so on. It is left to
the reader to work out the details of these concepts, we only
give $\forall r . (R \rightarrow Y{=}{=})$ as an example:
$$
\midsqcap_{i \leq n} (Y_i \rightarrow \forall r . (R \rightarrow Y_i))
\sqcap
\midsqcap_{i \leq n} (\ol Y_i \rightarrow \forall r . (R \rightarrow
\ol Y_i)).
$$
Intuitively, the concept $\mn{tree}_n$ generates a tree that contains
all the grid positions, where each subtree rooted at level $2(n+1)$
represents a small fragment of the grid. More specifically, such a
subtree has depth 1 and represents a grid node (the $P$-leaf, where $P$
stands for `current position'), its right neighbor (the $R$-leaf), and
its upper neighbor (the $U$-leaf). To achieve that each such fragment
has a proper tiling, define
$$
\begin{array}{l}
 \mn{tiling}_{\Tmf,c} = \\[3mm]
~~~~~\forall r^{2(n+1)} . \big ( \; \midsqcup_{(t,t') \in H} ( \forall r . (P \rightarrow T_t) \sqcap
  \forall r . (R \rightarrow T_{t'}) \; \big ) \, \sqcap \\[4mm]
~~~~~   \forall r^{2(n+1)} . \big ( \; \midsqcup_{(t,t') \in V} ( \forall r . (P \rightarrow T_t) \sqcap
  \forall r . (U \rightarrow T_{t'}) \; \big ) \, \sqcap\\[4mm]
~~~~~ \forall r^{2(n+1)+1} . \big ( \; \neg \midsqcap_{t,t' \in T} (T_t \sqcap T_{t'}) \; \big ) \, \sqcap \\[5mm]
~~~~~ \forall r^{2(n+1)+1} . \big ( \; \midsqcap_{i < n} ( (P \sqcap
(X{=}{=}i) \sqcap (Y{=}{=}0)) \rightarrow T_{c_i} ) \; \big ) 
\end{array}
$$
where $c_i$ is the $i$-th bit of the initial condition $c$,
$(X{=}{=}i)$ is a concept expressing that the horizontal value at the
current node is identical to the constant $i$, and similarly for
$(Y{=}{=}0)$. 

Note that each position (except those on the fringes of the grid)
occurs at least three times in the tree: as a $P$-node, as an
$R$-node, and as a $U$-node.  To represent a proper solution to the
tiling system, it remains to ensure that multiple occurrences of the
same grid position are labelled with the same tile type. To achieve
this, we
use products. Assume there are \emph{two} tree interpretations of the
above form.  The following concept is true in the root of their
product interpretation iff the two component interpretations 
disagree on the tiling of some position:
$$
\begin{array}{rcl}
\mn{defect}_{n} &=& \exists r^{2(n+1)+1} . (\midsqcap_{i \leq
  2(n+1)} (X_i \sqcup \ol X_i) \sqcap \midsqcap_{t \in T} \neg T_t ) \\[3mm]
\end{array}
$$
To assemble all the pieces into a single concept, set $C_{\Tmf,c}=D_1
\sqcup D_2 \sqcup D_3$ where
$$
\begin{array}{rcl}
D_1 &=&(\mn{tree}_n \sqcap
  \mn{tiling}_{\Tmf,c} \sqcap M) \\[1mm]
D_2&=&  (\mn{tree}_n \sqcap
  \mn{tiling}_{\Tmf,c} \sqcap M')  \\[1mm]
D_3&=&  (\mn{tree}_n \sqcap
  \mn{defect}_{n})
\end{array}
$$
The above encoding of solutions of tiling systems works purely on the
level of concepts, and does not necessarily need TBoxes. We believe
that this is interesting and start with proving a strong form of
correctness: the following lemma shows that given a concept $C$, it is
co-\NExpTime-hard to decide whether $C$ is preserved under
products. In a subsequent step, we will raise this result to the level
of TBoxes.
\begin{lemma}
\label{lem:concprodredcorr}
  There is a solution for \Tmf and $c$ iff $C_{\Tmf,c}$ is not
  preserved under products.
\end{lemma}
\begin{proof}
  First assume that \Tmf and $c$ have a solution $\tau$. We define tree
  interpretations $\Imc_1$ and $\Imc_2$ such that $\Imc_i$ is a model
  of $D_i$ for $i \in \{1,2\}$ (thus both $\Imc_1$ and $\Imc_2$ are
  models of $C_{\Tmf,c}$), but their product is not a model of
  $C_{\Tmf,c}$. For $i \in \{1,2\}$, define 
  $$
    \Delta^{\Imc_i} = \displaystyle \bigcup_{i \leq 2(n+1)} \{ 0,1 \}^i \cup \{0,1\}^{2(n+1)} \cdot \{ P, R, U \},
  $$
  i.e., $\Delta^{\Imc_i}$ is the set of all words over the alphabet
  $\{0,1\}$ of length at most $2(n+1)$ plus all words over $\{0,1\}$
  of length exactly $2(n+1)$ concatenated with a symbol from
  $\{P,R,U\}$. We will not distinguish between words of the former
  kind and numbers represented in binary, lowest bit first. We now
  define a function $\mu: \Delta^{\Imc_i} \rightarrow \Nbbm$ and extend
  $\tau$ to elements of $\Delta^{\Imc_i} \cap (\{0,1\}^{2(n+1)} \cdot \{ P, R, U \})$ as
  follows:
  \begin{itemize}

  \item for each $w \in \Delta^{\Imc_i} \cap \{0,1\}^*$, $\mu(w)=w$;

  \item for each $w \cdot P \in \Delta^{\Imc_i}$ with $w=w_x \cdot
    w_y$, where $w_x,w_y \in \{0,1\}^{n+1}$, set $\mu(w \cdot P)=
    \mu(w)$ and $\tau(w \cdot P)=\tau(w_x,w_y)$;
    
  \item for each $w \cdot R \in \Delta^{\Imc_i}$ with $w=w_x \cdot w_y$,
    where $w_x,w_y \in \{0,1\}^{n+1}$, set $\mu(w \cdot R)= \mu((w_x+1) \cdot w_y)$
    and $\tau(w \cdot R)=\tau(w_x+1,w_y)$;
    
  \item for each $w \cdot U \in \Delta^{\Imc_i}$ with $w=w_x \cdot w_y$,
    where $w_x,w_y \in \{0,1\}^{n+1}$, set $\mu(w \cdot U)= \mu(w_x \cdot (w_y+1))$
    and $\tau(w \cdot U)=\tau(w_x,w_y+1)$.

  \end{itemize}
  For $n \geq 0$, we use $\mn{bit}_j(n)$ to denote the $j$-th bit of
  the binary representation of the number~$n$. To complete the
  definition of $\Imc_1$ and $\Imc_2$ set for $i \in \{1,2\}$, $j
  \leq 2(n+1)$, $t \in T$, and $G \in \{P,R,U\}$:
    $$
    \begin{array}{r@{\,}c@{\,}l}
      r^{\Imc_i} &=& \{ (w,w \cdot c) \mid w \cdot c \in
      \Delta^{\Imc_i}, c \in \{0,1,P,R,U\} \}
      \\[1mm]
      X_j^{\Imc_i} &=& \{ w \mid w \in \Delta^{\Imc_i}, \mn{bit}_{j+1}(\mu(w)) = 1\} \\[1mm]
      \ol X_j &=& \{ w \mid w \in \Delta^{\Imc_i}, \mn{bit}_{j+1}(\mu(w)) = 0\} \\[1mm]
      T_t^{\Imc_i} &=& \{ w \cdot c \mid w \cdot c \in \Delta^{\Imc_i}, c \in \{ P, R, U\}, 
      \tau(w)=t\} \\[1mm]
      G^{\Imc_i} &=& \{ w \cdot G \mid w \cdot G \in \Delta^{\Imc_i} \} \\[1mm]
      M^{\Imc_1}&=&{M'}^{\Imc_2} = \{ \varepsilon \} \\[1mm]
      {M'}^{\Imc_1}&=&{M}^{\Imc_2} = \emptyset \\[1mm]
   \end{array}
    $$
    Now consider the element $(\varepsilon,\varepsilon) \in \Imc_1
    \times \Imc_2$. It is neither an instance of $M$ nor of $M'$, thus
    does neither satisfy $D_1$ nor $D_2$.  Finally, it is not hard to
    show that $(\varepsilon, \varepsilon) \notin \mn{defect}_n^{\Imc_1
      \times \Imc_2}$, essentially because all nodes in $\Imc_1$ and
    $\Imc_2$ that are on the same level and satisfy the same
    combination of $X_i$ and $\ol X_i$ concepts are labelled with the
    same concept name $T_t$. Thus, $(\varepsilon,\varepsilon)$ does
    not satisfy $D_3$, and thus also not $C_{\Tmf,c}$.

\medskip

  First assume that \Tmf and $c$ have a solution. Then take tree
  models \Imc and $\Imc'$ encoding that solution such that \Imc and
  $\Imc'$ are identical except that \Imc satisfies $M \sqcap \neg M'$
  at the root while $\Imc'$ satisfies $M' \sqcap \neg M$ there. In the
  product, $C_{\Tmf,n}$ is false: $D_1$ is false as $M$ is not
  satisfied at the root, $D_2$ is false as $M'$ is not satisfied at
  the root, and $D_3$ is false as there is no defect.

  \medskip

  Conversely, assume that there is no solution for \Tmf and $c$ and
  take two objects $(\Imc_1,d_1)$ and $(\Imc_2,d_2)$ such that $d_1
  \in C_{\Tmf,c}^{\Imc_1}$ and $d_2 \in C_{\Tmf,c}^{\Imc_2}$. We have
  to show that $(d_1,d_2) \in C_{\Tmf,c}^{\Imc_1 \times \Imc_2}$. As a
  preliminary, we state the following claim, which is easily proved
  by induction on the structure of the concept $C$.
  \\[2mm]
  {\bf Claim}. All concepts $C$ built from concept names, conjunction,
  existential restriction, universal restriction, and implication $A
  \rightarrow D$ with $A$ a concept name, are preserved under 
  products.
  \\[2mm]
  Using this claim, it is easy to show that the concept $\mn{tree}_n$
  is preserved under products: it suffices to consider each conjunct
  separately, all conjuncts except the last one are captured by the claim,
  and the last conjunct is clearly also preserved under products.

  \medskip
  \noindent
  We have $d_1 \in \mn{tree}_n^{\Imc_1}$ and $d_2 \in
  \mn{tree}_n^{\Imc_2}$, thus $(d_1,d_2) \in \mn{tree}_n^{\Imc_1
    \times \Imc_2}$. It thus suffices to show $(d_1,d_2) \in
  \mn{defect}_n^{\Imc_1 \times \Imc_2}$: then, $(d_1,d_2) \in
  D_3^{\Imc_1 \times \Imc_2}$, thus $(d_1,d_2) \in C_{\Tmf,c}^{\Imc_1
    \times \Imc_2}$. Distinguish the following cases.

  \medskip
  \noindent
  (i)~$d_1 \in D_3^{\Imc_1}$ or $d_2 \in D_3^{\Imc_2}$.
  
  \smallskip
  \noindent
  We only treat the case $d_1 \in D_3^{\Imc_1}$, as $d_2 \in
  D_3^{\Imc_2}$ is symmetric.  Since $d_1 \in \mn{defect}_n^{\Imc_1}$,
  there is an $e_1 \in \Delta^{\Imc_1}$ and a sequence
  $Z_0,\dots,Z_{2(n+1)}$, $Z_i \in \{X_i, \ol X_i\}$, such that $e_1$
  is reachable from $d_1$ in $2(n+1)+1$ steps along $r$-edges, $e_1
  \in Z_i^{\Imc_1}$ for all $i \leq 2(n+1)$, and $e_1 \notin
  T_t^{\Imc_1}$ for any $t \in T$.  Since $d_2 \in
  \mn{tree}_n^{\Imc_2}$, there is a node $e_2 \in \Delta^{\Imc_2}$
  such that $e_2$ is reachable from $d_2$ in $2(n+1)+1$ steps along
  $r$-edges and $e_2 \in Z_i^{\Imc_2}$ for all $i \leq 2(n+1)$. Then
  $(e_1,e_2)$ is reachable from $(d_1,d_2)$ in $2(n+1)+1$ steps along
  $r$-edges in $\Imc_1 \times \Imc_2$, witnessing that $(d_1,d_2) \in
  \mn{defect}_n^{\Imc_1 \times \Imc_2}$ as desired.

  \medskip
  \noindent
  (ii)~$d_1 \in (D_1 \sqcup D_2)^{\Imc_1}$ and $d_1 \in (D_1 \sqcup D_2)^{\Imc_2}$.

  \smallskip
  \noindent
  Then $d_1 \in \mn{tree}_n^{\Imc_1}$ and $d_1 \in
  \mn{tiling}_{\Tmf,c}^{\Imc_1}$.  Since there is no solution for \Tmf
  and $c$, at least one position of the grid must have non-unique tile
  types, i.e., there is a sequence $Z_0,\dots,Z_{2(n+1)}$, $Z_i \in
  \{X_i, \ol X_i\}$ and distinct $t,t' \in T$ such that 
  $$
  d_1 \in ( \exists r^{2(n+1)+1} . (\midsqcap_{i \leq 2(n+1)} Z_i \sqcap T_t))^{\Imc_1}
  $$
  and
  $$
  d_1 \in ( \exists r^{2(n+1)+1} . (\midsqcap_{i \leq 2(n+1)} Z_i \sqcap T_{t'}))^{\Imc_1}.
  $$
  Take witnesses $d_t$ and $d_{t'}$ for this, i.e., $d_t$ is reachable from
  $d_1$ in $2(n+1)+1$ steps along $r$-edges and satisfies the concept inside
  the upper existential restriction, and analogously for $d_{t'}$.
  Since $d_2 \in \mn{tree}_n^{\Imc_2}$, there is a node $e_2 \in
  \Delta^{\Imc_2}$ such that $e_2$ is reachable from $d_2$ in
  $2(n+1)+1$ steps along $r$-edges and $e_2 \in Z_i^{\Imc_2}$ for all
  $i \leq 2(n+1)$.   Since $d_2 \in \mn{tiling}_{\Tmf,c}^{\Imc_2}$, we do
  not have both $d_2 \in T_t^{\Imc_2}$ and $d_2 \in T_{t'}^{\Imc_2}$. It follows
  that $(d_t,e_2)$ or $(d_{t'},e_2)$ is a witness for $(d_1,d_2) \in
  \mn{defect}_n^{\Imc_1 \times \Imc_2}$ as desired.
\end{proof}
It is now easy to reproduce this on the level of TBoxes.
\begin{lemma}
\label{lem:tboxprodredcorr}
  There is a solution for \Tmf and $c$ iff the TBox $\{ \top
  \sqsubseteq \exists s . C_{\Tmf,c}\}$ is not
  preserved under products.
\end{lemma}
\begin{proof}(sketch)
By Lemma~\ref{lem:concprodredcorr}, it suffices to show that 
$C_{\Tmf_c}$ is preserved under products iff $\exists s . C_{\Tmf,c}$ is.
This is straightforward.
\end{proof}
From Lemma~\ref{lem:tboxprodredcorr}, we obtain the desired lower bound
stated in Theorem~\ref{prodnew}.

\subsection{Proofs for DL-Lite}

\begin{theorem}
It is decidable in \ExpTime whether an $\ALCI$-TBox is invariant under
$\sim_{\text{DL-Lite}_{\sf horn}}$.
\end{theorem}
\begin{proof}
Assume $\Tmc$ is given. Let ${\sf sig}(\Tmc)$ be the set of concept
and role names that occur in $\Tmc$ and denote by $B(\Tmc)$ the closure
under single negation of the set of basic concepts over ${\sf sig}(\Tmc)$.
In this proof, the set ${\sf tp}$ denotes the set of types over
${\sf sub}(\Tmc)\cup B(\Tmc)$; i.e., all subsets $t$ of 
${\sf sub}(\Tmc)\cup B(\Tmc)$ such that
\begin{itemize}
\item $C_{1} \sqcap C_{2}\in t$ iff $C_{1}\in t$ and $C_{2}\in t$, for all $C_{1}\sqcap C_{2} \in {\sf sub}(\Tmc)\cup B(\Tmc)$;
\item $\neg C\in t$ iff $C \not\in t$ iff $\neg C\in {\sf sub}(\Tmc)\cup B(\Tmc)$.
\end{itemize}
For $t \in {\sf tp}$, we set $t^{B}= t \cap B(\Tmc)$.
Let ${\sf tp}^{B}=\{ t^{B} \mid t\in {\sf tp}\}$
and call elements of ${\sf tp}^{B}$ \emph{b-types}.
The following is readily checked:

\medskip

\noindent
Claim 1. If $\Imc_{1},\Imc_{2}$ only interpret symbols in ${\sf sig}(\Tmc)$, then
$\Imc_{1} \sim_{\text{DL-Lite}_{\sf horn}}^{g} \Imc_{2}$ iff the sets of b-types realized
in $\Imc_{1}$ and $\Imc_{2}$ coincide.

\medskip

\noindent
Denote by ${\sf tp}(\Tmc)$ the set of $t \in {\sf tp}$ that are realizable in models of $\Tmc$
and set ${\sf tp}^{B}(\Tmc)= \{ t^{B} \mid t\in {\sf tp}(\Tmc)\}$.
Apply the following rule exhaustively to $Q=\{ t\in {\sf tp}\mid t^{B} \in {\sf tp}^{B}(\Tmc)\}$:
\begin{itemize}
\item If $\exists r.C\in t\in Q$ and there does no exists $s\in Q$ such that $t\leadsto_{r} s$
and $s\in Q$, then remove $t$ of $Q$.
\end{itemize}
Denote the resulting set by $P$. The following is readily checked.

\medskip

\noindent
Claim 2. $P$ consists of the set of all types that are realizable in interpretations
realizing b-types from ${\sf tp}^{B}(\Tmc)$ only.

\medskip

Observe that $P \supseteq {\sf tp}(\Tmc)$.

\medskip

\noindent
Claim 3. $P \not\subseteq {\sf tp}(\Tmc)$ iff $\Tmc$ is not invariant under
$\sim_{\text{DL-Lite}_{\sf horn}}$.

\medskip

\noindent
Assume $P\not\subseteq {\sf tp}(\Tmc)$ and take the disjoint union
$\Imc_{1}$ of interpretations $\Imc_{t}$, $t\in P$, that realize $t$ and realize b-types from
${\sf tp}^{B}(\Tmc)$ only. On the other, take a model $\Imc_{2}$ of $\Tmc$ that realizes
all b-types in ${\sf tp}^{B}(\Tmc)$. Now, $\Imc_{1}$ and $\Imc_{2}$ realize exactly the b-types
in ${\sf tp}^{B}(\Tmc)$, and since we may assume that $\Imc_{1}$ and $\Imc_{2}$ only interpret
symbols from ${\sf sig}(\Tmc)$, we obtain from Claim 1 that $\Imc_{1} \sim_{\text{DL-Lite}_{\sf horn}} \Imc_{2}$.
On the other hand, $\Imc_{2}$ is a model of $\Tmc_{1}$ but $\Imc_{1}$ is not, since it realizes a type
from ${\sf tp}\setminus{\sf tp}(\Tmc)$.

The converse direction is clear.

As ${\sf tp}(\Tmc)$ can be computed in exponential time (since satisfiability
of $\ALCI$-concepts w.r.t.~$\ALCI$-TBoxes is decidable in \ExpTime) and since $P$ is computed in exponential time,
we have proved the \ExpTime-upper bound. 
\end{proof}

\medskip
\noindent
{\bf Theorem~\ref{thm:dllitechar}.}
Let $\varphi$ be a first-order sentence.
Then the following conditions are equivalent:
\begin{enumerate}

\item $\varphi$ is equivalent to a DL-Lite$_\mn{core}$ TBox (resp.\ DL-Lite$^{d}_\mn{core}$ TBox);

\item  $\varphi$ is invariant under $\sim_{\text{DL-Lite}_\mn{horn}}^{g}$ and disjoint unions, and is 
preserved under products and unions (resp.\ compatible
unions). 

\end{enumerate}
%

\noindent
\begin{proof}
  The proof is a variation of the proof of Theorem~\ref{DLLite-TBox}.  
  We again concentrate on $2 \Rightarrow 1$, in particular on showing
  that $\mn{cons}(\vp) \models \vp$, where ${\sf cons}(\varphi)$ is the
  set of all DL-Lite concept inclusions that are a consequence of $\vp$.
  Assume to the contrary that $\mn{cons}(\vp) \not\models \vp$.

  Let $\mn{cons}^{\sqcap,\sqcup}(\vp)$ denote the set of extended
  \Lmc-CIs that are a consequence of $\mn{cons}(\vp)$, where an
  \emph{extended \Lmc-CI} has the form $$ B_{1}\sqcap \cdots \sqcap
  B_{m} \sqsubseteq D_{1}\sqcup \cdots \sqcup D_{n},
  $$
  with both the $B_i$ and the $D_i$ basic DL-Lite concepts.  Our aim
  is to construct interpretations $\Imc^-$ and $\Imc^+$ such that
  $\Imc^- \not\models \vp$, $\Imc^+ \models \vp$, and both $\Imc^-$
  and $\Imc^+$ satisfy precisely those extended \Lmc-CIs that are in
  ${\sf cons}^{\sqcap,\sqcup}(\varphi)$.

  $\Imc^-$ is constructed as in the proof of
  Theorem~\ref{DLLite-TBox}.  For every extended \Lmc-CI $C
  \sqsubseteq D \notin {\sf cons}^{\sqcap,\sqcup}(\varphi)$, take a
  model $\Imc_{C\not\sqsubseteq D}$ of ${\sf cons}(\varphi)$ that
  violates \mbox{$C \sqsubseteq D$}. Then $\Imc^-$ is the disjoint
  union of all $\Imc_{C\not\sqsubseteq D}$ and a model of
  \mbox{$\mn{cons}(\vp) \cup \{ \neg \vp \}$}.  Clearly, $\Imc^-$
  satisfies the desired properties.

  The main step in constructing~$\Imc^+$ is to build a model $\Imc'_{C
    \not\sqsubseteq D}$ of $\vp$ that violates $C \sqsubseteq D$, for
  every extended \Lmc-CI $C \sqsubseteq D \notin {\sf
    cons}^{\sqcap,\sqcup}(\varphi)$. When this is done, $\Imc^+$ will
  simply by the disjoint union of all $\Imc'_{C \not\sqsubseteq D}$.
  Let 
  $$ C \sqsubseteq D = B_{1}\sqcap \cdots \sqcap
  B_{m} \sqsubseteq D_{1}\sqcup \cdots \sqcup D_{n}.
  $$
  Let $i\leq m$ and $j\leq n$. Since $\varphi \not\models C
  \sqsubseteq D$, we also have $\vp \not\models B_i \sqsubseteq D_j$
  and thus there is a model $\Imc_{i,j}$ of $\varphi$ that violates
  $B_i \sqsubseteq D_j$. Assume that $d_{i,j} \in (B_i \setminus
 D_j)^{\Imc_{i,j}}$.  For $1 \leq i \leq m$, take the product
$$
\Imc_{i} = \prod_{1\leq j \leq n}\Imc_{i,j}.
$$
Since $\vp$ is preserved under products, $\Imc_i$ is a model of
$\vp$. Moreover, the element $\overline{d}_i: j \mapsto d_{i,j} \in
\Delta^{\Imc_i}$ satisfies $\overline{d}_i \in B_{i}^{\Imc_i}$ and
$\overline{d}_i \notin (D_{1}\sqcup \cdots \sqcup D_{n})^{\Imc_i}$. By
renaming domain elements, we can achieve that there is a $d \in
\bigcap_{1 \leq i \leq m} \Delta^{\Imc_i}$ such that $d \in
B_i^{\Imc_i} \setminus (D_{1}\sqcup \cdots \sqcup D_{n})^{\Imc_i}$ for
$1 \leq i \leq m$. Now, $\Imc'_{C \not\sqsubseteq D}$ is the union of
the interpretations $(\Imc_i)_{1 \leq i \leq m}$. Then $\Imc'_{C
  \not\sqsubseteq D}$ is a model of $\vp$ since $\vp$ is preserved
under unions and we have $d \in (B_{1}\sqcap \cdots \sqcap
B_{m})^{\Imc'_{C \not\sqsubseteq D}}$ and $d \notin (D_{1}\sqcup
\cdots \sqcup D_{n})^{\Imc'_{C \not\sqsubseteq D}}$, thus $\Imc'_{C
  \not\sqsubseteq D}$ is as required.

We can again assume w.l.o.g.\ that $\Imc^-$ and $\Imc^+$ are
$\omega$-saturated. It remains to show $\Imc^- \equiv^g_\Lmc \Imc^+$,
which can be done as in the proof of Theorem~\ref{DLLite-TBox}.
\end{proof}
\begin{theorem}
  Let $\Lmc_1 \in \mn{ExpDL}$ contain inverse roles and $\Lmc_2 \in
  \{$DL-Lite$_\mn{core}$, DL-Lite$_\mn{core}^d \}$. Then the
  complexity of $\Lmc_1$-to-$\Lmc_2$-TBox rewritability coincides with
  the complexity of TBox satisfiability in $\Lmc_1$.
\end{theorem}
\begin{proof}
  It it common knowledge that for all $\Lmc_1 \in \mn{ExpDL}$, TBox
  satisfiability and Boolean TBox satisfiability have the same
  complexity. It thus suffices to give a reduction from $\Lmc_1$-TBox
  unsatisfiability to $\Lmc_1$-to-$\Lmc_2$-TBox rewritability and from
  $\Lmc_1$-to-$\Lmc_2$-TBox rewritability to the unsatisfiability of
  Boolean $\Lmc_1$-TBoxes. The former is easy since an $\Lmc_1$-TBox
  \Tmc is satisfiable iff $\Tmc \cup \Tmc'$ is not $\Lmc_2$-rewritable,
  where $\Tmc'$ is any fixed TBox that is not $\Lmc_2$-rewritable.

  For the reduction of $\Lmc_1$-to-$\Lmc_2$-TBox rewritability to the
  unsatisfiability of Boolean $\Lmc_1$-TBoxes, fix an $\Lmc_1$-TBox
  \Tmc.  Let $\Sigma$ be the signature of \Tmc, i.e., the set of all
  concept names, role names, and nominals that occur in
  \Tmc. Moreover, let $\Gamma$ be the set of all $\Lmc_2$-concept
  inclusions over $\Sigma$, and
  $$
    \Gamma_\Tmc = \{ \alpha \in \Gamma \mid \Tmc \models \alpha \}.
  $$
  Note that $\Gamma$ is finite, and that its cardinality is polynomial
  in the size of $\Tmc$.
  \\[2mm]
  {\bf Claim 1}. $\Tmc$ is $\Lmc_2$-rewritable iff $\Tmc \equiv \Gamma_\Tmc$.
  \\[2mm]
  The ``if'' direction is trivial. For the ``only if'' direction,
  assume that \Tmc is $\Lmc_2$-rewritable and that $\Tmc'$ is an
  $\Lmc_2$-TBox that is equivalent to \Tmc. Clearly, every concept
  inclusion in $\Gamma_\Tmc$ is a consequence of $\Tmc'$. Conversely,
  every concept inclusion in $\Tmc'$ must also be in
  $\Gamma_\Tmc$. Thus, $\Gamma_\Tmc \equiv \Tmc_1 \equiv \Tmc$.

  \medskip

  By Claim~1, it suffices to reformulate the question `is \Tmc
  equivalent to $\Gamma_\Tmc$?' in terms of unsatisfiability of
  Boolean $\Lmc_1$-TBoxes. This is what we do in the following.
  First, we may assume w.l.o.g.\ that \Tmc is of the form $\{ \top
  \sqsubseteq C_\Tmc \}$ with $C_\Tmc$ and $\Lmc_1$-concept in
  \emph{negation normal form (NNF)}, i.e., negation is only applied to
  concept names and nominals. For each concept name $A \in \Sigma$
  (role name $r \in \Sigma$, nominal $a \in \Sigma$) and $\alpha \in
  \Gamma$, reserve a fresh concept name $A_\alpha$ (role name
  $r_\alpha$, nominal $a_\alpha$). Moreover, for each $\alpha \in
  \Gamma \uplus \{ \bullet \}$ reserve an additional concept name
  $R_\alpha$. The new symbols give rise to signature-disjoint and
  relativized copies $\Tmc_\alpha$ of \Tmc, for each $\alpha \in
  \Gamma \uplus \{ \bullet \}$, defined as follows:
  \begin{enumerate}

  \item replace in $C_\Tmc$ each concept name $A$ with $A_\alpha$,
    each role name $r$ with $r_\alpha$, and each nominal $a$ with
    $a_\alpha$; call the result $C_{\Tmc,\alpha}$;

    \item replace $\top \sqsubseteq C_{\Tmc,\alpha}$ with $R_\alpha \sqsubseteq C_{\Tmc,\alpha}$;

    \item replace each subconcept $\exists r . C$ in $C_{\Tmc,\alpha}$
      with $\exists r .  (R_\alpha \sqcap C)$ and each subconcept
      $\forall r . C$ in $C_{\Tmc,\alpha}$ with $\forall r .
      (R_\alpha \rightarrow C)$.

  \end{enumerate}
  Note that $R_\alpha$ is used for relativization, i.e., the TBox
  $\Tmc_\alpha$ does not `speak' about the entire domain, but only
  about the part identified by $R_{\alpha}$.

  Analogously, we introduce a renaming and relativization
  $\alpha_\alpha$ for each $\alpha \in \Gamma \uplus \{ \bullet \}$:
  first rename symbols as in Step~1 above, then replace $\alpha = B_1
  \sqsubseteq B_2$ with $R_\alpha \sqcap B_1 \sqsubseteq B_2$. Note
  that the modified $\alpha$ is not in $\Lmc_1$, but in $\Lmc_2$
  (since the latter contains inverse roles).

  Define a Boolean TBox
  $$
  \vp=  \bigwedge_{\alpha \in \Gamma} (\Tmc_\alpha \wedge (\alpha_\alpha
  \rightarrow \alpha_{\bullet})) \wedge \neg \Tmc_{\bullet}
%
  $$
  It suffices to prove the following
  \\[2mm]
  {\bf Claim 2}. \Tmc is $\Lmc_2$-rewritable iff $\vp$ is unsatisfiable.
  \\[2mm]
  ``if''. Let \Tmc not be $\Lmc_2$-rewritable. By Claim~1, we then
  have $\Gamma_\Tmc \not\models \Tmc$. For each $\alpha \in \Gamma$,
  take a model $\Imc_\alpha$ of $\Tmc_\alpha$ such that $\Imc_\alpha
  \models \alpha_\alpha$ iff $\alpha \in \Gamma_\Tmc$. Additionally,
  take a model $\Imc_\bullet$ of $\{ \alpha_{\bullet} \mid \alpha \in
  \Gamma_\Tmc\}$ with $\Imc_{\bullet} \not \models \Tmc_{\bullet}$. We can
  assume w.l.o.g.\ that $\Imc_\bullet$ and each $\Imc_{\alpha}$ have
  the same domain $\Delta$: if they don't, then the relativization to
  $R_\alpha$ allows us to extend each domain $\Delta^{\Imc_\alpha}$ to
  $\bigcup_{\alpha \in \Gamma \uplus \{ \bullet \}}
  \Delta^{\Imc_\alpha}$. Define a new interpretation \Imc as follows:
  \begin{itemize}

  \item $\Delta^\Imc = \Delta$;

  \item $A_\alpha^\Imc = A_\alpha^{\Imc_\alpha}$, $r_\alpha^\Imc =
    r_\alpha^{\Imc_\alpha}$, and $a_\alpha^\Imc =
    a_\alpha^{\Imc_\alpha}$ for all concept names $A$ including the
    relativization names $A_\alpha$, role names $r$, nominals $a$, and
    $\alpha \in \Gamma \uplus \{ \bullet \}$;

  \end{itemize}
  It is not hard to prove that \Imc is a model of $\Tmc_{\alpha}$ for
  all $\alpha \in \Gamma$, but not for $\Tmc_{\bullet}$. Moreover, it
  is a model of $\alpha_\alpha$ iff $\alpha \in
  \Gamma_\Tmc$ and a model of $\alpha_\bullet$ if $\alpha \in \Gamma_\Tmc$,
  for all $\alpha \in \Gamma$. Therefore, $\Imc$ satisfies~$\vp$.

  \smallskip

  ``only if''. Let \Tmc be $\Lmc_2$-rewritable. Then $\Tmc \equiv
  \Gamma_\Tmc$ by Claim~1. Assume to the contrary of what is to be
  shown that there is a model \Imc of $\vp$.  Since $\Tmc \models
  \Gamma_\Tmc$, it is easy to see that $\Tmc_\alpha \models
  \alpha_\alpha$ for all $\alpha \in \Gamma_\Tmc$. For this reason,
  the first conjunct of $\vp$ yields $\Imc \models \alpha_{\bullet}$
  for all $\alpha \in \Gamma_\Tmc$. Since $\Tmc \equiv \Gamma_\Tmc$,
  this yields $\Imc \models \Tmc_{\bullet}$, in contradiction to \Imc
  satisfying the last conjunct of~$\vp$.
\end{proof}
Note that, when inverse roles are not contained in $\Lmc_1$, the above
proof yields a reduction of $\Lmc_1$-to-$\Lmc_2$-TBox rewritability to
satisfiability of Boolean $\Lmc_1\Imc$-TBoxes, where $\Lmc_1\Imc$ is the
extension of $\Lmc_1$ with inverse roles.